\documentclass[11pt]{article}
\usepackage[utf8]{inputenc}
\usepackage{fullpage}
\usepackage{url}
\usepackage{parskip}
\usepackage{setspace}
\usepackage{amsmath}
\usepackage{amssymb}
\usepackage{amscd}
\usepackage{accents}
\usepackage{bbm}
\usepackage{listings}
\usepackage{ifpdf}
\usepackage{hyperref}
\usepackage{epic}
\usepackage{enumerate}
\usepackage{setspace}
\usepackage{lscape}
\usepackage{caption}
\usepackage{subcaption}
\usepackage{accents}
\usepackage{float}
\usepackage{dsfont}
\usepackage{verbatim}
\usepackage{todonotes}
\usepackage{multirow,array}
\usepackage{changepage}
\usepackage{lipsum}
\usepackage{advdate}
    \usepackage{sectsty}
    
    \usepackage{tikz}
    \usetikzlibrary{arrows.meta}
\tikzset{>={Latex[width=3mm,length=3mm]}}

\usepackage[round]{natbib}   
    \newtheorem{theorem}{Theorem}
\newtheorem{lemma}[theorem]{Lemma}
\newtheorem{claim}[theorem]{Claim}
\newtheorem{proposition}[theorem]{Proposition}
\newtheorem{corollary}[theorem]{Corollary}

\newtheorem{assumption}{Assumption}
\newtheorem{condition}{Condition}

\newtheorem{definition}{Definition}[section]

\newenvironment{proof}[1][Proof]{\begin{trivlist}
\item[\hskip \labelsep {\bfseries #1}]}{\end{trivlist}}

\newcommand{\qed}{\nobreak \ifvmode \relax \else
      \ifdim\lastskip<1.5em \hskip-\lastskip
      \hskip1.5em plus0em minus0.5em \fi \nobreak
      \vrule height0.75em width0.5em depth0.25em\fi}

\def\0{{\mathbb 0}}

\setcitestyle{authoryear}
\title{Inference from Selectively Disclosed Data}
\date{\AdvanceDate[-1]\today}
\author{Ying Gao\thanks{Department of Economics, MIT. Email: gaoy@mit.edu. I thank Drew Fudenberg and Stephen Morris for their generous advice and guidance, and Ian Ball, In-Koo Cho, Laura Doval, Wioletta Dzuida, Robert Gibbons, Ellen Muir, Alexander Wolitzky, and audiences at the Stony Brook International Conference in Game Theory, Econometric Society Summer School in Game Theory, European Summer Meetings of the Econometric Society, MIT Organizational Economics Seminar, and MIT Theory Lunch for helpful comments.}}

\begin{document}
\onehalfspacing
\begin{titlepage}
    \maketitle
    
    \begin{abstract}
    We consider the disclosure problem of a sender with a large data set of hard evidence who wants to persuade a receiver to take higher actions. Because the receiver will make inferences based on the distribution of the data they see, the sender has an incentive to drop observations to mimic the distributions that would be observed under better states. We predict which observations the sender discloses using a model that approximates large datasets with a continuum of data. It is optimal for the sender to play an \emph{imitation strategy}, under which they submit evidence that imitates the natural distribution under some more desirable target state. We characterize the partial-pooling outcomes under these imitation strategies, and show that they are supported by data on the outcomes that maximally distinguish higher states. Relative to full information, the equilibrium with voluntary disclosure reduces the welfare of senders with little data or a favorable state, who fully disclose their data but suffer the receiver's skepticism, and benefits senders with access to large datasets, who can profitably drop observations under low states.
    \end{abstract}
    \setcounter{page}{0}
    \thispagestyle{empty}
\end{titlepage}

\section{Introduction}
Many decisions -- including technology adoption, regulatory approval, and research grantmaking -- are  based on self-disclosed data. The datasets used can often be very large, on the order of tens of thousands of trials for drug approval, and often hundreds of thousands of datapoints about locations and sales in merger cases. In and of themselves, big datasets may paint an accurate picture of reality, but the sender can disclose them strategically: it is easier to verify that the data submitted are real than that they are complete, and even in the presence of mandatory disclosure rules, deciding which observations are admissible to include in the dataset is largely at the sender's discretion. 

We want to understand the role data play in strategic communication between the sender and the decision-maker when receivers have uncertainty about the underlying dataset from which the sender extracted the submitted data. We consider the case of a sender with state-independent motives to persuade the receiver towards a particular action, and a receiver who observes a dataset the sender discloses, but interprets it with partial skepticism that the data are incomplete. Equilibrium play between the sender and receiver involves the sender submitting data as ``proof" that the receiver should take a favorable action, and the receiver evaluating how persuasive the proof is depending on how likely it is sent by a sender with less persuasive data who has trimmed some discouraging observations. This can be modeled under the framework of an evidence game in which senders that have access to datasets with weakly more observations of each outcome can always mimic senders with fewer observations. A special case, in which senders either have or do not have access to a single data point,  with probability known to the receiver, is already well-understood (\citealt{Dye85}), and demonstrates that senders can manipulate the receiver by disclosing nothing when the evidence is sufficiently poor.

Our primary innovation is to characterize disclosure in the opposite extreme, when datasets contain many observations. We propose a continuous-data model of the  asymptotic distribution over potential datasets of the sender that depends on two things: the true state of the world that generates the data, and a random variable that describes the amount of data the sender collects. The continuum assumption captures the fact that empirical distributions are approximately deterministic in the limit with large numbers, and allows us to eliminate uncertainty over the randomness of draws, which makes the model more tractable than  directly modeling large, finite $N$. Instead, we show that the outcome we characterize in the continuous model describes the limit outcome of communication in finite-data games as $N \rightarrow \infty$.

In addition to an extensive list of observations, a second characteristic feature of ``big" data is a large outcome space. This motivates the novel use of a framework that encompasses general statistical settings, including those in which outcome and state spaces are large and the relationship between them complex. In particular, we place essentially no restrictions on the state-contingent data distribution. In general, unlike in a ``good news-bad news" model of data, the ranking of states and the shape of their data-generating distributions endogenously affects the interpretation of different outcomes, and context determines whether data take on a positive or negative connotation. 

Indeed, our first main result, Prop. \ref{prop:sufficiency}, is a sufficiency result that says that in the receiver-optimal partial pooling equilibrium outcome, the state-contingent experimental outcome distribution affects the information transmitted only through a handful of key features: what matters are the observations of outcomes with the greatest likelihood ratio under a better vs. a worse state. Strikingly, since the distribution of data that distinguishes one state from another depends solely on the relative probabilities of likelihood ratio-maximizing outcomes, a receiver who wants to distinguish a relatively small number of states with many observations of high-dimensional data can do just as well restricting the dataset to only retain information about these outcomes. When state-contingent distributions of experimental outcomes satisfy the monotone likelihood ratio property (MLRP), we return to the case of only one ``good news" outcome, and distinguishing it from other outcomes is sufficient to support receiver-optimal communication.

Our second result, Theorem \ref{thm:imitate}, characterizes an ``imitation" equilibrium implementation of the receiver-optimal equilibrium outcome, in which senders always show the receiver a dataset that can correspond to a naturally-generated dataset, so that on path, the receiver always places positive probability on the event that the sender is sending all their data. However, the receiver also infers from some datasets that the sender has with positive probability observed data corresponding to a different state than the revealed data suggest, but has dropped observations in order to imitate a more favorable distribution. When MLRP fails, it is important for the imitating sender to send a large-enough mass of realizations of a certain outcome, but not too much. The resulting outcome benefits senders under low states with more data at the expense of senders with less data in high states, since the former pool with the latter. The extent of pooling depends on the receiver's uncertainty about the sender's data collection capabilities: the greater the variance in the receiver's belief about how much data the sender starts out with, the more senders can profitably imitate other senders, with outcomes converging to the full-information one as uncertainty vanishes. 

The other contribution is an algorithm to construct the limit game equilibrium outcomes that follows a top-down logic: senders with more data receive weakly greater payoffs, and we can construct the payoff frontiers of the continuous payoff function by specifying the burden of proof, or how much data of a given state's distribution a sender needs, to induce a particular belief in the receiver. The algorithm is applicable to any number of states and for distributions with finite, discrete support, and we illustrate it with representative 2 and 3 state examples.

\subsection{Related literature}
Strategic disclosure has been studied since the work of \citet{Grossman81} and \citet{Milgrom81}, which showed that full disclosure is the unique outcome when receivers know that the sender wishes to prove the value of a good is high using verifiable information that they could choose to disclose. The assumption that receivers know the sender is informed is crucial to this benchmark, as \citet{Dye85} and \citet{JungKwon88} show. They consider a case in which the sender has access to a single, real-valued piece of evidence with interior probability $p \in (0,1)$, and shows that only senders for whom the evidence exceeds a threshold will choose to disclose it, with the rest withholding it in order to pool with those senders who lack evidence altogether. \citeauthor{Shin94} (1994, 2003) shows that in the case where senders have an uncertain endowment of good news and bad news, the fact that senders withhold bad enough evidence implies a ``sanitation equilibrium", in which all bad news is disposed of.

We extend these results by considering evidence structures with large, multidimensonal datasets. In our data-based setting, evidence is neither exogenously good nor bad, but the receiver draws inferences statistically, based on knowledge of the relationship between relevant state-parameters and the distributions of data they generate. The setting we consider  encompasses the settings above and captures a special case of more abstract evidence games of the type considered by \citet{Green86} and \citet{Hart17}. The main focus in those settings has been on receiver-optimal mechanisms to induce beneficial disclosures from the sender. \citet{Hart17} in particular is foundational to our equilibrium selection criterion. Their observation that the optimal mechanism, the receiver-optimal equilibrium, and the unique truth-leaning equilibrium all yield the same outcome generalizes straightforwardly to our setting.\footnote{The optimal mechanism equivalence result has been noted by others, in particular \citet{Glaezer06}, \citet{Sher11}, and \citet{BenPorath19}, who show that the fact that commitment is not necessary for the optimum is robust to other settings, in particular with binary actions and multiple senders with type-dependent preferences.}

\citet{Rappoport22} and \citet{Jiang22} use an iterative algorithm to solve for truth-leaning  equilibrium outcomes in finite evidence games, but it is computationally demanding to use it in games with large type spaces, and therefore infeasible to directly compute the large-$N$ limit of outcomes in finite-data games. Our approach is instead to use a continuous-data approximation to solve for asymptotic outcomes without explicitly computing outcomes of finite-data games, and to show that it is exactly the big-data limit outcome.
Only one other paper that we know of, \citet{Dzuida11}, uses a continuous measure of evidence to solve for communication with verifiable evidence. The model considers only evidence with a continuum of states but a simple ``good news-bad news" outcome structure, and assumes there is a positive probability of a behavioral, honest type of the sender. The existence of the honest type, along with the assumption of continuity in outcomes, also selects the most plausible equilibrium in a similar fashion to truth-leaning. The honest type also drives an observation that providing interior amounts of negative evidence can be optimal in otherwise sanitation-like equilibria, because honest types will send some negative evidence; however, when the probability of honest types vanishes, so does this behavior. We observe a similar finding but for two different reasons. First, when the state space is finite, giving fewer observations of ``bad" outcomes may not change the receiver's belief conditional on other outcomes, so retaining them can be lossless to the sender. Secondly, with more than two states or two outcomes, there is also no single ``good" outcome, and sending interior amounts of some outcomes may rule in just the right set of rational types, while sending either none or as much as possible of an outcome might both be strictly worse.

We also relate to a broader literature about the optimal collection and disclosure of evidence, that considers costly (\citet{Migrow22}) and dynamic evidence acquisition (\citealt{Felgenhauer14},  \citealt{Henry19}), sender-optimal disclosure mechanisms (\citealt{Haghtalab22}), and discretionary disclosure after test or information design (\citealt{Shishkin22}, \citealt{Dasgupta22}). Several papers use a restricted notion of evidence but are also explicitly concerned with the effect of allowing sample selection: \citet{Fishman90} and \citet{DiTillio21} study the case in which only a subset of observations are disclosed, and give conditions under which it is better that an informant have discretion over which data are selected. Finally, work in econometrics by \citet{Simonsohn14}, \citet{Andrews19} and others studies the bias that arises from exogenously selective reporting, and describes inference procedures that correct for it.

\section{Model}
\noindent
\textbf{States and payoffs.}
There is a sender ($S$), who wishes to communicate to a receiver ($R$) about an unknown state of the world, $\theta  \in \Theta = \{\theta_1, \ldots, \theta_J\}$. The sender and receiver share a common prior $\beta_0(\cdot)$ over $\Theta$. We assume that the receiver takes an action $a_r \in \mathbb{R}$ and that $\theta_1, \ldots, \theta_J$ are real numbers ordered with $\theta_j \le \theta_{j+1}$, representing the optimal action for the receiver under each state, if it was known with certainty. The sender's payoff is simply (a monotone function of) $a_r$;\footnote{Because the receiver will always play a pure strategy, the sender's problem is unchanged if their payoffs are rescaled through a monotone mapping.} in short, regardless of their type, they want to induce the receiver to take the highest action possible.

Finally, we assume the receiver has an expected utility that is differentiable and single-peaked at the action that matches their expectation of the value of $\theta$, that is, that for any belief $\beta \in \Delta \Theta$, the receiver's expected payoff $\mathbb{E}_\beta[u_r(a)]$ is single-peaked at $a_r(\beta) = \mathbb{E}_\beta[\theta]$.\footnote{The assumption of single-peakedness is necessary to identify the receiver-optimal equilibrium and the receiver-optimal mechanism.} We work with the sender's indirect utility as a function of the receiver's beliefs, which induces them to maximize the receiver's posterior expectation of $\theta$:
\begin{equation}
u_s(\beta) = \mathbb{E}_\beta[\theta].
\end{equation}
For example, when the receiver is a policymaker, states can represent the true optimal policy. While the policymaker might be uncertain, they wish to enact a policy that matches the optimal policy in expectation, while the sender wishes for them to take as high an action as possible.

\noindent
\textbf{Evidence.}
The private information of the sender comes in the form of hard evidence about the state of the world. In particular, the sender has access to a dataset of observations drawn from a finite set of outcomes, $\mathcal{D} = \{1,\ldots, D\}$. The underlying data-generating distribution is state-contingent: under state $\theta_j$, the observations are i.i.d. draws from distribution $f_j$.

We model the amount of data the sender has access to as a mass, $\mu \in [0,1]$, that represents the fraction of total potential data that the sender can access, and has a continuous distribution, $g$, that is state-independent\footnote{For simplicity of exposition, we focus on the case in which their belief about $\mu$ conditional on $\theta$ is given by a probability density $g$ that is independent of $\theta$, although most results hold identically for cases in which the distribution of $\mu$ is state-specific.}, supported on $[0, 1]$ with $g(1) = 0$, and infinitely left-differentiable\footnote{The assumption that $g$ has a vanishing right tail ensures that it is continuous on $\mathbb{R}^+$ while being supported on $[0,1]$, and simplifies the equilibrium construction: specifically, it ensures that the equilibrium payoffs are continuous in $\mu$.}.
The continuum assumption models big datasets in which the large number of draws essentially removes all uncertainty about the impact of randomly realized outcomes on the sender's dataset: conditional on state $\theta_j$, the empirical distribution of data the sender observes is certain to be $f_j$, and $\mu$ does not affect the distribution of their evidence, only the amount of it. In other words, with probability $1$, a sender with a mass $\mu$ of data under state $\theta_j$ observes the dataset $t = \mu f_j$. Any nonzero measure of data fully informs the sender of the state, and the set of possible complete datasets and types of the sender is $\mathcal{T} = [0,1] \times \Theta$.

The receiver, on the other hand, is uninformed about how much data the sender has. 
Their prior belief about the sender's type is given by the density
\begin{equation}
q(\mu f_j) = \beta_0(\theta_j) g(\mu).
\end{equation}

\noindent
\textbf{Messaging and inference.} Senders can choose a subset of observations from their dataset to submit to the receiver. We assume total flexibility in the choice of subset:

\begin{assumption}
The sender can send any message $m \in \mathcal{M} = [0,1] \times \Delta \mathcal{D}$ that is a \emph{subset} of their dataset ($m \tilde \subseteq \mu f_j$), where
\[
m \tilde \subseteq \mu f_j \ \Leftrightarrow \ m(d) \le  \mu f_j(d) \ \  \forall d \in \mathcal{D}.
\]
\end{assumption}

\begin{figure}
    \centering
    \begin{subfigure}{.5\textwidth}
    \includegraphics[scale = 0.4]{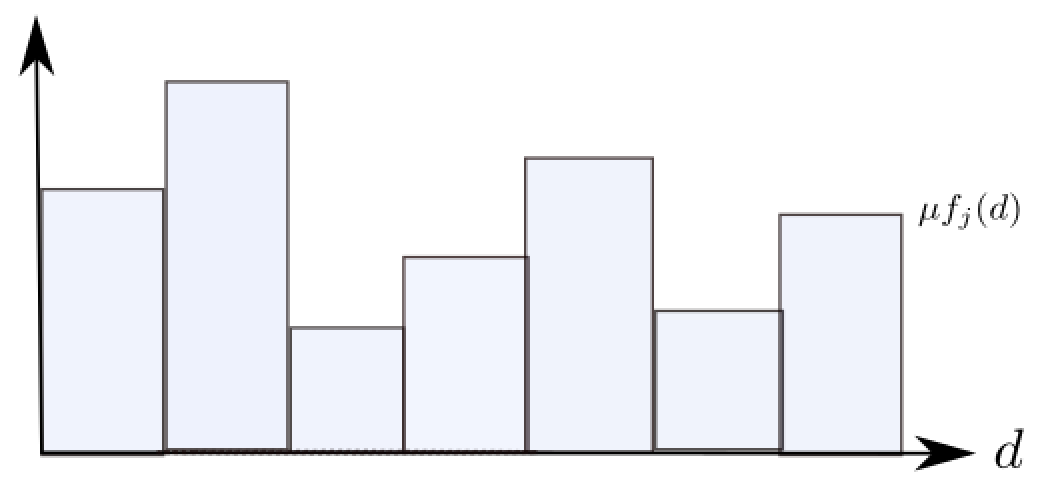}
    \end{subfigure}%
    \begin{subfigure}{.5\textwidth}
    \includegraphics[scale = 0.4]{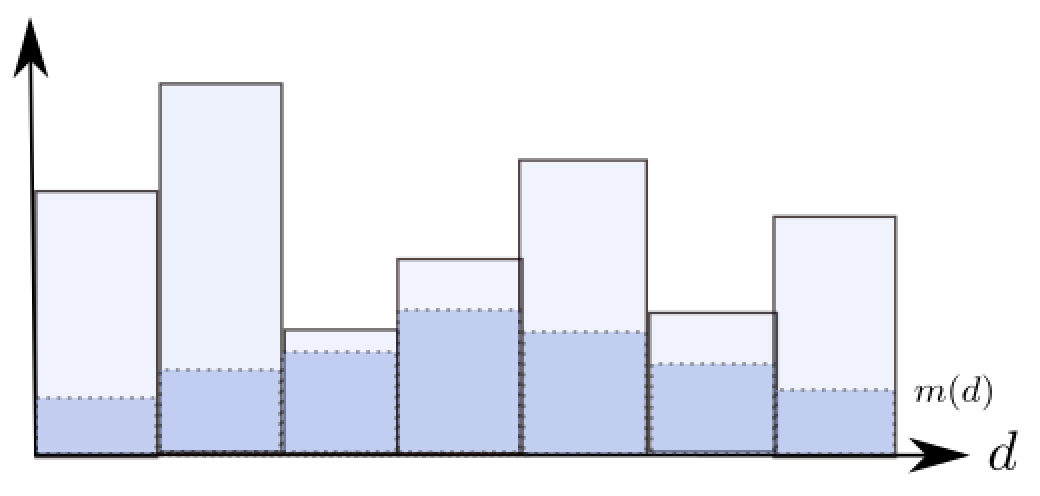}
    \end{subfigure}
    \label{fig:feasible_type_msg}
    \caption{A feasible type and a feasible message.}
\end{figure}

That is, a sender can drop an arbitrary mass of observations from their data, and then show the remaining ones to the receiver. By dropping observations, they can arbitrarily alter the relative frequencies of each outcome in the submitted dataset in order to imitate any distribution. However, this is costly in that it reduces the size of the submitted dataset, which is observable. 

We have that $\mathcal{M} \supset \mathcal{T}$: the message space contains the set of all possible complete datasets, but also a $D$-dimensional set of other datasets that could be disclosed to the receiver after excluding part of their dataset. For any set of messages $M$, define the upper set $U(M)$ to be the set of types that can send a message in $M$, and for any set of types $T$, define the lower set $L(T)$ as the set of messages that some $t \in T$ can send.

Call the disclosure game with these parameters $\mathcal{G}(\Theta, \mathcal{D}, \beta_0, \{f_j\}_{j=1}^J, G)$. Upon observing the sender's message, the receiver updates their belief about the sender's type to $q(t|m)$, and then forms a new belief about the state,
\begin{equation} \label{eq:state_beliefs}
    \beta(\theta_j|m) = \frac{\sum_{j = 1}^J \int_{\mu = 0}^1 q(\mu f_j|m)\theta_j}{\sum_{j = 1}^J \int_{\mu = 0}^1 q(\mu f_j|m)}.
\end{equation}

\subsection{Equilibrium}
The sender plays a messaging strategy $\sigma^*: \mathcal{T} \rightarrow \Delta \mathcal{M}$, knowing which the receiver infers the content of message they receive.
As usual, the equilibrium we consider will be a Perfect Bayesian Equilibrium (\citealt{FudenbergTirole}), that is, $\beta^*(\cdot|m)$ must be consistent with the sender's strategy $\sigma^*$, and the sender must optimize, so
$\sigma^*(m|t) > 0$ only if $m \in \arg \max_{m' \in L(t)} \mathbb{E}_{\beta(\cdot|m')}[\theta]$.

Call the map from types to payoffs, $u_{\sigma^*}(t)$, the \emph{outcome} of the equilibrium.\footnote{This is a departure from the usual definition of an outcome of an extensive-form game, but consistent with the definition in \citet{Hart17} and \citet{Rappoport22}. It describes the action the receiver plays after communicating with each type, and so describes the consequences of communication in the game.} In the perfectly separating outcome, the sender obtains a payoff of $\theta$. As in \citet{Milgrom81}, \citet{Grossman81}, and \citet{Dye85}, when $g$ is a degenerate distribution such that $\mu$ is known to the receiver, then all attempts to mislead the receiver unravel, and the fully separating outcome obtains in every PBE. When $g$ and all $f_j$ have full support, there is partial pooling in every PBE. However, PBE are often not unique, and in this case, there may be multiple $\beta^*$, differing on off-path messages, that are consistent with $\sigma^*$, and the game generically has multiple, non-payoff-equivalent PBE outcomes. Any message that can be played by some type of sender under a state $\theta_j \ge \mathbb{E}_{\beta_0}[\theta]$ is played on-path in some PBE.

Intuition suggests that the game is fundamentally one of imitation: senders tailor their data to increase the receiver's belief that the state is a higher one, and they can only do so by imitating the datasets submitted by higher-state types, who themselves may be imitating others or trying to distinguish themselves as well as possible from lower-state types. One way to imitate a higher-state type of sender is to try to prove you have all the data that they would, and no more -- that is, to imitate their complete dataset. We define an \emph{imitation equilibrium} to capture the idea that sender masquerades as other type by imitating their full datasets.

\begin{definition} \label{def:imitation_eq}
    $(\sigma^*, \beta^*)$ is an imitation equilibrium if it is an equilibrium, and under $\sigma^*$,
\begin{enumerate}
    \item[a.] Every on-path message is in $\mathcal{T}$,
    \item[b.] Type $\mu f_j$ plays $m \neq \mu f_j$ if and only if $\theta_j < \max_{m' \in L(t)} \mathbb{E}_{\beta^*(\cdot|m)}[\theta]$, and otherwise reports their full dataset.
\end{enumerate}
\end{definition}
In other words, with an imitation messaging strategy every type of the sender either fully reveals their data or imitates another type's full dataset, and they only consider the latter if it could give them a better payoff than letting the receiver be fully informed of the state. 

Why do we focus on these equilibria? Imitation equilibria are \emph{truth-leaning}, as first defined by \citet{Hart17} in the context of general evidence games with finite types. The idea applies identically in this setting. Formally, given a base game $\mathcal{G}$, for $\epsilon = (\epsilon_t, \epsilon_{t|t})_{t \in \mathcal{T}}$, let a game $\mathcal{G}_{\epsilon}$ be the game with an identical type set and type distribution, but with two differences. First, type $t$'s payoffs to playing $t$ are perturbed by $\epsilon_t$, so that $t$'s payoff to playing $t$ is $\mathbb{E}_{\beta(\cdot|t)}[\theta] + \epsilon_t$. Secondly, type $t$ plays $t$ with at least probability $\epsilon_{t|t}$ -- i.e. with probability $\epsilon_{t|t}$ a sender with dataset $t$ is a commitment type that plays their full dataset regardless of whether doing so is optimal, while with probability $1-\epsilon_{t|t}$ type $t$ is strategic. A truth-leaning equilibrium is an equilibrium of the base game that can be obtained as a limit of equilibria of $\epsilon$-perturbed games as $\epsilon \rightarrow 0$.

While truth-leaning equilibrium strategies capture a sender’s slight bias towards truth-telling, the truth-leaning equilibrium outcome has desirable properties in its own right. When the receiver’s expected payoffs are single-peaked in their action, the truth-leaning equilibrium outcome is also receiver-optimal, and is the outcome of the optimal mechanism when the receiver can commit to a single action as a response to each message. This is well-known in the finite case studied by \citet{Hart17}, and continues to be true in the continuous model that we study. It is also the only equilibrium outcome robust to a slightly stronger version of a credible announcement (\citealt{Matthews91}). We say that under equilibrium $\sigma^*$ a coalition $T$ of senders can benefit from an \emph{inclusive} credible announcement if there is a set of messages such that the coalition is comprised of every sender that 1) finds some message in the set feasible and 2) weakly benefits from the receiver updating that their type is in $T$ from the prior, relative to the receiver's equilibrium inference; and there is at least one sender in the coalition that strictly benefits.\footnote{The departure from the usual credible announcement is that the coalition must also contain all senders who are indifferent between participating in the announcement and their equilibrium payoff.} Robustness to such announcements means that the equilibrium survives even if senders are able to override the receiver's beliefs by proposing sensible reinterpretations of messages, and can coordinate to do so; it rules out, for example, play that is ``stuck" in a bad equilibrium due to immalleable off-path beliefs. For a deeper discussion of these refinements, see \citet{Hart17} and Appendix \ref{app:refinement}. 

\begin{claim}
    Imitation equilibrium messaging strategies are the truth-leaning equilibrium messaging strategies of $\mathcal{G}$. Imitation equilibrium outcomes are:
    \begin{itemize}
        \item Receiver-optimal among equilibria and pure-strategy mechanisms;
        \item The unique inclusive announcement-proof equilibrium outcomes.
    \end{itemize}
\end{claim}

\subsection{Examples}

\subsubsection{A 2-state prediction problem} \label{ex:2state}

\noindent
A sender wishes to provide evidence to prove the quality of a prediction algorithm that aims to classify whether a future event is likely or unlikely. The quality of the algorithm is either high or low ($\Theta = \{\theta_L, \theta_H$\}), with $\theta_L = 0$ and $\theta_H = 1$.
Suppose that there are $4$ possible outcomes, $\mathcal{D} = \{1,2,3,4\}$ with the following distribution of outcomes per state:\footnote{For example, the problem could be predicting whether it will rain, and outcomes $1, 2, 3, 4$ could be (predict no rain, no rain), (predict rain, rain), (predict rain, no rain) and (predict no rain, rain), respectively. We can think of the high-quality algorithm as being able to more accurately make the right call when it will not rain, while the low-quality algorithm often predicts rain even when it will not rain.}

\begin{table}[H]
    \centering
    \begin{tabular}{c|cccc}
         j & $f_j(1)$ & $f_j(2)$ & $f_j(3)$ & $f_j(4)$ \\
         \hline
        H &  0.6 & 0.2 & 0.1 & 0.1 \\
        L &  0.4 & 0.2 & 0.3 & 0.1
    \end{tabular}
    \caption{The generating distribution of outcomes under states $\theta_H$ and $\theta_L$.}
    \label{tab:ex_2_state}
\end{table}

Imitation implies that every on-path message either contains data distributed like $f_H$ or like $f_L$, which the receiver can interpret as a claim that ``the state is $H$" or ``the state is $L$", respectively. But since the sender strictly prefers the receiver to believe the state is $H$ with higher probability, there is no reason to imitate $f_L$. Indeed, part \ref{def:imitation_eq}(b) of the definition of an imitation equilibrium ensures that the only on-path messages take the form $\mu f_H$, since no posterior belief of the receiver is worse for the sender than full certainty that $\theta = \theta_L$.

Additionally, the sender chooses an amount of data to send, which the receiver can interpret as an amount of support to back up their claim. The sender's true dataset determines whether they are able to submit more or less data that fits the distribution, and it is optimal for the receiver distinguish them along this margin to encourage partial separation. When evidence is generated as in Table \ref{tab:ex_2_state}, the sender can send $m = \mu f_H$ if and only if the true data are $\mu' f_H$ with $\mu' \ge \mu$, or $\mu' f_L$ with $\mu' \ge \frac{3}{2}\mu$. The \emph{distinguishability factor} of $\frac{3}{2}$ reflects the relative advantage to a sender under $\theta_H$ of imitating $f_H$, and comes from the fact that in order to be able to submit enough observations of outcome $1$ to imitate $\mu f_H$, a sender under $\theta_L$ must start with $\frac{3}{2}$ as much data.

As a naive first guess, suppose that the sender's strategy is to always send the maximum possible amount of data that is distributed like $f_H$.
\begin{equation} \label{eq:naive_guess}
m_{max}(\mu f_H) = \mu f_H, \ \ \ m_{max}(\mu f_L) = \frac{2}{3}\mu f_L.
\end{equation}
Consider the uniform prior $\beta_0(\theta_H) = \frac{1}{2}$ and a data-mass distribution that is ``triangular",
\[
g(\mu) = 2 - 4|x-1/2|.
\]
The receiver's inference upon receiving a message $m_{max} = \mu f_H$, plotted by the solid line in Figure \ref{fig:binary_ex_2}(a), is
\[
\rho_{\max}(\mu) = \begin{cases}
1, \ \ \ \ \ \ \ \  \ \mu \ge 2/3\\
\frac{4-4\mu}{10 - 13 \mu}, \ \ \ \mu \in [1/2,2/3)\\
\frac{4\mu}{6 - 5\mu}, \ \  \ \ \ \mu \in [1/3,1/2)\\
\frac{4}{13}, \ \ \ \ \ \  \ \ \mu < 1/3.
\end{cases}
\]
To visualize how the receiver constructs the posterior inference, observe that the density of senders who send a message $\mu f_H$ for a $\mu$ for whom the true state is $\theta_H$ and $\theta_M$ are $g(\mu)$ and $\frac{3g(3\mu/2)}{2}$, which are plotted as two dotted lines. Their ratio is the likelihood ratio of the high vs. the low state given message $\mu f_H$.
\begin{figure}
    \centering
    \begin{subfigure}[t]{0.5\textwidth}
        \includegraphics[scale=0.4]{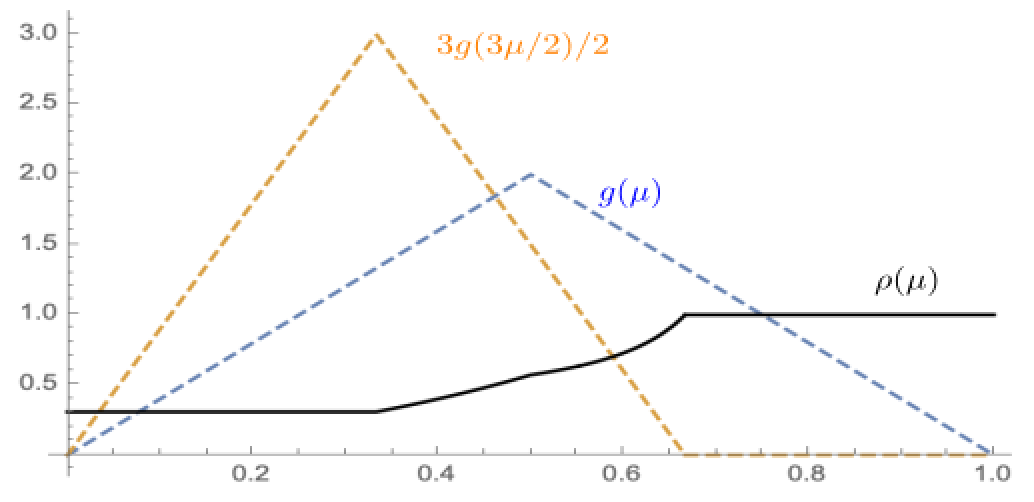}
        \caption{Posterior beliefs when $g$ is triangular}
    \end{subfigure}%
    \begin{subfigure}[t]{0.5\textwidth}
        \includegraphics[scale=0.4]{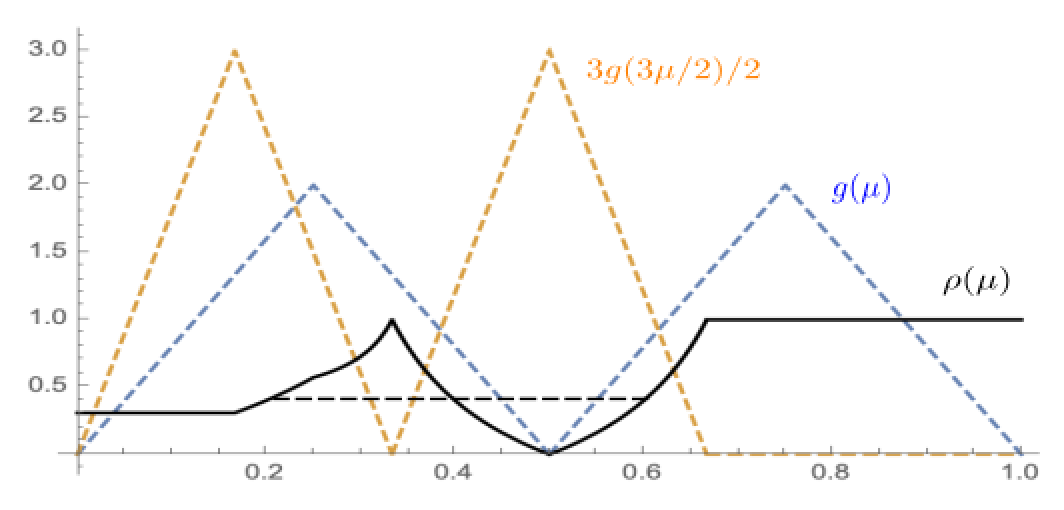}
        \caption{Posterior beliefs when $g$ is double-triangular}
    \end{subfigure}
    \caption{Inferences from message $m = \mu f_H$ in the binary-state example.}
    \label{fig:binary_ex_2}
\end{figure}

\noindent
\textbf{Observation 1.} \label{obs:only_r} \textit{$\rho_{\max}$ depends only on $\beta_0$, $g$, and the distinguishability factor.} 

In other words, the distinguishability of $f_H$ from $f_L$ is a sufficient statistic for both distributions that captures their implications for inferences under the naive strategy. In fact, we can verify that the naive messaging strategy in eq. \ref{eq:naive_guess} supports an equilibrium, under the assumption that any off-path messages feasible for some low-state type of the sender are evidence of the low state. More generally, the naive strategy is the unique imitation equilibrium strategy whenever it induces monotone inferences from the receiver. 

In some cases, $\rho_{max}(\mu)$ is nonmonotone, such as when $\mu$ takes the ``double triangular" distribution
\[
g(\mu) = \begin{cases}
2 - 8|x-1/4|,  \ \ \ \ \ x\in [0,1/2]\\
2 - 8|x-3/4|, \ \ \ \ \ x \in (1/2,1].
\end{cases}
\]
If all types of the sender send the maximal mass of data imitating $f_H$, then the message $1/2 f_H$ makes the receiver more pessimistic than the message $1/3 f_H$, and incentive compatibility fails because a sender who was to send the former would choose to send the latter instead. This is easily fixed, however, if, within a pooling interval, all types of the sender still imitate $f_H$, but send less than the maximal mass. The dashed line in Figure \ref{fig:binary_ex_2}(b) shows that the receiver's inferences given all messages in an interval can be equalized this way, so that the unique equilibrium inference is instead an ironed version of $\rho_{\max}(\mu)$.\footnote{The ironing process can be described as follows. If all types that would send $m = \mu f_H$ for some $\mu \in [\underline{\mu}, \bar \mu]$ were pooled, the receiver's inference given the pool would be 
\[
p(\underline{\mu}, \bar \mu) = \frac{\int_{\underline{\mu}}^{\bar \mu} g(\mu) d \mu}{\int_{\underline{\mu}}^{\bar \mu} \left(g(\mu) + \frac{3g(3\mu/2)}{2}\right) d \mu}.
\]
Given some $\mu^*$ at which $\rho_{\max}(\mu)$ is decreasing, we can find $\underline{\mu} < \mu^* < \bar{\mu}$, such that either $\rho_{\max}(\mu)$ is increasing at both $\underline{\mu}$ and $\bar{\mu}$, and $\rho(\underline{\mu}) = \rho(\bar \mu) = p(\underline{\mu}, \bar \mu)$;
or $\underline{\mu} = 0$ and $\rho_{\max}(\mu)$ is increasing at $\bar{\mu}$ with $\rho(\bar \mu) = p(\underline{\mu}, \bar \mu)$; or, $\bar \mu = 0$ and $\rho_{\max}(\mu)$ is increasing at $\underline{\mu}$ with $\rho(\underline{\mu}) = p(\underline{\mu}, \bar \mu)$.
There is a pair $(\underline{\mu}, \bar \mu)$ satisfying these criteria that are closest to $\mu^*$, and they are the endpoints of the ironing interval.
}

\subsubsection{A 3-state extension} \label{ex:3state}
Now suppose there is a 3rd possible quality of the prediction model, represented by state $\theta_M$. The medium-quality model yields a different distribution of predictions; to summarize, the distributions of the same $4$ outcomes under all states are given by Table \ref{tab:ex_2_3state}.\footnote{In the example weather-prediction application, the state-$M$ algorithm is better than the state-$H$ algorithm at calling the presence of rain, but worse at identifying when it will not rain. It is correct less often than the state-$H$ algorithm, but more often than the state-$L$ algorithm.}
\begin{table}[H]
    \centering
    \begin{tabular}{c|cccc}
         & $f_j(1)$ & $f_j(2)$ & $f_j(3)$ & $f_j(4)$ \\
         \hline
        H & 0.6 & 0.2 & 0.1 & 0.1 \\
        M & 0.4 & 0.25 & 0.3 & 0.05 \\
        L & 0.4 & 0.2 & 0.3 & 0.1
    \end{tabular}
    \caption{Data-generating distributions under states $\theta_H$, $\theta_M$ and $\theta_L$.}
    \label{tab:ex_2_3state}
\end{table}

Consider first the problem of a sender who knows that $\theta = \theta_L$. There are now $2$ distributions that they can imitate: $f_M$ and $f_H$. On the other hand, a sender for whom $\theta = \theta_M$ may wish to imitate is $f_H$, but never $f_L$. It takes at least $\frac{5}{4}\mu f_L$ and $\frac{3}{2}\mu f_L$ to imitate $\mu f_M$ and $\mu f_H$, respectively, and $2 \mu f_M$ to imitate $\mu f_H$. We can now keep track of three distinguishability factors, $r_L(M) = \frac{5}{4}$, $r_L(H) = \frac{3}{2}$, and $r_M(H) = 2$. 

Relative to the binary-state case, solving for the equilibrium when $|\Theta| \ge 3$ involves an extra step: understanding which state a sender will choose to target in imitation. Nevertheless, construction can proceed from the top down. First observe that types $\mu f_H$ with $\mu > \frac{1}{r_L(H)}$ can separate and obtain a payoff of $\theta_H$. We then ask which types of senders obtain a payoff $v \in (\theta_M, \theta_H)$. For this restricted set of payoff frontiers, it suffices to consider imitating $f_H$ only, since no message imitating $f_M$ can yield a payoff greater than $\theta_M$. Similarly to the binary-state case, in this regime the receiver can conjecture that the sender ``imitates as much of $f_H$ as possible", and restore monotonicity if needed by ironing. For payoff frontiers corresponding to $v < \theta_M$, one of two things is possible. If the state is $\theta_M$ and the sender has enough data to separate from all other types that cannot obtain $v > \theta_M$ by imitating $f_H$, then they play their full dataset and separate. Otherwise, unless the state is $\theta_L$, the sender plays their full dataset, but their full dataset is imitated by some type for whom the state is low, and who plays a strategy that mixes between imitating $f_H$ and $f_L$. Figure \ref{fig:3state_new} summarizes how the three distinguishability factors $r_L(M), r_L(H)$, and $r_M(H)$ determine the equilibrium: it projects all types onto a space that summarizes how imitable $f_H$ and $f_M$ are, as the vertical and horizontal dimensions, and shows their imitation strategies and payoffs in equilibrium.

\begin{figure}
    \centering
    \includegraphics[scale=0.6]{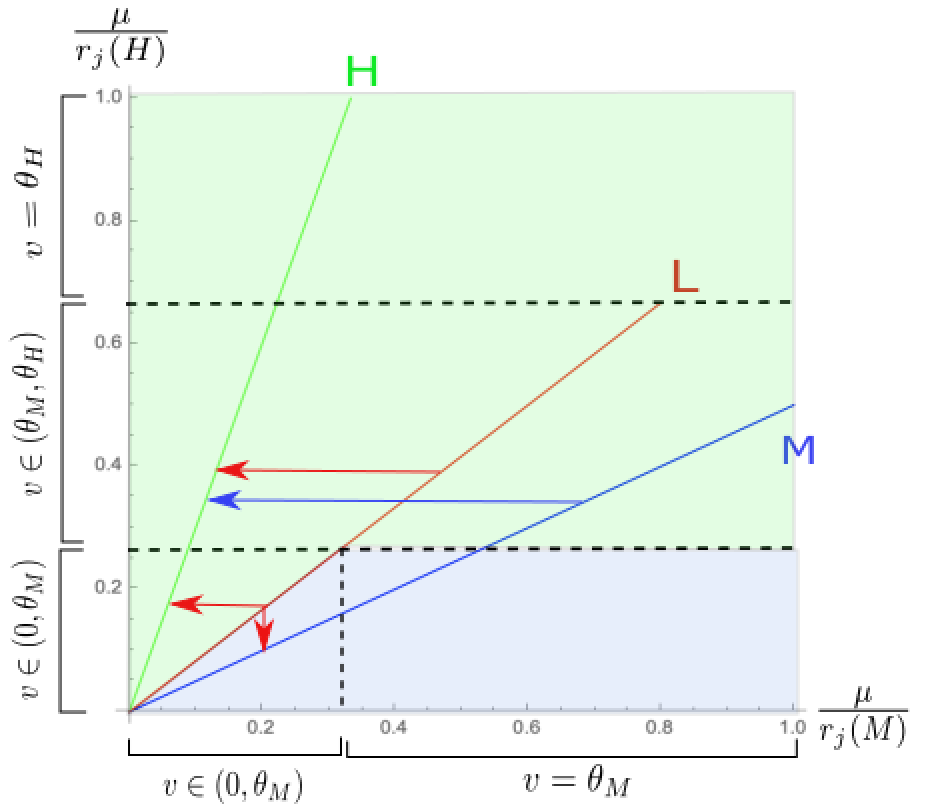}
    \caption{Under states $\theta_L$, $\theta_H$, and $\theta_M$, the sender either imitates $f_H$ (types in green region) or $f_M$ (types in blue region), or mixes (boundary).}
    \label{fig:3state_new}
\end{figure}
A novelty of the payoff structure with $3$, and indeed more, states is that the sender will separate and fully inform the receiver of the state only if they possess an intermediate amount of data -- ignoring the best and worst state, under $\theta_M$ there is a temptation to drop evidence with too much data, and an inability to distinguish oneself from imitators when too little data is acquired. The generality of multiple states also has other implications.

\noindent
\textbf{Observation 2.} \textit{The multi-state case has features that do not occur when $|\Theta| = 2$ or $|\mathcal{D}| = 2$:
\begin{itemize}
    \item Sending an interior mass of observations of some outcomes may be strictly optimal.
    \item Fixing the amount of data, the sender can receive greater payoffs under a state that is worse under full information.
\end{itemize}}
As an example of the first point, consider type $\mu f_L$ imitating $\frac{2}{3} \mu f_H$ by sending a mass $\frac{1}{15} \mu$ of observations of outcome $4$. Sending a greater mass would rule out the type $\frac{2}{3} \mu f_H$ that it wants to imitate, but sending less would rule in types like $(\frac{2}{15} - \epsilon)\mu f_M$, which would worsen the receiver's inference from the message. To demonstrate the second point, observe that the type $\mu f_L$ obtains a greater payoff than the type $\mu f_M$ when $\mu = 1$, because the former can imitate $\frac{2}{3} \mu f_H$, while the latter can only imitate $\frac{1}{2}\mu f_H$. 

\section{Construction and characterization}
This section characterizes the imitation equilibrium, constructs it, and shows that it is essentially unique. The imitation equilibrium is distinguished among equilibria by the fact that in it, worse types imitate better types (condition \ref{def:imitation_eq}b). This is directly reflected in the structure of the receiver's beliefs once they receive an on-path message $m$: the best case for any message is that the receiver takes it literally to be the sender's full dataset, while any skepticism that this is true negatively affects their inferences. Any off-path dataset $m \in \mathcal{T}$ might as well be taken literally,
\begin{equation} \label{eq:truth_leaning_beliefs}
q^*(\cdot|m) = \mathbbm{1}_{m}  \ \text{for all off path } m \in \mathcal{T},
\end{equation}
and is off path not because the receiver's inferences are ``artificially depressed" but because imitating some other dataset is strictly preferred for the type $t = m$. Therefore, the sender benefits from selective disclosure if and only if they lie -- there are no imitation equilibria that increase the payoff of truthful senders relative to their payoff when the receiver is fully informed. On the other hand, truthful senders can suffer -- since other senders can dishonestly imitate them, the receiver can be skeptical of their dataset even if they tell the truth.

In addition, for any dataset \emph{not} resembling some raw dataset, $m \not \in \mathcal{T}$, there are off-path beliefs
\begin{equation} \label{eq:truth_leaning_beliefs}
q^*(\cdot|m) = q^*(t|\underset{t' \supseteq m, \ t'\in \mathcal{T}}{\arg \min}  \ \mathbb{E}_{\beta(\cdot|\sigma^*(t'))}[\theta])  \ \text{for all } m \in \mathcal{M} \setminus \mathcal{T},
\end{equation}
and given these beliefs, senders never benefit from playing a dataset that the receiver knows for sure to be incomplete. Because of this, an observer of the interaction between senders and receivers would not be able to tell if senders are strategically omitting data simply by looking at the distributions of the published data -- some prior about how much data the sender ought to have is necessary to know if observations are being dropped. 

We have established that, in an imitation equilibrium, a sender's ability to positively influence the receiver depends on the extent to which they can imitate another state. In turn, this depends on the mass of their own dataset, $\mu$, and the extent to which $f_k$ can be distinguished from $f_{j}$, which is  given by
\[
r_j(k) = \max_{d \in \mathcal{D}} \frac{f_{k}(d)}{f_j(d)}.
\]
This distinguishability factor $r_j(k)$ is a measure of the comparative advantage to a sender under state $\theta_k$ to reporting a dataset distributed like $f_k$, relative to a sender under state $\theta_j$.\footnote{Equivalently, we can consider its inverse, $\frac{1}{r_j(k)}$, an \emph{imitability} factor that describes how easily $f_k$ is imitated under state $\theta_j$.} It can be interpreted to mean that ``under state $\theta_j$, a sender would need $r_j(k)$ times as much data to imitate $\mu f_k$ than under $\theta_k$". A sharp feature of the continuum model is that pairwise distinguishability comparisons fully suffice to summarize the impact of the shape of generating distributions $\{f_j\}_{j=1}^J$ on the imitation equilibrium outcome. 
\begin{proposition}[Sufficiency] \label{prop:sufficiency}
    Two games $\mathcal{G}$ and $\mathcal{G}'$ must yield the same outcome if they share the same state space $\Theta$ and priors $\beta_0$ and $G$, and for all $j$ and $k$,
    \[
    \max_{d \in \mathcal{D}} \frac{f_k(d)}{f_j(d)} \equiv r_j(k) = r_j'(k) \equiv \max_{d' \in \mathcal{D}'} \frac{f_k'(d')}{f_j'(d')}.
    \]
\end{proposition}
In other words, even if $\mathcal{D}$ is very large, $\{f_j\}_{j=1}^J$ only affect the menu of possible beneficial manipulations through a select set of summary statistics, which are each supported by a single point in $\mathcal{D}$. We will delay discussion of the comparative statics of distinguishability, as well as their implications for optimal experimental design, to section \ref{sec:exp_design}. 
Now we leverage these factors to complete our characterisation of the imitation equilibrium. All equilibrium outcomes can be described by a vector-valued function $\hat{\mathbf{u}}(\mu) = (\hat u_j(\mu))_{j=1}^J$, with $\hat \mu_j(\mu) = u_{\sigma}(\mu f_j)$. The imitation equilibrium outcome has an even simpler description: each sender's messaging problem can be simplified down to the choice of a weakly better state to imitate, $k \in \{j,\ldots, J\}$, and an amount $\mu$ of that state's distribution to send -- ``as much as possible" is always weakly optimal, though, as with ironing in example \ref{ex:2state}, may not be the only strategy played in equilibrium. Since $\hat u$ describes the payoff under every state to every $\mu$, its inverse $\hat{\boldsymbol{\mu}}$, defined as
\[
\hat \mu_j(u) = \min\{\mu:\hat u_j(\mu) \ge u\},
\]
describes a \emph{burden of proof} in order to achieve payoff $u$, and what is necessary is that a type $t$ can provide at least a measure $\hat \mu_{k}(u)$ of distribution $f_{k}$, where $\theta_{k} \ge u$. Crucially, fixing the pairwise distinguishability factors, optimality of the sender's imitation strategy amounts to saying that a sender that achieves payoff $u$ via imitation is either truthful with $\theta_j \ge u$ or  imitates another state $\theta_{k} > u$ in the set
\[
A_j(\mu) = \left\{\theta_{k}: k \in \arg \max_{k > j}\hat u_k\left(\frac{\mu}{r_j(k)}\right)\right\}.
\]

\begin{theorem}[Existence and uniqueness] \label{thm:imitate}
There there exists an essentially\footnote{$\beta^*$ is uniquely determined, and $\sigma^*$ is uniquely determined up to payoff-irrelevant mixing probabilities.} unique imitation equilibrium, implemented by a vector-valued burden of proof function $\hat{\boldsymbol{\mu}}: [0,\theta_J] \rightarrow \mathbb{R}^J$ with outcome $\hat{\mathbf{u}}$ such that
\begin{enumerate}
    \item $\hat u_j(\mu)$ is continuous and (weakly) increasing in $\mu$ for all $j$.
    \item $\sigma^*(\mu f_j)$ is supported on $
    \left\{\mu' f_{k}: \mu' = \hat \mu_{k}\left(\hat u_{k}\left(\frac{\mu}{r_j(k)}\right)\right) \text{ and } \theta_{k} \in A_j(\mu)\right\}$.
\end{enumerate}
\end{theorem}

\subsection{Construction of the equilibrium}

\begin{figure}
    \centering
    \includegraphics[scale=0.6]{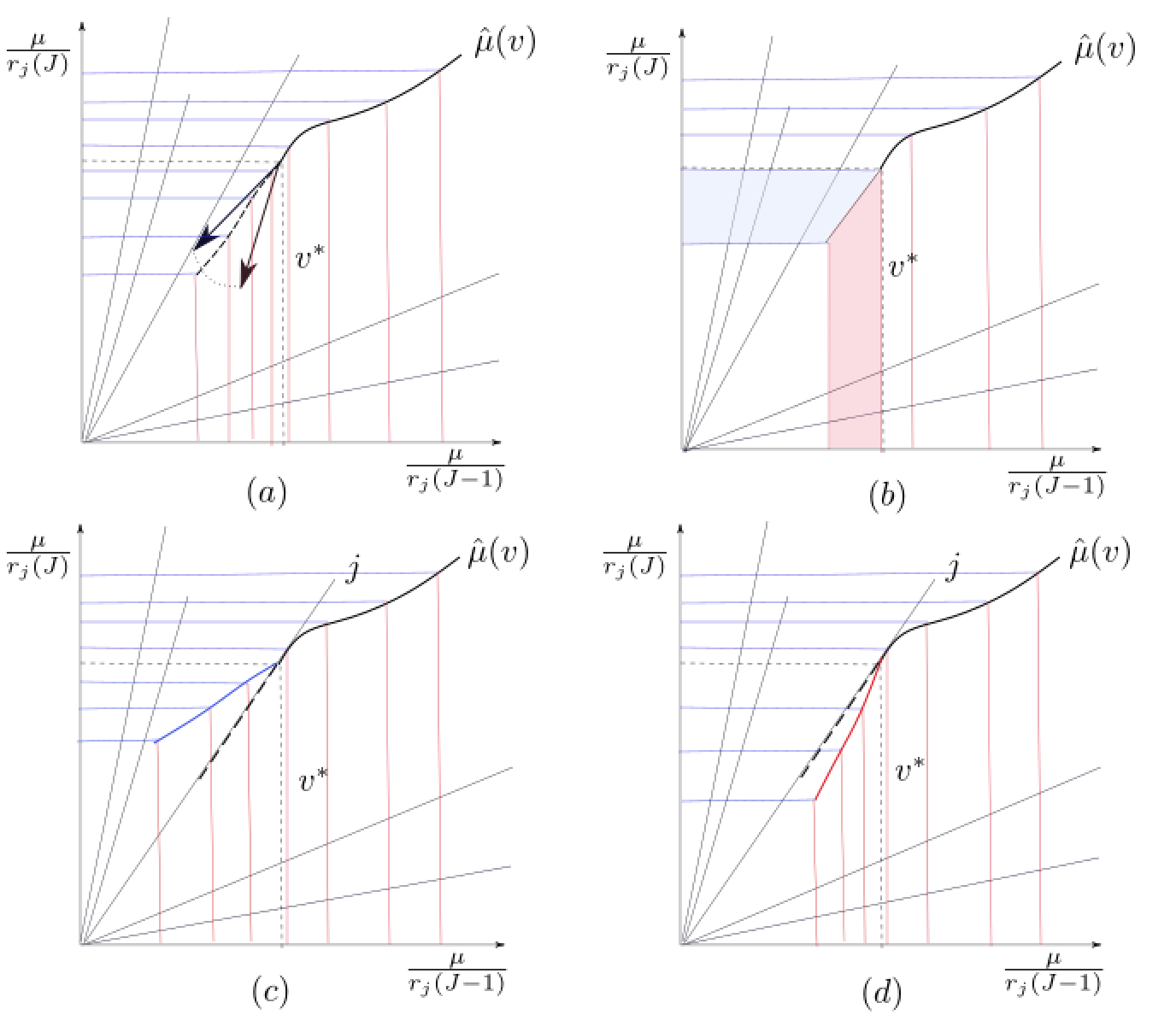}
    \caption{In equilibrium, $\hat{\boldsymbol{\mu}}(v)$ equalizes payoffs to imitating each state $\theta_k > v$. Rays represent types in $\mathcal{T}$ and red and blue lines represent payoff frontiers to those imitating $f_{J-1}$ and $f_J$, respectively.
    }
    \label{fig:construction}
\end{figure}

In the Appendix, we give the details of the step-by-step construction of $\sigma^*$ in general. But to capture the main idea, here we describe a minimal setup that illustrates the forces at play. Consider the problem of constructing $\hat{\boldsymbol{\mu}}(v)$ for $v \in [\theta_{J-1}, \theta_J]$, assuming that $\hat{\boldsymbol{\mu}}(\theta_{J-1})$ is known. Fig. \ref{fig:construction} shows that a typical type space can be projected onto 2 dimensions: one dimension describes the ability of each type to imitate $f_J$, given by $\frac{\mu}{r_j(J)}$, and the other dimension describes their ability to imitate $f_{J-1}$, given by $\frac{\mu}{r_j(J-1)}$. We can plot senders with all possible amounts of data under a given state as a ray when we describe the type space this way. Since any $v > \theta_{J-1}$ is obtained through imitating one of these two types, this description is sufficient to determine the imitation strategies used to obtain this subset of responses from the receiver.

The burden-of-proof vector lies in the same space and describes two simple things: which of the two states each type imitates, and what the highest action is that they can induce the receiver to take by doing so. A couple of observations allow us to identify the unique continuation of $\hat{\boldsymbol{\mu}}(v)$ at and to the left of any $v^*$ whenever $\hat{\boldsymbol{\mu}}(v)$ is already known for all $v > v^*$. 

Taking the higher payoff frontiers to be fixed, focus on the set of types unable to meet any component of $\hat{\boldsymbol{\mu}}(v)$ for any $v > v^*$. There may exist within this set a self-separating set of positive measure that can pool with each other to induce action $v^*$. Fig. \ref{fig:construction}(b) shows that if so, $\hat{\boldsymbol{\mu}}(v)$ is discontinuous at $v^*$, since the equilibrium construction then immediately pools these types and assigns them all a payoff of $v^*$. 
Otherwise, $\hat{\boldsymbol{\mu}}(v^*) = \lim_{\epsilon \rightarrow 0} \hat{\boldsymbol{\mu}}(v^* + \epsilon)$.

The key fact is that given $\hat{\boldsymbol{\mu}}(v^*)$, it is always possible to exactly specify $\hat{\boldsymbol{\mu}}(v)$ for $v$ in some, possibly small, nonempty interval $(v^*-\Delta, v^*)$. Consider first the case in which all types in the $v^*$-payoff frontier strictly prefer to imitate either $f_J$ or $f_{J-1}$. When $\hat \mu_j(v^*) f_j$ imitates distribution $f_{J}$, then all types $\mu f_j$ with $\mu$ close to $\hat \mu_j(v^*)$ behave likewise, and the same is true for those imitating distribution $f_{J-1}$. In other words, the payoff frontiers are locally determined because imitation strategies are fixed, up the amount of data submitted. Panel (a) of Fig. \ref{fig:construction} shows that $\hat{\boldsymbol{\mu}}(v)$ then follows along the path of equivalent payoffs from imitating either state, and is continuous, due to the continuity of $g$.

A second possibility is that for some $j$, type $\hat \mu_j(v^*) f_j$ may indeed be indifferent between imitating $f_J$ and $f_{J-1}$, and mixes between the two with interior probability. Locally, for $\mu$ close to $\mu_j(v^*)$, the types $\mu f_j$ must also be indifferent, and so for a set of values $v \approx v^*$, $\hat{\boldsymbol{\mu}}(v)$ coincides with the set of types under state $\theta_j$ that achieve the corresponding payoff. For all state-$\theta_j$ senders obtaining a payoff in this range, the mixed strategy played equalizes payoffs to imitating each of the two highest states. Fig. \ref{fig:construction}(c) shows that if $\sigma(\frac{\mu}{r_j(J)} f_J| \mu f_j)$ increases too quickly, this fails to hold, since then payoffs to imitating $\theta_J$ decrease quickly relative to those to imitating $\theta_{J-1}$, and (d) shows that payoffs to imitating $\theta_J$ decrease too quickly in the opposite case. There is, then, a unique continuation of the mixed strategy that respects the restriction on $\hat{\boldsymbol{\mu}}$, and it is continuous due to the continuity of $g$. 

When there are more than $2$ candidate states to imitate, the construction is slightly more complicated in that there may be more than one state under which types are indifferent across distributions to imitate, and a given type may be indifferent between imitating more than 2 different states. Nevertheless, the idea is the same. It is always possible to construct an interval of frontiers and their associated equilibrium strategies, given knowledge of higher-payoff frontiers. The construction technique then proceeds interval-by-interval, where we note that each interval formed in a step of the process is nonempty but may be small: it may be necessary to switch from handling the problem as in the first case to handling it as in the second case, and vice versa, multiple times as the algorithm proceeds to successively lower payoff frontiers. 

\subsection{A separation theorem}

Let us return briefly to the matter of why the imitation outcome stands out from other equilibrium outcomes. It turns out that, although we can construct the imitation equilibrium payoff frontiers iteratively, we can also characterize them each individually, and independently of the remainder of the equilibrium. Put simply, imitation equilibrium payoff frontiers universally divide the type space into a greater-value upper region and a lesser-value lower region, and they are the only frontiers to do so.

We start with some definitions.
\begin{definition}
    An \emph{upper pool} of payoff frontier $\hat{\boldsymbol{\mu}}(v)$ is a set
    \[
    \bar{T} = U(\hat{\boldsymbol{\mu}}(v)) \setminus U(M)
    \]
    for some collection of messages $M$.
\end{definition}

\begin{definition}
    A \emph{lower pool} of of payoff frontier $\hat{\boldsymbol{\mu}}(v)$ is a set
    \[
    \underline{T} =  U(M) \setminus U(\hat{\boldsymbol{\mu}}(v))
    \]
    for some collection of messages $M$.
\end{definition}
An upper pool consists of all types above the payoff frontier $\hat{\boldsymbol{\mu}}(v)$ but below some other frontier, while a lower pool consists of types below it but above another frontier.

We define the pooled value of any set of types, $u_{pool}(T)$, to be the receiver's expectation of the state given that the sender's type is in the set $T$, and state the separation theorem:
\begin{theorem}[Separation] \label{thm:separation}
	For any nonempty upper pool $\bar{T}$ and lower pool $\underline{T}$ of $\hat{\boldsymbol{\mu}}(v)$,
	\[
	u_{pool}(\bar T) \ge v > u_{pool}(\underline{T}).
	\]
\end{theorem}
In other words, upper pools are weakly improving and lower pools are strictly worsening — for any subset of $\mathcal{T}$ that is bounded by two frontiers and contains $\hat{\boldsymbol{\mu}}(v)$, the value of the part above $\hat{\boldsymbol{\mu}}(v)$ is at least $v$, while the value of the part below is less than $v$.\footnote{The former inequality is weak and the latter strict because we have defined the imitation payoff frontiers such that, when there are multiple types $\mu f_k$ that all achieve $v$, $\hat \mu_k(v)$ is the lowest such $\mu$.} 

The fact that upper pools are improving is a consequence of the conditions of imitation equilibria: the property holds because in each group of senders who send the same message under $\sigma^*$, only those with worse-than-average values can be truncated by excluding $U(M)$. On the other hand, the equilibrium we construct has worsening lower pools because in it, any potentially self-separating pool of senders below $\lim_{\epsilon \rightarrow 0}\hat{\boldsymbol{\mu}}(v+\epsilon)$ that achieves a value of at least $v$ must lie above the frontier $\hat{\boldsymbol{\mu}}(v)$.

These properties guarantee uniqueness of the imitation equilibrium outcome if we use them to compare outcomes under $\sigma^*$ and another PBE, $\sigma$. If the outcome under $\sigma$ differs from that under $\sigma^*$, then worsening lower pools under $\sigma^*$ imply that there is a frontier with a worsening upper pool under $\sigma$. Moreover, the only frontiers in $\mathcal{T}$ that satisfy either property are the frontiers of $\sigma^*$. Given any prospective frontier and its associated payoff, checking either of these properties in isolation is enough to verify that it shows up in the imitation equilibrium, and may in some cases be easier than constructing the entire imitation equilibrium outcome. 

The separation theorem is a general result — it also applies to finite evidence games, where it is related to the “downward biased” characterization of \citet{Rappoport22}. In all these cases, worsening lower pools rules out credible inclusive announcements, and improving upper pools turns out to imply that no other equilibrium is credible inclusive announcement-proof.

\section{Comparative statics} \label{sec:comp_stats}
The burden of proof characterization of equilibrium and the separation theorem characterization of payoffs imply that all imitation outcomes share some concrete features. Here, we present comparative statics of the sender's reports in $\mu$, of the sender's welfare with respect to the receiver's prior belief about $\theta$ and $\mu$, and of separation as $Var[\mu] \rightarrow 0$. We begin with a corollary to Theorem \ref{thm:imitate}.   

\begin{corollary} \label{thm:segment}
Under $\sigma^*$, there are thresholds $z^*_j > z^{**}_j$ for each state such that:
\begin{itemize}
    \item Whenever the sender's type is $\mu f_j$ with $\mu > z^*_j$, the sender masquerades as a higher type, and receives a payoff $\hat u_j(\mu) > \theta_j$. 
    \item Whenever $\mu \in (z^{**}_j, z^*_j]$, the sender is honest and the receiver knows it upon receiving the data: $\hat u_j(\mu) = \theta_j$.
    \item Whenever $\mu \le z^{**}_j$, the sender is honest, but the receiver believes they are a worse type with positive probability, and $\hat u_j(\mu) < \theta_j$.
\end{itemize}
\end{corollary}

We can think of senders with $\mu > z^*_j$ as high-data senders, with enough data to benefit from manipulating their data against the receiver's uncertainty about their data endowment. The costs of voluntary disclosure are borne by low-data senders, those with $\mu < z^{**}_j$, who the receiver is skeptical of even when they are truthful. These thresholds vary by $j$, and in particular, $z^*_1 = 0$ and $z^{*}_J = 1$. However, they need not be monotone in $j$.

The potential presence of an intermediate, full-information interval between disjoint upper and lower partial-pooling intervals when we fix $\theta$ and vary $\mu$ is a novel feature of these equilibria that occurs when there are multiple imitated states with different distinguishing outcomes. It is a consequence of the fact that it requires a strictly greater amount of data to benefit from imitating a different state than it does to send one's full dataset and discourage all imitators. The structure of pooling and separation contrasts with strategies in binary-state models of voluntary disclosure, or in models with ordered outcomes. In those cases, full separation only occurs at the very top, that is, for types with a maximal state and a maximal amount of evidence (see, for example, \citet{Dye85} and \citet{Dzuida11}). We show that this doesn't have to be true in general: although they remain able to separate, types with the most evidence are often more tempted to pool with others.

We have shown that even the receiver-optimal equilibrium must contain partial pooling. The basic reason for this is that low-data senders under high states can never separate themselves from high-data senders under low states: the most they can do to distinguish themselves is send their full dataset, but even so, the higher-data sender can imitate them. The extent of pooling in general depends on two things, the structure of the data, which is reflected through $\{r_j(k)\}_{j < k}$, and uncertainty about $\mu$, which is reflected in $g$. We discuss the former in the next section, and focus here on the receiver's beliefs. In the absence of uncertainty about $\mu$ -- that is, if $\mu$ is commonly known to the sender and the receiver -- the disclosure game is a case of the games studied by \citet{Grossman81} and \citet{Milgrom81}, in which unraveling occurs. The distribution of $\mu$ in our model, while nondegenerate, can be arbitrarily close to a point mass, and outcomes vary towards the full-information outcome continuously as the receiver's uncertainty about $\mu$ vanishes.
\begin{claim} \label{claim:var}
As $Var[\mu] \rightarrow 0$, $\hat u_j(\mu) \rightarrow \theta_j$ for all $\mu f_j \in \mathcal{T}$.
\end{claim}

Outcomes also vary monotonically towards the full-information outcome with increasing certainty about the state. When the receiver's belief about the ex-ante probability of a given state $\theta_j$ increases relative to others, the receiver's skepticism weakly increases for all messages that yield a higher payoff to the sender than full certainty of that state. The reverse is true of all messages that yield a lower payoff than $\theta_j$. An increase in the probability of $\theta_j$ therefore ``pulls" the receiver's action towards $\theta_j$ given any message, which has the consequence of decreasing ex-post payoffs for all types of the sender that would originally have achieved $\hat u_j(\mu) \ge \theta_j$, and increasing them if originally, $\hat u_j(\mu) \le \theta_j$.

To formalize this, let $\mathcal{G}$ be a disclosure game with prior $\beta_0$ about $\theta$ and $\mathcal{G}'$ be a game that is identical except for the prior $\beta_0'$ which differs from $\beta_0$, with $\beta_0'(\theta_j) > \beta_0(\theta_j)$ and $\frac{\beta_0(\theta_k)}{\beta_0(\theta_{k'})} = \frac{\beta_0'(\theta_k)}{\beta_0'(\theta_{k'})}$ for all other $k, k'$. 
\begin{claim} \label{claim:beta}
    Suppose that $\hat{\mathbf{u}}, \hat{\mathbf{u}}'$ are imitation equilibrium outcomes of $\mathcal{G}$ and $\mathcal{G}'$, respectively. Then $\hat u_{j'}(\mu) \ge \hat u_{j'}(\mu)$ whenever $\hat u_{j'}(\mu) \ge \theta_j$, and $\hat u_{j'}(\mu) \le \hat u_{j'}(\mu)$ whenever $\hat u_{j'}(\mu) \le \theta_j$.
\end{claim}

Finally, we point out that first-order shifts in the receiver's beliefs have a monotone impact on the sender's welfare. Simply put, every type of the sender benefits from a first-order shift in the receiver's belief about the state, and suffers from a first-order shift in their belief about the data-mass distribution. Intuitively, an upwards shift in the prior distribution in $\theta$ makes the receiver more willing to believe a claim that the state is high; this unambiguously benefits the sender both conditional their realized dataset, and ex-ante. On the other hand, when the receiver expects $\mu$ to be greater, they are more \emph{skeptical}: they infer a greater likelihood that a given message may have been selected from a larger dataset.\footnote{\citet{Rappoport22}'s result can be used to show that the latter holds in finite-data games viewed as an instance of an abstract evidence game, and a similar argument shows that this is directly true in the continuum.} 

\begin{claim}
    If two disclosure games $\mathcal{G}$ and $\mathcal{G}'$ are identical except for priors $\beta_0 \le_{FOSD} \beta_0'$ and $g \ge_{FOSD} g'$, then
    \[
    \hat u_{j}(\mu) \le \hat u_j'(\mu) \ \ \forall \mu f_j \in \mathcal{T}.
    \]
\end{claim}

\section{Experimental design} \label{sec:exp_design}
Our results highlight that the quality of the information the receiver obtains depends on how the data-generating process distinguishes states. This section focuses on interventions that aim to maximize distinguishability, and proposes a framework for optimally designing experiments to allow the receiver to extract payoff-relevant information from the sender through voluntary disclosure. In our model, an \emph{experiment} is the data-generating process that provides the sender with their raw dataset, and is captured by a tuple $\mathcal{E} = (\mathcal{D}, \{f_j\}_{j=1}^J)$ consisting of the space of reported outcomes and the generating distribution of data over them. We assume that the remaining primitives of the game -- state space, payoffs, and priors -- are fixed, and consider the effect of varying the experiment that the sender observes. 

A key fact is that whenever an experiment makes states pairwise more distinguishable, the receiver's welfare improves. Intuitively, increasing distinguishability allows higher-state types to separate themselves more effectively from lower-state types who would imitate them. The resulting equilibrium does not better separate every type from every other type -- indeed there are types that would play different messages under one experiment that would play the same message in the other, in both directions -- but, given the receiver's single-peaked expected utility, the more distinguishing experiment always makes the receiver better able to target the optimal action.\footnote{The proof that distinguishability improves payoffs uses the fact that a mechanism designer that takes a sender's submitted dataset as a report is weakly more constrained by a sender's ability to deviate to sending a false dataset if the experiment has poor distinguishability. If the receiver does not have single-peaked preferences, then the imitation equilibrium outcome and the outcome of the optimal mechanism do not necessarily coincide, and increasing distinguishability may force the receiver to take a higher action after observing a message that few low-state types can imitate, when they would instead like to commit to responding to it with a lower action.}
\begin{proposition}[Improvement]
    Suppose two experiments $\mathcal{E}$ and $\mathcal{E}'$ yield imitation equilibrium actions $a$ and $a'$, respectively. 
    \begin{itemize}
        \item If $r_j'(k) \ge {r_j}(k)$ for all $k > j$, then $\mathbb{E}_{a, \theta}[u_r(a)] \le \mathbb{E}_{a’,\theta}[u_r(a’)]$.
        \item If, in addition, $r_j'(k) > r_j(k)$ for some $j, k$ such that there is some $\mu f_j$ imitating $\hat \mu f_k$ under the imitation equilibrium with experiment $\mathcal{E}$, then $\mathbb{E}_{a, \theta}[u_r(a)] < \mathbb{E}_{a’,\theta}[u_r(a’)]$.
    \end{itemize}
\end{proposition}
In fact, by making a state arbitrarily distinguishable from others, we can guarantee that a sender under that state elicits at least their full-information action with high probability: for every $\delta > 0$, there is $R < \infty$ such that $Pr(\hat u_k(\mu) < \theta_k) < \delta$ as long as $r_j(k) > R$ for all $j < k$, where the likelihood is taken over realizations of $\mu$. In the limit as all states become highly distinguishable, the receiver also approximately attains their full information payoff.

One way to better distinguish two states is to undertake a more detailed experiment. Without changing the experimental technology -- that is, the underlying likelihood of events under different states -- a researcher could investigate and record a more detailed set of outcomes in order to obtain finer data. To formalize this, suppose there is an existing outcome space $\mathcal{D}$, and consider a notion of a more elaborate outcome space $\mathcal{D}'$ that the researcher can obtain by splintering an existing outcome into multiple sub-outcomes to track. 
\begin{definition}
If there are two experiments $\mathcal{E} = (\mathcal{D}, \{f_j\}_{j=1}^J)$ and $\mathcal{E}' = (\mathcal{D}', \{f_j'\}_{j=1}^J)$ and a partition $\mathcal{P} = \{P_{d}\}_{d \in \mathcal{D}}$ of $\mathcal{D}'$ such that
\[
\sum_{d' \in P_{d}} f_j'(d') = {f_j}(d)
\]
for all $d$ in $\mathcal{D}$, then $\mathcal{E}'$ \emph{splinters} the outcome space of $\mathcal{E}$ and $\mathcal{E}$ \emph{merges} the outcome space of $\mathcal{E}'$.
\end{definition}
Immediately, we observe that for all $\theta_j$ and $\theta_k$,
\[
\max_{d' \in P_d} \frac{f_k'(d')}{f_j'(d')} \ge \frac{f_k(d)}{f_j(d)},
\]
and so $r_j(k) \ge r_j’(k)$ whenever $\mathcal{D}'$ splinters $\mathcal{D}$.
\begin{claim}
Splintering the outcome space weakly improves the receiver's expected payoff.
\end{claim}

In some cases, there is a most elaborate possible experiment $\mathcal{E}^*$, i.e., one that is a splintering of every other possible experiment. Suppose that costs and constraints on gathering, storing, and transmitting data are negligible. Then it is optimal for a designer who acts on behalf of the receiver to choose the most elaborate possible experiment. If instead the sender chooses the experiment, then the receiver should, if possible, incentivize the sender to choose the most detailed experiment by committing to accept nothing else. Since they follow a simple rule of thumb, these recommendations don't require detailed knowledge of the true data-generating process, and would be easy for even an uninformed designer to implement.

On the other hand, in practice there is often no binding limit to the number of ways that an experiment can be refined and complicated, at ever increasing cost. Despite the fact it never hurts, further splintering a dataset does not always strictly improve distinguishability. With precise information about the data-generating process, Proposition \ref{prop:sufficiency} allows us to identify instances when it is without loss to the receiver to merge outcomes relative to $\mathcal{E}^*$. 

\begin{proposition}[Merging]
    Suppose that $\mathcal{E}^* = (\mathcal{D}^*, \{f_j^*\}_{j=1}^J)$. Let $S^* = \bigcup_{j < k} \arg \max_{d} \frac{f_k^*(d)}{f_j^*(d)}$. Then merging all outcomes in $\mathcal{D}^* \setminus S^*$ does not change the imitation equilibrium outcome.
\end{proposition}

The set $S^*$ consists of all outcomes that maximally distinguish one state from another, and that merging other outcomes is without loss follows from the fact that $S^*$ is sufficient to maximize every distinguishability factor in $\{r_j(k)\}_{j < k}$. We are left, generically, with a minimal experiment that suffices to reveal as much payoff-relevant information as possible to the receiver robustly over all possible priors. 

\begin{claim}[Minimality]
    Fix $u_r$, $\Theta$, and $g$, suppose that $\mathcal{E}$ is obtained from $\mathcal{E}^*$ by merging $S^*$ and that $\mathcal{E}'$ merges some outcomes in $\mathcal{E}$, and suppose that $\arg \max_{d} \frac{f_k^*(d)}{f_j^*(d)}$ is unique for all $j < k$.
    
    \noindent
    Then there exists $\beta_0$ such that the receiver is strictly better off with $\mathcal{E}$ than with $\mathcal{E}'$.
\end{claim}

This simplification of the experiment can be quite drastic, and in some familiar cases, including the case of a binary state space or an outcome space ordered by the monotone likelihood ratio property (MLRP), $S^*$ is a singleton with only one ``good news" outcome that maximally distinguishes higher states from lower ones, while all other outcomes in $\mathcal{D}^*$ can be merged and essentially ignored.\footnote{The imitation equilibrium in these cases has the same outcome as a sanitation equilibrium (\citet{Shin03}) in which the sender only reports observations of the outcome in $S^*$, and omits all others; however, it differs in that imitating senders generally report a positive mass of observations of these outcomes anyways, with no impact on the receiver's inferences.} Formally, we say that $\mathcal{D}^*$ satisfies MLRP with respect to $\{f_j^*\}_{j=1}^J$ if, for any $j < k$ and $d < d'$,
\[
\frac{f_k^*(d')}{f_j^*(d')} > \frac{f_k^*(d)}{f_j^*(d)}.
\]
It is straightforward to see that MLRP implies that $S^*$ comprises of the single maximal element in $\mathcal{D}^*$.

Even when $J > 2$, some degree of dimensionality reduction is often possible, especially if $J << |\mathcal{D}^*|$. In general, $|S^*| \le \frac{J(J-1)}{2}$. The 3-state example \ref{ex:3state} gives an instance in which this bound is tight because the maximizer, $\arg \max_d \frac{f_k(d)}{f_j(d)}$, is unique for all pairs $j < k$.

\begin{corollary}
    The minimal optimal experiment tracks at most $\frac{J(J-1)}{2} + 1$ outcomes, and furthermore, if $\mathcal{D}^*$ satisfies MLRP with respect to $\{f_j^*\}_{j=1}^J$, then a binary outcome space suffices.
\end{corollary}

\section{Relationship to finite data}
In the big picture, the purpose of modeling communication in a stylized, continuous-data disclosure game is to understand how senders will volunteer data in real-world disclosure settings, in which datasets are always finite. The comparative statics of section \ref{sec:comp_stats} and the experimental design implications of the previous section depend on the fact that datasets are well-described by $\mu$ and $f_j$, which is exactly true only in the continuum, but nearly true with large $N$ in such a way that those results approximately carry over. This section makes precise the finite-data settings that we aim to approximate, and describes how the continuous-data model captures their regularities in the limit.

We model a sender who has access to a finite dataset of $n$ i.i.d. observations drawn from $\mathcal{D}$ according to the state-contingent distribution $f_j$. The size of the sender's dataset is upper-bounded by $N$, but the sender may have access to $n <N$ observations as well, and the receiver is uninformed about how much data the sender has. Nature's sequence of moves in drawing the sender's dataset is: 1) draw the state, $\theta_j$, according to prior $\beta_0$; 2) draw the number of observations, $n$, from distribution $G_N$; 3) for each of the $n$ datapoints, draw their realized value i.i.d. from $f_\theta$. Call the disclosure game with these parameters $\mathcal{G}_N(\Theta, \mathcal{D}, \beta_0, \{f_j\}_{j=1}^J, G)$. The data mass distributions $G_N(\cdot)$ capture the receiver's uncertainty about how much raw evidence the sender has, prior to selecting observations to reveal: for example, there may be uncertainty about the number of total trials in an experiment, or the number of trials out of $N$ attempted that survived the entire trial period.

The sender's dataset is the empirical probability mass function $t = \frac{1}{N}(t_1, \ldots, t_D)$, where $t_d$ is the number of observations of outcome $d$ and $n(t) = \sum_{d=1}^D t_d$ is the number of observations they get. They are able to send any subset of their dataset as a message to the receiver, where
\[
m \tilde \subseteq t \ \Leftrightarrow \ m_d \le t_d \ \forall d \in \mathcal{D}.
\]
In summary, the type space is $\mathcal{T}_N = \bigcup_{n = 0}^N \mathcal{D}^n$, with type distribution
\[
     q_N(t) = \frac{n(t)!}{\Pi_{d=1}^D t_d!} \sum_{j'}\beta_0(\theta_{j'}) g_N(n(t))\Pi_{d=1}^D f_{j'}(d)^{t_d},
\]
and the message space $\mathcal{M}_N$ is identical to the space of types.

When datasets are finite, the sender's dataset does not perfectly inform them about the state: when $f_j$ all have full support, any state is possible after observing any dataset. The likelihood of $\theta_j$ given that the raw dataset is $t$ is
\[
\pi_N(\theta_j|t) = \frac{\beta_0(\theta_j) g_N(n(t))\Pi_{d=1}^D f_j(d)^{t_d}}{\sum_{j'}\beta_0(\theta_{j'}) g_N(n(t)) \Pi_{d=1}^D f_{j'}(d)^{t_d}},
\]
and so, when the receiver observes a message and updates their belief about the sender's type to $q_N(t|m)$, their posterior about the state updates to
\begin{equation} \label{eq:state_beliefs}
    \beta(\theta_j|m) = \frac{\sum_{t \in \mathcal{T}_N}q_N(t|m)\pi_N(\theta_j|t)}{\sum_{t \in \mathcal{T}_N}q_N(t|m)}.
\end{equation}

We highlight that the distribution of datasets in the finite-data setting converges to the distribution of datasets in a continuous-data model. In particular, $g(\mu)$ represents the likelihood of obtaining a fraction $\mu$ of total potential data under state $j$, and analogously, $\frac{n}{N}$ is the fraction of total data available to the sender in the finite-data game. We can study a sequence of games such that as $N$ increases, $N G_N(\frac{n}{N}) \rightarrow_{unif.} G(\mu)$, and note that if so, the type distributions also converge uniformly: $q_N \rightarrow_{unif.} q$. 
\begin{definition}
    $\mathcal{G}(\Theta, \mathcal{D}, \beta_0, \{f_j\}_{j=1}^J, G)$ is the \emph{limit game} for a sequence of finite-data games 
    \newline
    $\{\mathcal{G}(\Theta, \mathcal{D}, \beta_0, \{f_j\}_{j=1}^J, G_{N})\}_{N=1}^\infty$ if $N G_N(\frac{n}{N}) \rightarrow_{unif.} G(\mu)$.
\end{definition}
Despite the fact that the type distributions converge, the type space $\mathcal{T}_N$ is drastically different from $\mathcal{T}$: in particular, $\mathcal{T}_N \sim \mathcal{M}_N$ and both approximately span a $D$-dimensional space of datasets for large $N$, while $\mathcal{T}$ is only $2$-dimensional, as every dataset is described by $\mu$ and $\theta$. While datasets far away from $\mathcal{T}$, that have distributions unlike the data-generating distribution in any state, become vanishingly unlikely as $N$ grows large, they are never impossible except in the limit; this is why the continuum model is much easier to work with.

It remains possible to describe an imitation equilibrium and a truth-leaning equilibrium in the finite-data setting. The finite-data model is a special case of the evidence model in \citet{Hart17} and \citet{Rappoport22}. The former shows that truth-leaning equilibria exist and are unique and receiver-optimal in the finite-type setting, and also that they are always outcome-equivalent to imitation equilibria, although it does not guarantee that the strategies are equivalent. The latter includes an iterative algorithm to compute these equilibria; the number of steps is, however, exponential in $|\mathcal{T}_N|$, and as far as we can tell, there is no obvious way to obtain a significantly more efficient closed-form solution.

We can instead establish that the imitation equilibrium of the continuous-data model gives a perfect approximation to the limit outcome of communication in truth-leaning equilibria of finite-data games as $N$ and $\{\mathcal{G}^j_N\}_{j=1}^J$ converge.\footnote{We state the definition of convergence and the theorem below in terms of $N = 1, 2, \ldots$ rather than an arbitrary sequence of dataset sizes $N_1, N_2, \ldots$ only for the sake of notational brevity. The theorem applies just as well to any sequence of games $\{\mathcal{G}_{N_i}\}_{i=1}^\infty$ of increasing dataset size with uniformly convergent data distributions, since any such sequence is a subsequence of a convergent sequence of games $\{\mathcal{G}_{N}\}_{N=1}^\infty$.} To make the comparison, the notion of an outcome should be extended across type spaces. There is a global data space $[0,1] \times \Delta \mathcal{D}$, invariant to $N$, that contains $\mathcal{T}_1, \ldots$ and $\mathcal{T}$ as long as they all share a space of observations. Recall that $u_{\sigma^*}(t)$, the outcome of the game for type $t$, is their payoff from the best feasible message given equilibrium beliefs. If $t \in [0,1] \times \Delta \mathcal{D}$, it need not also be in the literal type set for the outcome to be well-defined, since we can already infer whether $t$ can feasibly send a message from the subset relation on $[0,1] \times \Delta \mathcal{D}$. The outcome to the hypothetical type can be understood as a thought experiment: ``if the receiver believes we are playing a game with equilibrium $\sigma^*$ or $\sigma^*_N$, and my dataset is $t$, what is the best payoff I can attain, even if $t$ is inconsistent with the receiver's perceived game?" 

\begin{definition}
 A sequence of equilibria $(\sigma_1, \sigma_2, \ldots)$ of games $\{\mathcal{G}^j_{N}(\Theta, \mathcal{D}, \beta_, \{f_j\}_{j=1}^J, G_N\}_{N=1}^\infty$ has outcomes that \emph{converge} to the outcome of an equilibrium $\sigma$ of the limit infinite-data game 
 \newline
 $\mathcal{G}(\Theta, \mathcal{D}, \beta_, \{f_j\}_{j=1}^J, G)$ if the payoffs $u_{\sigma^*_N}(t)$ converge uniformly to $u_{\sigma^*}(t)$ over types in $\mathcal{T}$.
\end{definition}

\begin{theorem} \label{thm:finitelimit}
If $\mathcal{G}$ is the limit game for finite-data games $\mathcal{G}_1, \mathcal{G}_2, \ldots$ with $N = 1, 2, \ldots$ respectively, then the truth-leaning equilibrium outcomes in $\mathcal{G}_1, \mathcal{G}_2, \ldots$ converge to the imitation equilibrium outcome of $\mathcal{G}$.
\end{theorem}
Outcome convergence shows that it’s reasonable to use the limit game to describe the distribution of actions the receiver takes after the sender discloses a large dataset, as well as the mapping from the truth to the receiver’s inferences. At a high level, the proof follows from the convergence of type distributions $\mathcal{T}_N$ to $\mathcal{T}$, and from the separation theorem, which holds as well in truth-leaning equilibria of finite-data games. Appendix \ref{app:limitpf} gives the formal argument and shows that the limit equivalence result partially extends to strategies, in addition to outcomes. 

In addition, outcome convergence shows that previous sections' results on comparative statics and experimental design hold approximately for large finite datasets. When the number of observations is finite, splintering the data always leads to a strict improvement in the receiver's welfare, even when the outcome space already contains $S^*$ and thus distinguishes the states as well as possible. However, in this case, the magnitude of the improvement vanishes and is negligible for large $N$. While merging non-distinguishing outcomes is only sharply optimal in the continuum, the convergence result guarantees us that it remains an actionable recommendation, yielding, in practice, nearly-optimal information to the receiver with minimally cumbersome datasets.

\section{Conclusion}
Inference under selective disclosure depends on an understanding of the underlying evidence and the sender's strategy. We have shown that an optimal strategy for the sender approximates a simple procedure: claim a possibly inflated state, and provide a large-enough body of evidence that supports it by mimicking the expected distribution under it. A receiver's inability to verify whether data were omitted or, indeed, how much data the sender observed leads to a muddling of information that reported datasets convey, which can be partially offset if the evidence is of good quality in the sense that its most informative outcomes distinguish one state from another state well. 

We have treated the extent of data acquired as exogenous, and so there is an open question about how the incentives for a sender to accumulate data to persuade a receiver would interact with an endogenous choice to acquire data at cost. Having more data benefits senders strategically regardless of its informational value, suggesting that data could be systematically over-collected for persuasion purposes precisely when they are cheap and plentiful. Separately, we believe that the need to prove the realized state of the world through an imperfect voluntary disclosure process will shape the incentives of an agent who can affect the state via their own actions, and taking into account the visibility of certain improvements over others can help determine where they target their effort. Finally, there is not much empirical work about how the disclosures people make relate to the evidence they have at their disposal when receivers can observe the evidence itself, rather than a summary statistic. More work in this direction would complement the theoretical analysis here.

\bibliography{bibliography}

\appendix

\section{Construction and uniqueness of the imitation equilibrium}

We will prove that Theorem \ref{thm:imitate} holds in a more general case with potentially state-contingent, rather than state-independent, data-mass distributions. Describe a game in this general setting by $\mathcal{G}(\Theta, \mathcal{D}, \beta_0, \{f_j\}_{j=1}^J, d\{G^j\}_{j=1}^J)$ where $G^j$ describes the distribution of $\mu$ under state $j$. The model we describe in the main text corresponds to the case in which $G^j = G$ for all $j$.

\begin{theorem} \label{thm:imitate2}
Suppose that $g^1, \ldots, g^J$, the densities of $\mu$ under states $\theta_1, \ldots, \theta_J$, respectively, are continuous on $\mathbb{R}$ and supported on $[0,1]$. There exists a unique imitation equilibrium outcome, implemented by a vector-valued burden of proof function $\hat{\boldsymbol{\mu}}(u): [0,\theta_J] \rightarrow \mathbb{R}^J$ with inverse $\hat u_k(\mu)$ such that
\begin{enumerate}
    \item $\hat u_j(\mu)$ is continuous and (weakly) increasing in $\mu$ for all $j$.
    \item $\sigma^*(\mu f_j)$ is supported on $\{\mu' f_{k}: \mu' = \hat \mu_{k}(\hat u_{k}(\mu/ r_j(k))) \text{ and } \theta_{k} \in A_j(\mu)\}$.
\end{enumerate}
\end{theorem}

To outline the argument, we first prove the existence of a imitation equilibrium by construction. Then we prove the separation theorem, which we use to show uniqueness.

Recall that $\hat u_{k}(\mu)$ is the equilibrium payoff to sending the message $\mu f_{k}$.

We construct $\hat u_{k}(\mu)$ that is monotone increasing in $\mu$ -- this implies that it must be almost-everywhere differentiable. Since it is also continuous, it is completely determined by its derivative over the points at which the derivative exists. To avoid confusion, we focus on the left derivative of $\hat u_{k}$, which we denote by $\hat u_{k}^{-}$ and, analogously to the top-down construction of the finite-data equilibrium, we construct the payoff function starting from the top down, starting from the frontier $v = \theta_J$.

Recall that $r_j(k) = \max_{d \in \mathcal{D}} \frac{f_k(d)}{f_j(d)}$ is the ratio of the amount of data necessary to imitate a certain amount of $f_k$ under state $j$ to the amount necessary under state $k$, and
\[
A_j(\mu) = \left\{\theta_{k}: k \in \arg \max_{k>j}\hat u_k\left(\frac{\mu}{r_j(k)}\right)\right\}.
\]
is the set of states that type $\mu f_j$ finds it weakly optimal to target given $\hat{\boldsymbol{\mu}}$. 

The range of $\hat u_k(\mu_k)$ is $[0, \theta_k]$ since no type of higher state ever targets state $\theta_k$, so payoffs to targeting $\theta_k$ cannot exceed $\theta_k$ itself.

Define
\[
S(v) = \{\theta_k: \theta_k > v\}
\]
to be the set of states under which the receiver optimally takes an action that yields the sender a payoff greater than $v$.
Then $\hat \mu_k(v) < \infty$ iff $\theta_k \in S(v)$, and since play is supported on $\{\hat \mu_k(u_k(\mu/ r_j(k))) f_k: \theta_k \in A_j(\mu)\}$ and $\sigma(\hat \mu_k(u) f_k| \hat \mu_k(u) f_k) = 1$, $S(v)$ is exactly the set of states that are targeted by some type under $\sigma$ to obtain a payoff of $v$. 

Given a burden of proof vector $\hat{\boldsymbol{\mu}}(v) = (\hat \mu_k(v))_{\theta_k \in S(v)}$, the associated \emph{frontier} consists of all types that are just able to meet some component of $\hat{\boldsymbol{\mu}}(v)$ with no slack, that is, all types $\tilde \mu_j f_j$ such that
\begin{equation} \label{eq:frontier}
r_j(k) \tilde \mu_j = \hat \mu_k(v) \text{ for some } \theta_k \in S(v), \text{ and } \not \exists \theta_{k'} \in S(v) \text{ s.t. } r_j(k') \tilde \mu_j > \hat \mu_{k'}(v). 
\end{equation}
Given a particular burden of proof function $\hat{\boldsymbol{\mu}}$, the implied frontier for payoff $v$ is $\tilde \mu(v|\hat{\boldsymbol{\mu}}) = (\tilde \mu_1, \ldots, \tilde \mu_{l-1}, \hat \mu_l(v), \ldots, \hat \mu_J(v))$ if $S(v) = \{\theta_l, \ldots, \hat \theta_J\}$ where $\tilde \mu_1, \ldots, \tilde \mu_{l-1}$ satisfy eq. \ref{frontier}. 

Let the set of states under which some type of sender obtains payoff $v$ and finds it weakly optimal to target state $\theta_k$ be
\[
\tau_{\hat{\boldsymbol{\mu}}}^{opt}(\theta_k, v) = \{\theta_j: f_k \in A_j(\tilde \mu_j(v|\hat{\boldsymbol{\mu}}))
\]
and let the set of states such that some type of sender obtains payoff $v$ by targeting a state $\theta_k$ with strictly positive probability under $\sigma$ be
\[
\tau_{\hat{\boldsymbol{\mu}}}^{supp}(\theta_k, v) = \{\theta_j: \hat \mu_k(v) f_k \in \text{supp } \sigma(\cdot | \mu f_j) \text{ for some } \mu\}.
\]
Of course, $\tau_{\hat{\boldsymbol{\mu}}}^{supp}(\theta_k, v) \subseteq \tau_{\hat{\boldsymbol{\mu}}}^{opt}(\theta_k, v)$.

For convenience of notation, we extend the definitions of these set-valued functions to any set of inputs (rather than a single input) by letting the function of the set be the union of the function applied to each individual element of the input set: thus for every set $S$ of states, $\tau_{\hat{\boldsymbol{\mu}}}^{opt}(S, v) = \bigcup_{\theta_k \in S} \tau^{opt}_{\hat{\boldsymbol{\mu}}}(\theta_k, v)$ and $\tau_{\hat{\boldsymbol{\mu}}}^{supp}(S, v) = \bigcup_{\theta_k \in S} \tau_{\hat{\boldsymbol{\mu}}}^{supp}(\theta_k, v)$, and for every set $\omega \subseteq [0,1]$, we let $A_j(\omega) = \bigcup_{\mu \in \omega} A_j(\mu)$.

Additionally, we define the expectation of the state under the (receiver's) belief that the the sender is a type that receives $v$ under $\hat{\boldsymbol{\mu}}$ and finds it weakly optimal to target a state in $S$ as follows. 
\[
V_{\hat{\boldsymbol{\mu}}}(S, \hat{\boldsymbol{\mu}}(v)) = \dfrac{\sum_{\theta_j \in \tau^{supp}_{\hat{\boldsymbol{\mu}}}(S, v)} \beta_0(\theta_j) \theta_j g^j(\tilde \mu_j[\hat{\boldsymbol{\mu}}(v)]) \frac{d \tilde \mu_j[\hat{\boldsymbol{\mu}}(v)]}{dv}}{\sum_{\theta_j \in \tau^{supp}_{\hat{\boldsymbol{\mu}}}(S, v)} \beta_0(\theta_j) g^j(\tilde \mu_j[\hat{\boldsymbol{\mu}}(v)]) \frac{d \tilde \mu_j[\hat{\boldsymbol{\mu}}(v)]}{dv}}.
\]

In contrast, the expectation of the state under the receiver's true belief over $\theta$ conditional on knowing that the sender has sent some message that yields payoff $v$ and targets a state in $S$ is
\begin{equation} \label{eq:pool_value_true}
W_{\hat{\boldsymbol{\mu}}}(S, v|\sigma) = \dfrac{\sum_{\theta_j \in \tau^{opt}_{\hat{\boldsymbol{\mu}}}(S, v)} \beta_0(\theta_j) \theta_j g^j(\tilde \mu_j[\hat{\boldsymbol{\mu}}(v)]) \frac{d \tilde \mu_j[\hat{\boldsymbol{\mu}}(v)]}{dv} \sigma(\{\hat \mu_k f_k\}_{\theta_k \in S}|\tilde \mu_j[\hat{\boldsymbol{\mu}}(v)] f_j)}{\sum_{\theta_j \in \tau^{opt}_{\hat{\boldsymbol{\mu}}}(S, v)} \beta_0(\theta_j) g^j(\tilde \mu_j[\hat{\boldsymbol{\mu}}(v)]) \frac{d \tilde \mu_j[\hat{\boldsymbol{\mu}}(v)]}{dv} \sigma(\{\hat \mu_k f_k\}_{\theta_k \in S}|\tilde \mu_j[\hat{\boldsymbol{\mu}}(v)] f_j)} = v.
\end{equation}

For any partial strategy $\hat \sigma$ that gives mixing probabilities between the messages $\hat \mu_{k_i}(v) \mathbf{f}$, the payoff $W_{\hat{\boldsymbol{\mu}}}(S, v|\hat \sigma(v)) $ is always weakly greater than $V_{\hat{\boldsymbol{\mu}}}(S,v)$.
The two are equal exactly when all types obtaining payoff $v$ that find it weakly optimal to target a state in $M$ do so with probability $1$. 

Fix a frontier $\hat{\mu}(v)$, where $\theta_{l-1} < v \le \theta_l$. It will be useful to define an undirected graph $H(v)$ on $S(v)$ by adding an edge between $\theta_k$ and $\theta_{k'}$ if and only if $\tau^{opt}_{\hat{\boldsymbol{\mu}}}(\theta_k, v) \bigcap \tau^{opt}_{\hat{\boldsymbol{\mu}}}(\theta_{k'}, v) \neq \emptyset$, that is, if there is some type that finds it optimal to target either state $\theta_k$ or state $\theta_{k'}$, and is indifferent between the two. Let $C$ be the collection of connected components of $H(v)$.

We use the following algorithm to partition $S(v)$ at a given frontier $\hat{\boldsymbol{\mu}}(v)$.

\textbf{Algorithm:} This algorithm calculates the payoffs to targeting a state in $S(v)$ at frontier $\hat{\boldsymbol{\mu}}(v)$ when all types that do not obtain higher payoffs than $v$ and who can target some $\hat \mu_k(v) f_k, \theta_k \in S(v)$ target the highest-payoff of these messages among those that they can, and assigns states $\theta_k$ to the same partition element if, across them, $\hat \mu_k(v) f_k$ must result in the same payoff, and for $\alpha$ close to $1$, $\alpha \hat \mu_k(v) f_k$ must also result in the same payoff, so that for states under which types at the frontier are indifferent between such messages, they remain so for nearby frontiers.

First, note that if $\sigma$ is such that, when there is a collection of states $\Sigma \subseteq S(v)$ such that, over an interval of payoffs, there always exists between any 2 states in $\Sigma$ a path of other states in $\Sigma$ such that there are types that mix with interior probability between any two successive states, then for all $\theta_k, \theta_{k'} \in \Sigma$,
\begin{equation} \label{eq:pool_burden}
\frac{r_j(k)}{r_j(k')} = \frac{\hat \mu_k(u)}{\hat \mu_{k'}(u)} = \dfrac{\frac{d \hat \mu_k(u)}{du}}{\frac{d \hat \mu_{k'}(u)}{du}} \ \ \left (= \dfrac{\frac{d\hat u_{k'}(\hat \mu_{k'}(u))}{d \mu}}{\frac{d \hat u_{k}(\hat \mu_k(u))}{d \mu}} \right)
\end{equation}
for all $u$ in the interval of payoffs and for all $j$ that target some state in $\Sigma$ at the frontier $\hat{\boldsymbol{\mu}}(u)$.

We define
\[
    \Delta_n(\Sigma, \hat \alpha) = \frac{d^n}{d \alpha^n} \dfrac{\sum_{\theta_j \in \tau^{supp}_{\hat{\boldsymbol{\mu}}}(\Sigma, v)} \beta_0(\theta_j) \theta_j g\left(\alpha \tilde \mu_j[\hat{\boldsymbol{\mu}}(v)]\right) \tilde \mu_j[\hat{\boldsymbol{\mu}}(v)]}{\sum_{\theta_j \in \tau^{supp}_{\hat{\boldsymbol{\mu}}}(\Sigma, v)} \beta_0(\theta_j) g\left(\alpha \tilde \mu_j[\hat{\boldsymbol{\mu}}(v)]\right) \tilde \mu_j[\hat{\boldsymbol{\mu}}(v)]} \Bigg |_{\alpha = \hat \alpha}.
\]
This is equal to the $n$th derivative of the payoff to the set of senders in states that target a state in $\Sigma$ with positive probability at frontier $\hat{\boldsymbol{\mu}}(v)$, that have an amount $\hat \alpha \tilde \mu_j[\hat{\boldsymbol{\mu}}(v)]$ of data, when we assume that eq. \ref{eq:pool_burden} holds over $\Sigma$.

Start with a collection of assigned partition elements, $\mathcal{A}_0 = \emptyset$, and a collection of sets of unassigned states, $\mathcal{C}_0 = C$. Given $\mathcal{A}_n$ and $\mathcal{C}_n$, initialize $\mathcal{A}_{n+1} = \mathcal{C}_{n+1} = \emptyset$, and, taking each set $S \in \mathcal{C}_n$ sequentially, proceed as follows:
\begin{enumerate}
    \item Take all subsets $\Sigma \subseteq S$ and calculate $\Delta_0(\Sigma, 1)$. Tiebreak any with the same value by $\Delta_1(\Sigma, 1), \Delta_2(\Sigma, 1), \ldots$, successively, and take the largest subset $\Sigma$ that is maximal. Label it with $\tau_{\hat{\boldsymbol{\mu}}}^{supp}(\Sigma, v)$, and add it to $\mathcal{A}_{n+1}$.
    
    Note that this implies that $\frac{du_k(\hat \mu_k(v))}{d \mu_k} = \frac{\Delta_1(\Sigma, 1)}{\hat \mu_k(v)}$ when equation \ref{eq:pool_burden} holds for $\theta_k, \theta_{k'} \in \Sigma$ over $[v-\epsilon, v]$, $\epsilon > 0$.
    \item Take $S \setminus \Sigma$, and let $C(S)$ be the collection of connected components of the graph on $S$ constructed analogously to $H(v)$. Add $C(S)$ to $C_{n+1}$ (i.e. augment $C_{n+1}$ as the union of itself and $C(S)$).
    \item Repeat on $\mathcal{A}_{n+1}$ and $\mathcal{C}_{n+1}$ until $\mathcal{C}_{n+1} = \emptyset$.
\end{enumerate}
Putatively, if senders of types $\alpha \tilde \mu_j[\hat{\boldsymbol{\mu}}(v)] f_j$ for some $\theta_j \in \tau_{\hat{\boldsymbol{\mu}}}^{supp}(v)$ pooled with each other, then payoffs are equal to
\[
   \hat u_k(\alpha \hat \mu_k(v)) = v_{\Sigma}(\alpha, \hat{\boldsymbol{\mu}}(v)) \equiv \dfrac{\sum_{\theta_j \in \tau^{supp}_{\hat{\boldsymbol{\mu}}}(\Sigma, v)} \beta_0(\theta_j) \theta_j g\left(\alpha \tilde \mu_j[\hat{\boldsymbol{\mu}}(v)]\right) \tilde \mu_j[\hat{\boldsymbol{\mu}}(v)]}{\sum_{\theta_j \in \tau^{supp}_{\hat{\boldsymbol{\mu}}}(\Sigma, v)} \beta_0(\theta_j) g\left(\alpha \tilde \mu_j[\hat{\boldsymbol{\mu}}(v)]\right) \tilde \mu_j[\hat{\boldsymbol{\mu}}(v)]} \Bigg |_{\alpha = \hat \alpha},
\]
which is continuous in $\alpha$ because $g$ is continuous everywhere in $(0,1]$.\footnote{It is important that $g(1) = 0$, since this ensures that $g^j$ is continuous at $\mu = 1$.} The burden-of-proof function for $\underline v \le v$ is then given by 
\[
\mu_{\Sigma}^{put}(\underline v) \equiv \{v_{\Sigma}^{-1}(\underline v, \hat{\boldsymbol{\mu}}(v)) \hat \mu_k(v) f_k\}_{\theta_k \in S(v)},
\] 
where $v_{\Sigma}^{-1}(\underline v, \hat{\boldsymbol{\mu}}(v))$ is the inverse of  $v_{\Sigma}(\cdot, \hat{\boldsymbol{\mu}}(v))$.

The reason that a partition element is a subset of targetable states in which all messages must achieve the same payoff at the  is that, since $\Sigma$ is a maximal highest-value subset over those that do not already have a higher value, it is either partitionable into smaller subsets, each of which also achieves the same value, or not; but in either case, in each minimal subset that achieves the maximal value, there is a path of messages between any two messages in the subset such that, in the targeting strategy, some type mixes with strictly positive probability between any two adjoining messages. The reason for this is that, for any smaller subset $\Sigma' \subset \hat \Sigma$, we have that $V_{\hat{\boldsymbol{\mu}}}(\Sigma', \hat{\boldsymbol{\mu}}(v)) < V_{\hat{\boldsymbol{\mu}}}(\hat \Sigma, \hat{\boldsymbol{\mu}}(v))$ if $\Sigma$ is a minimal subset that achieves the maximal value. Since the expectation of the state conditional on knowing the message played is in $\hat \Sigma$ is at least $V_{\hat{\boldsymbol{\mu}}}(\hat \Sigma, \hat{\boldsymbol{\mu}}(v))$, there must be some message that yields payoff at least $V_{\hat{\boldsymbol{\mu}}}(\Sigma, \hat{\boldsymbol{\mu}}(v))$. But since there is no message, and indeed no proper subset of messages in $\hat \Sigma$ that achieve payoff $V_{\hat{\boldsymbol{\mu}}}(\Sigma, \hat{\boldsymbol{\mu}}(v))$ if all types that can play one of them do, it must be that for any subset, there is a type that can play some message in the subset but plays a message outside the subset with positive probability.

The reason the same holds true in frontiers to the left of $\hat{\boldsymbol{\mu}}(v)$ is that, if $\Delta_0(\Sigma, 1)$ is uniquely maximal, then $\Delta_0(\Sigma, \alpha)$ is still greater than $\Delta_0(\Sigma', 1)$ for any $\Sigma'$ and $\alpha$ sufficiently close to 1. So, in any state under which senders target a state in $\Sigma$ at $\hat{\boldsymbol{\mu}}(v)$, it remains optimal for them to do so for $\alpha$ close to 1, assuming the putative payoffs above. In addition, the putative payoffs are feasible, because every subset of $\Sigma$ has lower value. If tiebroken by $\Delta_1, \Delta_2,$ and so on, then although $\Delta_0(\Sigma, 1)$ is not uniquely maximal, $\Sigma$ does maximize $\Delta(\cdot, 1)$ immediately to the left of $\hat{\boldsymbol{\mu}}(v)$. 

We will use the partition constructed by the algorithm to construct the equilibrium in chunks. For consistency, we want the following condition:

\textbf{Condition 1.} The value of each partition element constructed using the algorithm is the same, and is equal to $v$.

Under this condition, there is a partial strategy $\hat \sigma$ on each partition element such that $W_{\hat{\boldsymbol{\mu}}}(\theta_k,v|\hat \sigma)= v$ for all states $\theta_k$ in the partition element, and furthermore, there is no partial strategy on a subset of messages in that partition element such that all messages in the subset result in the same payoff that is greater than $v$.

If Condition 1 holds at $\hat{\boldsymbol{\mu}}(v)$ and $\Sigma$ is the partition constructed using the algorithm at $\hat{\boldsymbol{\mu}}(v)$, then there exists some $\epsilon > 0$ such that, for all $\underline{v} \in [v - \epsilon, v]$, Condition 1 holds for the frontier $\{v_{\Sigma}^{-1}(\underline v, \hat{\boldsymbol{\mu}}(v)) \hat \mu_k(v) f_k\}_{\theta_k \in S(v)}$. To show this, observe the following claim, which follows directly from statement of the condition and from continuity of $v_\Sigma(\alpha, \hat{\boldsymbol{\mu}}(v))$:
\begin{claim} \label{claim:cond_1}
Let the set of types that target a state in $\Sigma$ and achieve a payoff of $\underline v$ under $\mu_{\Sigma}^{put}$ be $\tau_{\Sigma}^{put}(\underline v)$.

If Condition 1 holds at $\hat{\boldsymbol{\mu}}(v)$, then if there exists no $v' \in (\underline v, v]$ such that either 
\begin{enumerate}
    \item There is a type $t \in \tau_{\Sigma}^{put}(v')$ such that $t$ can imitate a higher-value state, i.e. there exists partition element such that $\Sigma'$ $v_{\Sigma'}^{-1}(v'', \hat{\boldsymbol{\mu}}(v)) \hat \mu_k(v) f_k \tilde \subseteq t$ for some $v'' > v'$
    \item There is a partition element $\Sigma$ with a subset $\Sigma' \subseteq \Sigma$ such that $v_{\Sigma'}(v_{\Sigma}^{-1}(v', \hat{\boldsymbol{\mu}}(v)), \hat{\boldsymbol{\mu}}(v)) > v'$,
\end{enumerate}
then Condition 1 continues to hold at $\hat{\underline{v}}$.
\end{claim}
 Note that, because for any partition element $\Sigma' \neq \Sigma$ either $v_{\Sigma'}^{-1}(v, \hat{\boldsymbol{\mu}}(v)) \hat \mu_k(v) f_k \tilde {\not \subseteq} t$, or $\Delta_n(\Sigma', 1) < \Delta_n(\Sigma, 1)$ for some $n$ such that $\Delta_i(\Sigma', 1) = \Delta_i(\Sigma, 1)$ for all $i < n$, the continuity of $v_{\Sigma}(\alpha, \hat{\boldsymbol{\mu}}(v))$ implies that for $\underline v$ close to $v$ (1) cannot not hold. Again by continuity, (2) cannot hold for $\underline v$ close to $v$ because for all $\Sigma' \subseteq \Sigma$, $v_{\Sigma'}(v_{\Sigma}^{-1}(v, \hat{\boldsymbol{\mu}}(v)), \hat{\boldsymbol{\mu}}(v)) \le v$ and $\Delta_n(\Sigma', 1) < \Delta_n(\Sigma, 1)$ for some $n$ such that $\Delta_i(\Sigma', 1) = \Delta_i(\Sigma, 1)$ for all $i < n$.

We will use this to construct the equilibrium in segments over which Condition 1 holds, and re-construct partitions using the algorithm in at most countably many points at which either (1) or (2) holds. For every reasonable example we can think of, the number of such points (and thus steps in the construction) is not just countable, but finite. 

Now we turn to constructing larger pooling sets when there is a positive-measure set of types that can achieve the frontier payoff. Given that types support their play on $\{\hat \mu_k(u_k(\mu/ r_j(k))) f_k: \theta_k \in A_j(\mu)\}$, and $\hat u_k(\mu_k)$ is increasing, all types capable of sending a message in $\{\hat{\mu}_j(v) \mathbf{f}_j\}_{j=1}^J$ achieve a payoff of at least $v$. We define the set of types that are incapable of sending a message in $\{\hat{\mu}_j(v) \mathbf{f}_j\}_{j=1}^J$, but capable of sending a message in set $M$, as $T(v, M)$. We will denote the payoff to the sender of the receiver knowing they are one of a set of types that has positive probability measure under the receiver's prior as $U(T)$, and in particular,
\[
U(T(v,M)) = \frac{\sum_{j=1}^J \beta_0(\theta_j) \theta_j \max(\max_{k\ge l}(G^j(\frac{\hat \mu_k(v)}{r_j(k)})) - \min \{G^j(\mu): \exists m \in M \text{ s.t. } m \tilde \subseteq \mu f_j\}, 0)}{\sum_{j=1}^J \beta_0(\theta_j) \max(\max_{k\ge l}(G^j(\frac{\hat \mu_k(v)}{r_j(k)})) - \min \{G^j(\mu): \exists m \in M \text{ s.t. } m \tilde \subseteq \mu f_j\}, 0)}.
\]
Note that $\sup_{M} U(T(v,M)) \ge v$, because $\lim_{\alpha \rightarrow 1} U(T(v, \alpha \hat{\boldsymbol{\mu}})) = v$. If there is a positive-measure type set $T(v,M)$ that achieves the value $\sup_M U(T(v, M))$, then take the largest such set and call it $\hat T_{\hat{\boldsymbol{\mu}}}^{max}(v)$. 
Then the following hold:
\begin{enumerate}
    \item If there exists a set $T(v,M)$ that achieves the value $\sup_M U(T(v,M))$, then there is a unique largest set that does so, and so $\hat T_{\hat{\boldsymbol{\mu}}}^{max}(v)$ is well-defined.
    \item Whenever $\hat T_{\hat{\boldsymbol{\mu}}}^{max}(v)$ exists, there exist $\mu_l, \ldots, \mu_J$ such that $\hat T^{\max}_{\hat{\boldsymbol{\mu}}}(v) = T(v, \{\mu_l f_l, \ldots, \mu_J f_J\})$.
    \item Whenever $\hat T_{\hat{\boldsymbol{\mu}}}^{max}(v)$ exists, there exists a partial strategy $\hat \sigma: \hat T^{\max}_{\hat{\boldsymbol{\mu}}}(v) \rightarrow M = \{\mu_l f_l, \ldots, \mu_J f_J\}$ such that the payoff to any message $m \in M$ given that senders in $\hat T^{\max}_{\hat{\boldsymbol{\mu}}}(v)$ play according to $\hat \sigma$ is $\hat U_{\hat{\boldsymbol{\mu}}}(v)$.
\end{enumerate}
The first point follows from the fact that, unless the union of two such sets yields payoff at least $\hat U_{\hat{\boldsymbol{\mu}}}(v)$, then their intersection -- which corresponds to the pool of types implemented by a different message set -- yields strictly greater payoff. To see the 2nd point, simply take $\mu_k$ to be the minimum amount of data distributed $f_k$ such that the dataset still contains a message in $M$, for each $k \ge l$, and note that the resulting set of types is a subset of $T(v,M)$ that has a smaller mass of types $\theta_j$, $j <l$ but the same mass of types $\theta_k$, $k \ge l$. Since $U(T(v,M)) \ge v \ge \theta_{l-1}$, this can only improve the payoff to the pool. The last point comes from the fact that, if $\hat T^{\max}_{\hat{\boldsymbol{\mu}}}(v)$ is a maximum-payoff pool, then for each subset $S \subseteq M$, the payoff to the pool implemented by $S$ is no greater than $U(\hat T^{\max}_{\hat{\boldsymbol{\mu}}}(v))$, which is sufficient to ensure that $\hat \sigma$ exists. In addition, $U(T(v,M))$ is absolutely continuous with respect to every component of $\hat{\boldsymbol{\mu}}(v)$ and each $\mu_k$.

\begin{lemma}
If $\hat T^{\max}_{\hat{\boldsymbol{\mu}}}(v)$ exists, then Condition 1 is satisfied by the burden of proof vector $M = \{\mu_l f_l, \ldots, \mu_J f_J\}$ such that $\hat T^{\max}_{\hat{\boldsymbol{\mu}}}(v) = T(v,M)$.
\end{lemma}
\begin{proof}
Suppose not; then one of two cases is true:
\begin{enumerate}
    \item \underline{There is a collection of states $\Sigma \subset S(v)$ such that $V_{\hat{\boldsymbol{\mu}}}(\Sigma, M) > v$}.
    
    Then, since $V_{\hat{\boldsymbol{\mu}}}(\Sigma, \alpha (\mu_k f_k)_{k=l}^J)$ is continuous in $\alpha$, there is $\underline \alpha < 1$ such that $V_{\hat{\boldsymbol{\mu}}}(\Sigma, \alpha (\mu_k f_k)_{k=l}^J) > v$ for all $\alpha \in [\underline \alpha, 1]$. Consider an alternative type set, $T(M_{\underline \alpha, \Sigma}, v)$ where $M_{\underline \alpha, \Sigma}$ includes the messages $\mu_k f_k$ for $\theta_k \in S(v) \setminus \Sigma$, and the messages $\underline \alpha \mu_k f_k$ for $\theta_k \in \Sigma$. 
    
    For $\underline \alpha$ small enough, the set of types in $T(M_{\underline \alpha, \Sigma}, v) \setminus T(M, v)$ includes exactly those in frontiers $(\alpha M)_{\alpha = \underline \alpha}^1$ that find it weakly optimal to target a state in $\Sigma$. So, the expectation of the state given that the sender's type is in $T(M_{\underline \alpha, \Sigma}, v) \setminus T(M, v)$ exceeds $v$, and so $T(M_{\underline \alpha, \Sigma}, v)$ is higher-payoff than $T(M, v)$, contradicting that $T(M,v) = \hat T^{\max}_{\hat{\boldsymbol{\mu}}}(v)$.
    
    \item \underline{There is a element of the partition, $\Sigma' \subset S(v)$, such that $V_{\hat{\boldsymbol{\mu}}}(\Sigma', M) > v$}.
    
    Then WLOG let $\Sigma'$ be the lowest-value element of the partition. Similarly to the above, since $V_{\hat{\boldsymbol{\mu}}}(\Sigma, \alpha (\mu_k f_k)_{k=l}^J)$ is continuous in $\alpha$, there is $\bar \alpha > 1$ such that $V_{\hat{\boldsymbol{\mu}}}(\Sigma, \alpha (\mu_k f_k)_{k=l}^J) < v$ for all $\alpha \in [1, \bar \alpha]$. Consider an alternative type set, $T(M_{\bar \alpha, \Sigma'}, v)$ where $M_{\bar \alpha, \Sigma'}$ includes the messages $\mu_k f_k$ for $\theta_k \in S(v) \setminus \Sigma'$, and the messages $\bar \alpha \mu_k f_k$ for $\theta_k \in \Sigma'$. 
    
    For $\bar \alpha$ small enough, the set of types in $T(M,v) \setminus T(M_{\bar \alpha}, \Sigma')$ includes exactly those in frontiers $(\alpha M)_{\alpha = 1}^{\bar \alpha}$ that find it weakly optimal to target a state in $\Sigma$. Then the expectation of the state given that the sender's type is in $T(M,v) \setminus T(M_{\bar \alpha}, \Sigma')$ is less than $v$, so the expectation given that the type is in $T(M_{\bar \alpha}, \Sigma')$ exceeds $v$, contradicting that $T(M,v) = \hat T^{\max}_{\hat{\boldsymbol{\mu}}}(v)$.
\end{enumerate}
Since neither case is possible, $M$, taken as the payoff frontier corresponding to $v$, must satisfy Condition 1.
\end{proof}

The iterative algorithm to construct the equilibrium of \ref{thm:imitate} starts from the highest-potential-payoff senders and creates payoff frontiers that satisfy Condition 1. It proceeds as follows:

\begin{enumerate}
    \item Start with $l = J$ and $\hat \mu_J(\theta_J) = 1$.
    \item For each $l$, construct frontiers $\hat \mu_k(v)$ as follows:
    \begin{enumerate}
        \item Start at $v = \theta_l$ and burden-of-proof vector $\hat{\boldsymbol{\mu}}(\theta_l)$, as constructed from the previous step. For all $v > \theta_l$, let $\hat{\boldsymbol{\mu}}(v)$ be as already constructed. Define 
        \[
        \check \mu_l(\theta_l) = \max \{\mu: \exists j < l \text{ s.t. } \tilde \mu_j[\hat{\boldsymbol{\mu}}(\theta_l)] \ge \mu\},
        \]
        and rewrite $\hat{\boldsymbol{\mu}}(\theta_l) = (\check \mu_l(\theta_l), \hat \mu_{l+1}(\theta_l), \ldots, \hat \mu_J(\theta_l))$.
        Proceed as below to rewrite $\hat{\boldsymbol{\mu}}(v)$ for $v < \theta_l$:
        \item Fix $S = \{\theta_k\}_{k=l}^J$. Given the frontier $\hat{\boldsymbol{\mu}}(v)$, check if $\hat T_{\hat{\boldsymbol{\mu}}}^{max}(v)$ exists, and if so,
        find $M = \{\mu_l f_l, \ldots, \mu_J f_J\}$ that implements $\hat T_{\hat{\boldsymbol{\mu}}}^{max}(v)$ and rewrite $\hat{\boldsymbol{\mu}}(v) = M$.
        \item At $\hat{\boldsymbol{\mu}}(v)$, using the algorithm, partition $S$ into subsets of states, and calculate $v_\Sigma(\alpha, \hat{\boldsymbol{\mu}}(v))$ for all $\alpha \in [0,1]$ for each subset. Take the lowest-value frontier, $\hat{\boldsymbol{\mu}}(v')$, under putative payoffs $v_\Sigma(\alpha, \hat{\boldsymbol{\mu}}(v))$ such that the conditions of Claim \ref{claim:cond_1} are satisfied and such that $\hat T_{\hat{\boldsymbol{\mu}}}^{max}(v'')$ does not exist for any $v'' \in (v', v]$, and assign strategies according to Algorithm 2 between $\hat{\boldsymbol{\mu}}(v)$ and the new frontier $\hat{\boldsymbol{\mu}}(v')$.
        \item Set $v = v'$ and set $\hat{\boldsymbol{\mu}}(v')$ as the new frontier, and repeat the above 2 steps until $v' = 0$.
    \end{enumerate}
    \item Repeat the above steps for each $l$ in descending order until $l = 1$, and fix the resulting $\hat{\boldsymbol{\mu}}$.
\end{enumerate}

The existence of an imitation equilibrium, and the monotonicity of $\hat u_k$, follow directly from this construction. Continuity of $\hat u_k$ also follows from this construction. The value of $\hat u_k$ is defined on series of closed intervals on each of which it is continuous -- $v_\Sigma(\alpha, \mu)$ is continuous in $\alpha$, and $\hat u_k(\mu)$ is constant for $\mu f_k \in T(v,M)$. Together, these cover the domain of $u_k$, that is, $[0,1]$, and they overlap only at their endpoints, at which they coincide. 

Next, we prove the separation theorem. It has 2 parts, which we will prove as lemmas. We start by proving that upper pools are improving:

\begin{lemma} \label{lem:pool_payoff}
If $M$ is a collection of messages and $\{\tilde \mu_j(v) f_j\}_{j=1}^J$ is the frontier of types achieving a payoff of at least $v$ under $\sigma^{*}$, where $\theta_i < v \le \theta_{i+1}$, then
\[
\mathbb{E}_{q}[\theta|t \in U(\{\underline \mu_j f_j\}_{j=1}^J) \setminus U(M)] \ge v
\]
whenever $U(\{\tilde \mu_j f_j\}_{j=1}^J) \setminus U(M)$ is nonempty.
\end{lemma}

\begin{proof}[Proof of Lemma]
Denote $T(v,M) = U(\{\tilde \mu_j f_j\}_{j=1}^J) \setminus U(M)$. Let $(\bar \mu_1, \ldots, \bar \mu_i; \bar \mu_{i+1}, \ldots, \bar \mu_J)$ be the minimum masses of data distributed like $f_1, \ldots, f_i; f_{i+1}, \ldots, f_J$, respectively, necessary to send some message in $M$. Then
\[
\mathbb{E}_{q}[\theta|t \in T(v,M)] = \frac{\sum_{j=1}^J \beta_0(\theta_j) \theta_j (G^j(\bar \mu_j) - G^j(\tilde \mu_j))}{\sum_{j=1}^J \beta_0(\theta_j) (G^j(\bar \mu_j) - G^j(\tilde \mu_j))}.
\]

If $(\bar \mu_{i+1}, \ldots, \bar \mu_J) \le (\tilde \mu_{i+1}, \ldots, \tilde \mu_J)$ pointwise, then $T(v,M)$ is empty. Otherwise, let the states $j_1, \ldots, j_A$ be the maximal set such that $(\bar \mu_{j_1}, \ldots, \bar \mu_{j_A}) > (\tilde \mu_{j_1}, \ldots, \tilde \mu_{j_A})$ pointwise. 
Call the set of types that send $\mu' f_{j_a}$ with positive probability under $\sigma^{*}$ by $\tau_{\sigma^{*}}^{supp}(\mu' f_{j_a})$, and let $\theta(t)$ refer to the state corresponding to the distribution of dataset $t$. Denote by $\hat \sigma_v$ the partial strategy, restricting to types in $T(v,M)$, where those types play as they do in $\sigma^*$, and assume that the receiver knows the sender is in $T(v,M)$ and playing according to this strategy.

Let $\phi_{\sigma^*}$ be a joint density over types and messages induced by $\sigma^*$, so that for type $t = \mu f_j$ and message $m = \tilde \mu f_{\tilde j_a}$, we can define
\[
\phi(t, m) = g^j(\mu r_{j}(\tilde j)) \sigma^*(m| t) \beta_0(\theta_j) r_{\theta_j}(j')
\]
to be the density on the event that the sender is type $t$ and plays message $m$, when $t$ plays $m$ with positive probability.
In the case when payoffs under $u_{\sigma^*}$ are strictly increasing at $\mu' f_{j_a}$, each sender who plays $\mu' f_{j_a}$ is randomizing between at most a finite number of messages in their mixed strategy, one corresponding to each state that is weakly optimal for them to imitate. Thus, they play each message in the support of their strategy with strictly positive probability, rather than randomizing with some density over a continuum of messages; $\phi$ therefore fully captures the distribution of play for senders playing $\mu' f_{j_a}$. 

When payoffs are strictly increasing at $\mu' f_{j_a}$, we know that for every $\mu' f_{j_a}$ that is in $T(v,M)$ and is on-path in $\sigma^*$, the receiver's inference when they know the sender's type is in $T(v,M)$ in addition to knowing they played message $\mu' f_{j_a}$ is weakly better than if they only know $\mu'f_{j_a}$ was the message played. Formally,
\begin{equation}
\begin{split}
    \mathbb{E}_{q}[\theta|\mu' f_{j_a}] &= \frac{\sum_{t \in \tau_{\sigma^*}^{supp}(\mu' f_{j_a}) \bigcap T(v,M)} \theta(t) \phi(t, \mu' f_{j_a}))}{\sum_{t \in \tau_{\sigma^*}^{supp}(\mu' f_{j_a}) \bigcap T(v,M)} \phi(t, \mu' f_{j_a})}\\
    &\ge \frac{\sum_{t \in \tau_{\sigma^*}^{supp}(\mu' f_{j_a})} \theta(t) \phi(t, \mu' f_{j_a})}{\sum_{t \in \tau_{\sigma^*}^{supp}(\mu' f_{j_a}))} \phi(t, \mu' f_{j_a})}\\
    &\ge v
\end{split}
\end{equation}
 where the first inequality comes from the fact that $\theta_{j_a} \ge v > \theta(t)$ whenever $\theta(t) \neq \theta_{j_a}$, and $\mu' f_{j_a} \in T(v,M)$ only if all types that play it under $\sigma^*$ are also in $T(v,M)$.

Since, of course, payoffs under $u_{\sigma^*}$ may not be strictly increasing at every $\mu' f_{j_a}$ in $T(v,M)$, we have to separately consider the case in which they are constant, i.e. the case where there are positive-measure pools $T$ of senders achieving the same payoff $v' > v$ under $\sigma^*$ with $T \bigcap T(v,M)$ nonempty. Then let $M'$ be the set of messages that implements the pool, and 
\[
\mathbb{E}_{\hat \sigma_v}[\theta|m \in M'] = \mathbb{E}_{\hat \sigma_v}[\theta|t \in T \bigcap T(v,M)].
\]
The value of $T \setminus T(v,M)$ is equal to the value of $T \bigcap U(M)$, which is no more than $v'$ since $T = \hat T^{max}_{\hat{\boldsymbol{\mu}}}(v')$, and so it contains no subsets of higher value. Therefore, $\mathbb{E}_{\hat \sigma_v}[v(\theta)|t \in T \bigcap T(v,M)] \ge v' \ge v$.

Then, taking the total expectation over both cases, the expectation of $\theta$ given that the sender's type is in $T(v,M)$ is a weighted average of $\mathbb{E}_{\hat \sigma_v}[\theta|\mu' f_{j_a}]$ over on-path messages $\mu' f_{j_a}$ in $T(v,M)$ in which the payoff is strictly decreasing; and the value over positive-measure sets of equal payoff. We have shown that each component is no less than $v$, and so the weighted average is also at least $v$. \qed
\end{proof}

Next we prove that the imitation equilibrium we construct has worsening lower pools. This is relatively simple.
\begin{lemma} \label{lem:worsening_lower}
If $M$ is a collection of messages and $\{\tilde \mu_j(v) f_j\}_{j=1}^J$ is the frontier of types achieving a payoff of at least $v$ under $\sigma^{*}$, where $\theta_i < v \le \theta_{i+1}$, then
\[
\mathbb{E}_{q}[\theta|t \in U(M) \setminus U(\{\tilde \mu_j f_j\}_{j=1}^J)] < v
\]
whenever $U(M) \setminus U(\{\tilde \mu_j f_j\}_{j=1}^J)$ is nonempty.
\end{lemma}

\begin{proof}
    If there was a payoff frontier $\hat{\boldsymbol{\mu}}(v)$ that had a nonempty, weakly improving lower pool lower-bounded by messages $M$, then there is a frontier $\hat{\boldsymbol{\mu}}(w) \neq M$ for some $w \ge v$ such that
    \[
    u_{pool}(U(M) \setminus U(\hat{\boldsymbol{\mu}}(w))) = w.
    \]
    The construction algorithm rules this out, because if indeed the payoff frontiers above $\hat{\boldsymbol{\mu}}(w)$ are correctly constructed, then it would next set $\hat{\boldsymbol{\mu}}(w) = M$.
\end{proof}

Finally, we show that the constructed equilibrium outcome is the only imitation equilibrium outcome, and thus that the imitation equilibrium outcome is unique.

\begin{proof}
    Let the constructed equilibrium be $\sigma^*$, and let $\sigma$ be an alternative equilibrium, with a different outcome. We aim to show that $\sigma^*$ does not have improving upper pools, and therefore cannot be an imitation equilibrium.

    To see this, let $M$ represent the frontier of messages that are used to achieve payoff $v$ in $\sigma$. Worsening lower pools under $\sigma^*$ imply that $u_{pool}(U(M) \setminus U(\hat{\boldsymbol{\mu}}(v))) \le v$, implying that $M$ has a worsening upper pool. Since $M$ is a payoff frontier of $\sigma$, the alternative equilibrium $\sigma$ does not have improving upper pools, and is therefore not an imitation equilibrium.
\end{proof}

\section{Inclusive announcement-proofness} \label{app:refinement}
Here we discuss a way in which the truth-leaning equilibrium outcome arises from optimal behavior for the sender. The concept of optimality we use, \emph{inclusive} announcement-proofness, refines PBE by requiring that there is no self-separating set of sender types who could weakly improve their payoffs by announcing a strategy that uses some set of messages differently than they are used in the baseline equilibrium.
 
\begin{definition}
    Given an outcome $u_{\sigma^*}$, a set of types $T$ has a \emph{credible inclusive announcement} that they will play a strategy $\hat \sigma_M$ supported over message set $M$ for payoff $v$ if
    \begin{itemize}
        \item $\hat \sigma_M: M \times T \rightarrow \mathbb{R}$ is such that $\sum_{t \in T} \hat \sigma_{M}(m|t) = 1$ for all $m \in M$, $\sum_{m \in M} \hat \sigma_{M}(m|t) = 1$ for all $t \in T$, and $\mathbb{E}_{\beta_{\hat \sigma_M}(\cdot|m)}[\theta] = v$ for all $m \in M$.
        \item $T = \{t \in U(M): u_{\sigma^*}(t)\}$, and there is some $t \in T$ with $u_\sigma(t) < v$.
    \end{itemize}
\end{definition}
A very closely-related notion, that we take the name from, is the idea of a credible announcement, from \citet{Matthews91}. There is, however, a subtle difference, which is that in a credible announcement, $T = \{t \in U(M): u_{\sigma^*}(t) \le v\} \bigcup S$ where $S \subseteq \{t \in U(M): u_{\sigma^*}(t) = v\}$. Thus what we use is an ``inclusive" notion of a credible announcement in that the set of announcing types must include all who weakly prefer to participate; it is stronger to claim there exists an credible inclusive announcement than that there exists a credible announcement, and correspondingly, inclusive announcement-proofness is weaker than announcement-proofness. In fact, there may exist no announcement-proof equilibrium at all in the game we study, while there always exists exactly one inclusive announcement-proof equilibrium outcome.

\begin{claim}
    In $\mathcal{G}$, the unique inclusive announcement-proof equilibrium outcome is the imitation (equivalently, truth-leaning) equilibrium outcome.
\end{claim}
\begin{proof}
    For any equilibrium $\sigma$ with a different outcome than the imitation-equilibrium outcome $\sigma^*$, there is some $v$ such that the $v$-payoff frontier under $\sigma$ differs from that under $\sigma^*$, and such that some types that achieve a payoff of $v$ or greater under $\sigma^*$ achieve a payoff no more than $v$ under $\sigma$. Lemma \ref{lem:pool_payoff} ensures that when all such types pool, the expected value of the state is at least $v$. Then, from the continuity of $\hat u_j(\mu)$, there exists some $v' < v$ such that when the set of all types that achieve a payoff of at least $v'$ under $\sigma^*$, but a payoff of no more than $v$ under $\sigma$, is pooled, the expected value of the state is exactly $v$. Starting from equilibrium $\sigma$, this set of types has a credible inclusive announcement that yields a payoff of $v$ to each type, and so $\sigma$ is not inclusive announcement-proof.

    On the other hand, any credible announcement relative to baseline equilibrium $\sigma^*$ requires the existence of some $v$ and set of messages $M$ such that there exists a pool of types 
    \[
    T = \{t \in U(M): u_{\sigma^*}(t) \le v\}
    \]
    such that $\mathbb{E}[\theta|t\in T] = v$, with at least one type $t' \in T$ such that $u_{\sigma^*}(t') < v$. Since $T$ contains all types $t \tilde \supset t'$ with $u_{\sigma^*}(t) \le v$, we know $T$ is a set of positive measure. The construction algorithm for $\sigma^*$, however, rules out the presence of any such set $T$, since if all frontiers for payoffs in $(v, \theta_J]$ are correctly constructed, then all types in $T$ must be pooled under $\sigma^*$ and must obtain a payoff of $v$ exactly.
\end{proof}

\begin{claim}
    In any finite-data game $\mathcal{G}_N$, the unique inclusive announcement-proof equilibrium outcome is the truth-leaning equilibrium outcome.
\end{claim}

To show that the truth-leaning equilibrium outcome is inclusive announcement-proof in finite-data games, I construct it, using the algorithm from Rappoport, which I summarize here. In short, the equilibrium is constructed by iteratively choosing a frontier of types such that the set of types ``above" the frontier, in the sense of being able to imitate some frontier type, yields as favorable a belief as possible.

\noindent \textbf{Algorithm (Finite $N$).}
First, define for any type set $T$ the subset of types $T^+(M) = T \bigcap U(M)$ as the set of types in $T$ that are capable of sending some message in message set $M$, and define $u_{pool}(T)$ to be the payoff to the sender if the receiver knows only that their type must be in $T$.
\begin{enumerate}
    \item Let $T_1 = \mathcal{T}_{N}$, and find the set of messages $M_1 \subseteq T_1$ that maximizes the payoff to a pool consisting of the set of senders in $T_1$ who can send at least one message in it:
    \[
    M_1 \in \arg \max_{M \subseteq T_1} u_{pool}(T_1^+(M)).
    \]
    If there are multiple such pools, then we take their union, which is also such a pool. 
    \item For $s = 2$ onwards, restrict the set of types to $T_s = T_{s-1} \setminus T_{s-1}^+(M_{s-1})$, and find (the union of)
    \[
    M_s \in \arg \max_{M \in T_{s}}  u_{pool}(T_s^+(M)). 
    \]
    \item Continue until $T_s \setminus T_s^+(M_s) = \emptyset$. Given each set $M_s$, there always exists a mixed strategy profile $\sigma_{pool}^M$ defined over types in $T_1^+(M)$ such that each message in $M$ yields the same payoff under the receiver's induced beliefs from $\sigma_{pool}^M$.\footnote{Otherwise, the worst possible payoff to particular message in $M$ over all strategy profiles over $M_s$ is better than the best possible payoff to some other message; then there always exists $M \subset M_s$ such that $T_s^+(M) > T_s^+(M_s)$.}
    Define $\sigma^*$ by $\sigma^*(m|t) = \hat \sigma_{pool}^{M_s}(m|t)$ where $M_s$ is the pool containing $m$. 
\end{enumerate}

\begin{proof}[Proof. (Unique credible inclusive announcement-proof outcome)]
By construction, there is no credible inclusive announcement, since such an announcement would constitute a better set of types than the one constructed at some step of the algorithm; this violates the optimality of the pool of types constructed in each step. No other outcome is immune: if $u_{\sigma^*_{alt}} \neq u_{\sigma^*}$, then there exists a $v$ such that the set of pools achieving a payoff greater than $v$ is identical in $u_{\sigma^*_{alt}}$ and $u_{\sigma^*}$, but the pool of types $T$ achieving payoff $v$ under $u_{\sigma^*}$ is a strict superset of that under $u_{\sigma^*_{alt}}$. Then types in $T$ can make a credible inclusive announcement that they will play as they do in $\sigma^*$.
\end{proof}

\section{Imitation, truth-leaning, and optimality} \label{app:truth_leaning_char}

We prove that truth-leaning equilibria and imitation equilibria coincide in $\mathcal{G}$, that the imitation equilibrium outcome is unique, and that it is the optimal outcome of communication under commitment for the receiver.

\begin{claim}
    Every imitation equilibrium of $\mathcal{G}$ is a truth-leaning equilibrium of $\mathcal{G}$.
\end{claim}

\begin{proof}
    We take the 2 perturbations separately. First, perturb the likelihood of honest commitment types by a sequence with $\epsilon^k_{t|t} = \epsilon^k \rightarrow 0$. There exists an equilibrium $u_{\sigma^*_{\epsilon^k}}$ of $\mathcal{G}^{\epsilon^k}$ in which strategies of non-commitment types are identical to the imitation equilibrium strategies in a game $\tilde {\mathcal{G}}^{\epsilon^k}$ under which 
    \[
    q(\mu f_j) = \begin{cases}
        \frac{\beta_0(\theta_j) (g(\mu)-\epsilon^k)}{1-\epsilon^k \sum_i \beta_i (1-G^i(\hat \mu_i(\theta_i))}, \ \ \ \ \ \  \  \mu \ge \hat \mu_j(\theta_j)\\
        \frac{\beta_0(\theta_i) g(\mu)}{1-\epsilon^k \sum_i \beta_i (1-G^i(\hat \mu_j(\theta_i))}, \ \ \ \ \ \ \ \mu < \hat \mu_j(\theta_j).
    \end{cases}
    \]
    Under the metric induced by the L2 norm, the set of equilibrium strategies is compact, and payoffs in $\tilde {\mathcal{G}}^{\epsilon}$ are continuous in $\epsilon$, so the limit point as $k \rightarrow \infty$ of the imitation equilibria of $\tilde {\mathcal{G}}^{\epsilon^k}$ must also be an equilibrium of $\mathcal{G}$. It is easy to verify that it must also satisfy the conditions in \ref{def:cts_truth_leaning_reduction}, so it is the imitation equilibrium of $\mathcal{G}$.

    Now, for fixed $\epsilon_k$, consider in addition the perturbation of payoffs by an additional payoff bump $\nu$ to a truthful report. When $\nu < \min_{j, k} |\theta_k - \theta_j|$, there exists an equilibrium $\sigma^*_{\epsilon^k, \nu}$ that is identical to the equilibrium $u_{\sigma^*_{\epsilon^k}}$ specified above, except for types $\mu f_j$ with $u_{\sigma^*_{\epsilon^k}} \in (\theta_j, \theta_j + \nu)$, who instead play the truth with positive probability. In particular, for a given message $\mu' f_k$ that yields a payoff in  $(\theta_j, \theta_j + \nu)$ and is played by $\mu f_j$ under $\sigma^*_{\epsilon^k}$, the probability that it is played by $\mu f_j$ in the equilibrium of the further-perturbed game is $0$ if the expected state over types playing $\mu' f_k$ for whom the state is not $\theta_j$ is no greater than $\theta_j + \nu$, and otherwise, the probability that $\mu f_j$ plays $\mu f_k$ is exactly such that the payoff to playing $\mu f_k$ is $\theta_j + \nu$, so that $\mu f_j$ is indifferent between playing message $\mu' f_k$ and revealing all their data. As $\nu \rightarrow 0$, the set of affected types shrinks towards a measure-0 set, and so these equilibria converge to $u_{\sigma^*_{\epsilon^k}}$ as $\nu \rightarrow 0$.

    Finally, given the equilibria $\{\sigma^*_{\epsilon^k, \nu^j}\}$ for $\epsilon^k \rightarrow 0$, $\nu^j \rightarrow 0$, diagonalize by taking, for every $k$, some $j_k$ such that $|| \sigma^*_{\epsilon^k, \nu^{j_k}} - \sigma^*_{\epsilon^k}|| < \frac{1}{k}$, and observe that then the sequence of perturbations $(\epsilon_{t|t} = \epsilon^k \forall t, \epsilon_t = \nu^{j_k} \forall t)_{k=1}^\infty$ yields equilibria that converge to $\sigma^*$.
\end{proof}

\begin{claim}
    Every truth-leaning equilibrium in $\mathcal{G}$ is an imitation equilibrium of $\mathcal{G}$.
\end{claim}

\begin{proof}
    If $t$'s dataset is off-path then the receiver plays a best response to the belief $\mathbbm{1}_t$ upon seeing t. This suffices to show that every truth-leaning equilibrium messaging strategy $\sigma$ is a best response to $q_{\sigma}$, as defined by eq. \ref{eq:truth_leaning_beliefs}.

     For part a), note that if a message $m$ is on-path in $\sigma$, then there exists $K_1$ such that for all $k > K_1$, $m$ is on-path in $\sigma^*_{\epsilon_k}$. For every $k$, however, all on-path messages are in $\mathcal{T}$, since if $m$ is on-path and $m \not \in \mathcal{T}$, then there is a type $t = \mu f_j$ with $\theta_j > u_{\sigma^*_{\epsilon_k}}(m)$ that plays $m$, and $t$ itself is not played as a message on path by any non-commitment types. But then $\mathbb{E}_{\beta_{\sigma^*_{\epsilon^k}}(\cdot|t)}[\theta] = \mathbb{E}_{\pi(\cdot|t)}[\theta] \ge \mathbb{E}_{\beta_{\sigma^*_{\epsilon_k}}(\cdot|m)}[\theta]$, leading to a contradiction. Hence, all on-path $m$ must be in $\mathcal{T}$.

    To prove that a truth-leaning equilibrium messaging strategy satisfies c), suppose there is $t$ such that $\mathbb{E}_{\pi(\cdot|t)}[\theta] > \max_{m \tilde \subseteq t} \mathbb{E}_{\beta_{\sigma}(\cdot|m)}[\theta]$ but $\sigma(t|t) < 1$. 
    
.    We will show that there is no sequence of perturbations $\{\epsilon^k_t, \epsilon^k_{t|t}\}_{k=1}^\infty \to 0$ such that equilibria of the associated perturbed games $\mathcal{G}^k$ converge to $\sigma$. Start by supposing for the sake of contradiction that there is. First, we know $t$ must be on path in $\sigma$. If $\sigma^k$ is an equilibrium of game $\mathcal{G}^k$ with $\epsilon^k_t > 0$, there cannot $t' \neq t$ such that $\sigma^k(t|t') > 0$, otherwise $\mathbb{E}_{\beta_{\sigma^k}(\cdot|t)}[\theta] \ge \max_{t' \tilde \subset t} \mathbb{E}_{\beta_{\sigma^k}(\cdot|t')}[\theta]$ and so $\mathbb{E}_{\beta_{\sigma^k}(\cdot|t)}[\theta] + \epsilon_t^k > \max_{t' \tilde \subset t} \mathbb{E}_{\beta_{\sigma^k}(\cdot|t')}[\theta]$ and we would have to have $\sigma^k(t|t) = 1$. Then, likewise, in the limit $\sigma$, we must have $\sigma(t|t') = 0$ for all $t'$. Since $t$ is on-path in $\sigma$, it must be that $\sigma(t|t) \in (0,1)$.
    
    Take a type $t'' \neq t$ such that $\sigma(t''|t) > 0$. We know that there exists $K$ such that for all $k > K$, $\sigma^k(t''|t) > 0$ as well. Then whenever $k > K$, $\mathbb{E}_{\pi(\cdot|t)} + \epsilon^k_t = \mathbb{E}_{\beta_{\sigma^k}(\cdot|t'')}[\theta]$. Because $\sigma^k \rightarrow \sigma$, we have that 
    \[
    \lim_{k \rightarrow \infty} \mathbb{E}_{\beta_{\sigma^k}(\cdot|t'')}[\theta] = \mathbb{E}_{\beta_\sigma(\cdot|t'')}[\theta] = \max_{m \tilde \subseteq t} \mathbb{E}_{\beta_{\sigma}(\cdot|m)}[\theta].
    \]
    But this contradicts that $\mathbb{E}_{\pi(\cdot|t)}[\theta] > \max_{m \tilde \subseteq t} \mathbb{E}_{\beta_{\sigma}(\cdot|m)}[\theta]$ and
    \[
    \lim_{k \rightarrow \infty} \mathbb{E}_{\beta_{\sigma^k}(\cdot|t'')}[\theta] = \lim_{k \rightarrow \infty}\mathbb{E}_{\pi(\cdot|t)} + \epsilon^k_t = \mathbb{E}_{\pi(\cdot|t)}.
    \]

    To show that b) holds, note that for any $k$, if $t$ is on-path and played by some $t' \neq t$, then $\sigma^k(t|t) = 1$. By c), $\mathbb{E}_{\pi(\cdot|t')} \le \mathbb{E}_{\beta_{\sigma^k}(\cdot|t)}$, but if $t$ also plays $t$ and $\mathbb{E}_{\pi(\cdot|t)} < \mathbb{E}_{\beta_{\sigma^k}(\cdot|t)}$, then the receiver cannot Bayesian. On the other hand, if $t$ is on-path and only $t$ plays $t$, then we must have $\mathbb{E}_{\pi(\cdot|t)} = \mathbb{E}_{\beta_{\sigma^k}(\cdot|t)}$. \qed
\end{proof}

Finally, closely following the idea in \citet{Hart17}, we show that the imitation equilibrium outcome is the outcome of the optimal pure-strategy mechanism, that is, the best outcome the receiver can achieve when they can commit to a pure action as a response to the message the sender sends. The revelation principle shows that it suffices to look at direct mechanisms, in which the sender truthfully reports their type and the receiver commits to a deterministic response to the sender's reported type.

A mechanism under which type $t$ elicits the action $a(t)$ is implementable if it satisfies IC:
\begin{equation}
\tag{IC} t \tilde \subseteq t' \ \Rightarrow \ a(t') \ge a(t).
\end{equation}

\begin{claim} \label{claim:opt_mech}
    The imitation equilibrium outcome is the optimal outcome for the receiver under commitment to pure strategies.
\end{claim}

To prove this claim, first define $T_{\mu f_k}$ be the set of types that imitate $\mu f_k$ under $\sigma^*$, including $\mu f_k$ itself. We start with a lemma.

\begin{lemma} \label{lem:finite_im}
    There always exists an imitation equilibrium $\sigma^*$ such that $T_{\mu f_k}$ is finite for every $\mu f_k \in \mathcal{T}$.
\end{lemma}

\begin{proof}[Proof of Lemma \ref{lem:finite_im}]
    First, for any imitation equilibrium, if $\{t: u_{\sigma^*}(t) = u_{\sigma^*}(\mu f_k)\}$ is a measure-0 set, since then it is necessarily true that at most one type under each state lies in the same payoff frontier as $\mu f_k$ under $\sigma^*$, and thus at most one type under each state imitates it.

    Now consider the case in which there is a positive-measure set of senders who achieve the payoff $u^* = u_{\sigma^*}(\mu f_k)$, where we have $\theta_l \le u_{\sigma^*}(\mu f_k) < \theta_{l+1}$. We know that there exists a way to divide the types by which state they imitate, and with what probability, given by sets $S_{l+1}, \ldots, S_J$ and any imitation equilibrium $\sigma^*$, such that
    \[
    \frac{\sum_{t \in S_j} \theta(t) q(t) \int_{\hat \mu_j(u^*)}^{\inf_{v > u^*}\hat \mu_j(v)}\sigma^*(\mu f_j|t) d\mu}{\sum_{t \in S_j} q(t) \int_{\hat \mu_j(u^*)}^{\inf_{v > u^*}\hat \mu_j(v)}\sigma^*(\mu f_j|t) d\mu} = u^*
    \]
    and for all $\mu^* \in (\hat \mu_j(u^*), \inf_{v > u^*}\hat \mu_j(v))$,
    \[
    \frac{\sum_{t \in S_j: t \tilde \supseteq \mu^*} \theta(t) q(t) \int_{\hat \mu_j(u^*)}^{\inf_{v > u^*}\hat \mu_j(v)}\sigma^*(\mu f_j|t) d\mu}{\sum_{t \in S_j: t \tilde \supset \mu^*} q(t) \int_{\hat \mu_j(u^*)}^{\inf_{v > u^*}\hat \mu_j(v)}\sigma^*(\mu f_j|t) d\mu} \le u^*.
    \]
    But it is always feasible to reorder the imitation strategy to construct $\sigma^{**}$ such that $S_{l+1}, \ldots, S_J$ are unchanged, but if $\mu_1 f_{j}$ imitates $\mu_1' f_i$ and $\mu_2 f_{j}$ imitates $\mu_2' f_i$, with $\mu_1 > \mu_2$, then $\mu_1' > \mu_2'$ also. That is, conditional on imitating the same state, higher-data senders always imitate types with more data under $\sigma^{**}$. Then any type is imitated by either a single type or an interval of types under any other state; the latter is ruled out by the fact that it would result in a payoff no more than $\theta_l$ to the message. Once again, since there is a finite set of states, this ensures that each type is imitated by at most a finite set of other types.
\end{proof}

\begin{proof}[Proof of Claim \ref{claim:opt_mech}]
    Suppose $A$ to be the subset of types in $\mathcal{T}$ that are imitated under the imitation equilibrium $\sigma^*$, and suppose that $\sigma^*$ is an imitation equilibrium in which each type is imitated by a finite set of other types, which exists by the previous lemma. Given $\mu f_j \in A$, let $T_{\mu f_j}$ be the set of types that play $\mu f_j$ under $\sigma^*$, including $\mu f_j$ itself. Define a distribution over $T_{\mu f_j}$,
\[
q_{\mu f_j}(t) = \frac{q(t) \sigma(\mu f_j|t)}{\sum_{t \in T_{\mu f_j}}q(t) \sigma(\mu f_j|t)},
\]
which is the probability of type $t$ conditional on the message $\mu f_j$.
    
    Call the optimal direct mechanism $a^*$, that responds with the action $a^*(t)$ after receiving the report $t$. It must satisfy IC across any subset of types, $T \subseteq \mathcal{T}$, but let us consider instead $w$, the solution to a relaxed local problem where we impose that IC must hold only between $t, t' \in T_{\mu f_j}$  when types are distributed according to $q_{\mu f_j}$. We will show that for all $t \in T_{\mu f_j}$, we have $w(t) = \mathbb{E}_{q_{\mu f_j}}[\theta]$, and that taking this solution across all $\mu f_j \in A$ assigns a response for the receiver to all $t \in T$ while preserving global IC, and therefore gives the optimal direct mechanism.

We know that $w(\mu f_j) \le w(t)$ for all $t \in T_{\mu f_j}$. Let $S_{\mu f_j} = \{t \in T_{\mu f_j}: w(t) = w(\mu f_j)\}$. First, note that if $S_{\mu f_j} = T_{\mu f_j}$, then we optimally have $w(t) = \mathbb{E}_{q_{\mu f_j}}[\theta]$ for all $t \in T$. This leaves us to rule out that $w(t) \neq w(t’)$ for some $t, t’ \in T_{\mu f_j}$.
    
We rule out that $w(\mu f_j) \ge \mathbb{E}_{q_{\mu f_j}}[\theta]$ and $w(t) \neq w(\mu f_j)$ for some $t \in T_{\mu f_j}$, due to the fact that the receiver can then improve their payoff while preserving IC by instead responding to every type with $w(\mu f_j)$. Next, we rule out that $w(\mu f_j) < \mathbb{E}_{q_{\mu f_j}}[\theta]$ and $w(t) \neq v(\mu f_j)$ for some $t \in T_{\mu f_j}$, since then it is possible to instead respond to every $t$ such that $w(t) = w(\mu f_j)$ with $\min_{t \in T_{\mu f_j}\setminus S} w(t)$, and, by single-peakedness of the receiver’s payoff function, this improves the receiver’s payoff.

This suffices to show that $w$ corresponds exactly to the outcome of the imitation equilibrium for all $t \in T_{\mu f_j}$, regardless of the choice of $\mu f_j \in A$. As $w$ optimizes the receiver’s payoff under a weaker set of IC constraints than $a^*$, we know that the imitation equilibrium outcome is at least as good as $a^*$ for the receiver; the reverse statement is immediate since every equilibrium outcome is implementable with commitment, and so the two are identical. \qed
\end{proof}

\begin{corollary}
    The imitation equilibrium outcome is the receiver-optimal equilibrium outcome.
\end{corollary}

\begin{proof}
    In every equilibrium $\sigma$, the receiver has a unique best response to each message, given by the action
    \[
    a_r(\beta(\cdot|m))) = \mathbb{E}_{\beta(\cdot|m)}[\theta].
    \]
    Any type of the sender therefore has an optimal feasible message to send that results in a unique optimal action that they can induce the receiver to take given the receiver's inference function. Any equilibrium outcome can therefore be implemented by the receiver through a direct mechanism that responds to every type with a deterministic message, and so there is no equilibrium that increases the receiver's payoff relative to the optimal pure-strategy mechanism outcome that is also the imitation equilibrium outcome. \qed
\end{proof}

\section{Proofs of properties of imitation equilibrium} \label{app:full_info}

\subsection{Strategies under MLRP}

\begin{proof}[Proof of Claim \ref{claim:mlrp}]
    MLRP implies that
    \[
    \frac{f_k(d)}{f_j(d)} \ge \frac{f_k(d')}{f_j(d')}
    \]
    whenever $k > j$ and $d > d'$. Then $r_j(k) = \frac{f_k(D)}{f_j(D)}$ for all $k > j$. 

    For all $j, k$, we have $r_j(J) = r_j(k)r_k(J)$; then all types with states $j < k$ that can send $\mu f_k$ can send $\frac{\mu}{r_k(J)}f_k$, and vice-versa. So, $\hat u_k(\mu) = \hat u_J(\frac{\mu}{r_k(J)})$, and for every imitation equilibrium $\sigma^*$ in which $\mu f_k$ is an on-path message, there is an outcome-equivalent equilibrium ${\sigma^*}'$ such that ${\sigma^*}'(\frac{\mu}{r_k(J)} f_J| \mu f_k) = 1$ and for every type $t$ that imitates $\mu f_k$ under $\sigma$, ${\sigma^*}'(\frac{\mu}{r_k(J)} f_J|t) = \sigma(\frac{\mu}{r_k(J)} f_J|t) + \sigma(\mu f_k|t)$, and otherwise strategies are unchanged. 
    
    It is therefore possible to construct a sequence of outcome-equivalent equilibria, beginning with the imitation equilibrium, that terminates in an equilibrium in which all types play $\mu f_J$ for some $\mu$.
\end{proof}

\subsection{Convergence to full-information outcome as $Var(g) \rightarrow 0$}

\begin{proof}[Proof of Claim \ref{claim:var}]
    We show that given any infinite sequence of games with data-mass distributions $g_1, g_2, \ldots $ on $[0,1]$ with a fixed mean and variances $Var_1, Var_2, \ldots \rightarrow 0$, that are identical in the set of states and their ex-ante distribution, the payoff to a sender conditional on the state converges in probability to their full-information payoff.
    
    In order to do so, we show that for any $\delta$ and $\epsilon$, there exists $L$ such that for all $l \ge L$, the distribution $g_l$ is such that $Pr[u_{\sigma^*}(\mu, \theta_k) < \theta_k - \delta] < \epsilon$ under every state.
    
    Define the mean of $\mu$ to be $\bar \mu$, and
    \[
    B = \max_{j \neq k} \frac{1}{r_j(k)} 
    \]
    so that for any two states $j$ and $k$, the difference between the amount of the state-$k$ distribution that the mean type under state $k$ has and the amount the mean type under state $j$ has is $\bar \mu(1-B)$.
    
    Suppose that the variance of $\mu$ under density $g_L$ is less than $\Delta^2 \epsilon^2$, where $\Delta > 0$ is an arbitrary parameter. Then there can be at most a probability $\epsilon^2$ that the state is $k$ and the sender has less than $\bar \mu - \Delta$ data distributed like $f_k$.
    A sender under state $j$ has more than $\frac{\bar \mu-\Delta}{B}$ data with probability no more than $\frac{\Delta^2 \epsilon^2 B^2}{(\bar \mu(1-B)-\Delta)^2}$.

    Recall that whenever $u_{\sigma^*}(\mu, \theta) < \theta$, the type with dataset $\mu f_\theta$ is truthful in equilibrium. So, if under state $\theta_k$ we have $Pr[u_{\sigma^*}(\mu, \theta) < \theta - \delta] \ge \epsilon$, then the type with $\mu = G^{-1}(\epsilon)$ must obtain payoff less than $\theta - \delta$, and so must all types with less data, and all such types must be truthful. But the total mass of all types \emph{not} in state $k$ that can pool with types with $\mu \in [G^{-1}(\epsilon^2), G^{-1}(\epsilon)]$ cannot exceed 
    \[
    (J-1) (1-\beta_0(\theta_k)) \frac{\Delta^2 \epsilon^2 B^2}{(\bar \mu(1-B)-\Delta)^2}
    \]
    and so the payoff to type $G^{-1}(\epsilon) f_k$ cannot be less than
    \[
    \frac{\epsilon(1-\epsilon) \theta_k}{\epsilon(1-\epsilon) + (J-1) (1-\beta_0(\theta_k)) \frac{\Delta^2 \epsilon^2 B^2}{(\bar \mu(1-B)-\Delta)^2}}
    \]
    which, for small enough $\Delta$, must be at least $\theta_k - \delta$. Since there is always $L$ large enough that $Var_L < \epsilon^2 \Delta^2$, we are done.
    
    All that remains is to note that, since the ex-ante expected payoff must always be $\mathbb{E}_{\beta_0}[\theta]$, this lower bound on the probability of payoffs less than the full-information payoffs implies a corresponding upper bound on payoffs exceeding the full-information payoffs, and so we obtain convergence of the distribution of payoffs, state-by-state, to those in the outcome where the receiver knows the truth.
\end{proof}

\subsection{Comparative statics of welfare with respect to $\beta_0(\theta_j)$}

\begin{proof}[Proof of Claim \ref{claim:beta}]
    First, let $M(v)$ be the frontier of types that attain payoff $v$ under $\mathcal{G}$ and let $M'(v)$ be the frontier of types that do so under $\mathcal{G}'$. Let $q$ be the distribution of types in $\mathcal{G}$ and $q'$ be the type distribution for $\mathcal{G}'$.

    Let $v \ge \theta_j$. Suppose for the sake of contradiction that $U(M') \setminus U(M)$ is nonempty. By Lemma \ref{lem:pool_payoff}, in the game $\mathcal{G}'$, 
    \[
    \mathbb{E}_{q'}[\theta|t \in U(M') \setminus U(M)] \ge v.
    \]
    But we also have $\mathbb{E}_{q}[\theta|t \in U(M') \setminus U(M)] \ge \mathbb{E}_{q'}[\theta|t \in U(M') \setminus U(M)]$. So then $\mathbb{E}_{q'}[\theta|t \in U(M') \setminus U(M)] \ge v$, but then the construction algorithm in game $\mathcal{G}$, if it ever reached $M$, would instead set $M'$ as a frontier for payoff $v$, and so this is impossible.

    Similarly, let $v \le \theta_j$. As with the above, we observe that if $U(M) \setminus U(M')$ is nonempty, then
    \[
    \mathbb{E}_{q}[\theta|t \in U(M) \setminus U(M')] \ge v,
    \]
    but since $\mathbb{E}_{q'}[\theta|t \in U(M) \setminus U(M')] \ge \mathbb{E}_{q}[\theta|t \in U(M) \setminus U(M')]$, this implies that $\mathbb{E}_{q}[\theta|t \in U(M) \setminus U(M')] \ge v$, which is likewise impossible by the algorithm.
\end{proof}

\begin{proposition}
    Suppose that two games $\mathcal{G}$ and $\mathcal{G}'$ are identical except for their space of outcomes $\mathcal{D}$ and $\mathcal{D}'$ and the generating distributions of data under each state, $\{f_j\}_{j=1}^J$ and $\{f_j'\}_{j=1}^J$, and let $\sigma^*$ and $\sigma^{*'}$ be their respective imitation equilibria.

    If the $r_j(k) \ge r_j'(k)$ for all $j,k$, then the receiver's payoff is greater under $\sigma^*$ than under $\sigma^{*'}$.
\end{proposition}

\begin{proof}
    Under game $\mathcal{G}$, there exists a (pure-strategy) mechanism that implements the outcome of $\sigma^{*'}$. To see this, note that the outcome of $\sigma^{*'}$ is also the outcome of $v'$, the optimal mechanism for the receiver in $\mathcal{G}'$, which respects the IC constraints that can be rewritten as
    \begin{equation}
    \tag{IC-$\mathcal{G}'$} v'(\mu f_j) \ge v'(\frac{\mu}{r_j(k)} f_k) \ \forall \mu, j, k \ \text{ and } v'(\mu_1 f_j) \ge v'(\mu_2 f_j) \ \forall j, \mu_1 > \mu_2. 
    \end{equation}
    On the other hand, in order to be implementable in $\mathcal{G}$, $v'$ need only respect the IC constraints
    \begin{equation}
    \tag{IC-$\mathcal{G}$} v'(\mu f_j) \ge v'(\frac{\mu}{r_j(k)} f_k) \ \forall \mu, j, k \ \text{ and } v'(\mu_1 f_j) \ge v'(\mu_2 f_j) \ \forall j, \mu_1 > \mu_2, 
    \end{equation}
    which are weaker.

    Since $v'$ is implementable in $\mathcal{G}$, the outcome of the optimal mechanism, and therefore the imitation equilibrium, in $\mathcal{G}$ gives at least a weak improvement over $v'$ for the receiver.
\end{proof}

\section{Proof of convergence to continuum limit} \label{app:limitpf}

\begin{proof}
The proof of theorem \ref{thm:finitelimit} uses Lemma \ref{lem:pool_payoff} to establish that, for any set of messages $M$, when the set of all types in $\mathcal{T} \setminus U(M)$ that attain a payoff of at least $v$ in $\sigma^*$ is nonempty, their payoff when they form a pool is at least $v$. Using this, we show that $u_{\sigma^*}$ is a lower bound on payoffs for types in $\mathcal{T}$ in the limit. Then, Bayes plausibility implies that
\[
\lim_{N\rightarrow \infty} \sum_{j=1}^j \beta_0(\theta_j) \int_{\mu=0}^1 u_{\sigma_N}(\mu f_j) g^j(\mu) d\mu = \mathbb{E}_{q}[\mathbb{E}_{\beta(\cdot|\sigma(t))}[\theta]|t] = \mathbb{E}_{\beta_0}[\theta],
\]
which in conjunction with the lower bound implies that in the limit outcomes must coincide exactly with $u_{\sigma^*}$ for types in $\mathcal{T}$.

Before proceeding to construct bounds on payoffs in the finite games, it is helpful to define a neighborhood of $\mathcal{T}$ as the set of types in each finite game with datasets distributed similarly to the underlying distribution in some state. For $\eta \in (0,1]$ and $k \in [0,1]$, define
\[
S_{N}(\eta, k) = \{t \in \mathcal{T}_{N}: |t| \ge k \text{ and } \exists \theta \text{ s.t. } \sup_{d}|t(d) - |t| f_\theta(d)| \le \eta \}.
\]

Fix an integer $n$. Conditional on $|t| = n$ and the true state being $\theta_j$, the Glivenko-Cantelli theorem implies that there is a bound on the probability that $\sup_d |\sum_{x=1}^d t(x) - \frac{n}{N} F_j(d)| > \eta$ that decreases to $0$ for large $n$, irrespective of $N$. Because data have a discrete distribution, this implies a similar bound on the empirical probability mass function: if $|t| = n$ and $\theta_j$ is the true state, the probability that $\sup_d |t(d) - \frac{n}{N}f_j(d)| > \eta$ is at most $b_{=}(n, \eta)$, with $\lim_{n \rightarrow \infty} b_=(n, \eta) = 0$ for all $\eta > 0$. If the true state is $\theta_{j'} \neq \theta_j$ and $|t| = n$, then the probability that $\sup_d |t(d) - \frac{n}{N}f_j(d)| > \eta$ is at least $b_{\neq}(n, \eta)$, with $\lim_{\eta \rightarrow 0} \lim_{n \rightarrow \infty} b_{\neq}(n, \eta) = 1$.

When $N$ and $k$ are large, the proportion of types that lie in $S_N(\eta,k)$ is close to 1, for all $\eta$. In particular, $\lim_{k \rightarrow 0} \lim_{\eta \rightarrow 0} \lim_{N \rightarrow \infty} q_{N}(S_{N}(\eta, k)) = 1$, 
since:
    \begin{itemize}
        \item With probability decreasing to $0$ as $k \rightarrow 0$, $|t| < k$.
        \item For fixed $k$ and $\eta$, the probability that there does not exist $\theta$ such that $\sup_{d}|t(d) - |t| f_\theta(d)| \le \eta$ given that $|t|  \ge k$ decreases to $0$ as $N k \rightarrow \infty$.
    \end{itemize}
    
We may further subdivide $S_{N}(\eta, k)$ into a set of types associated with each state,
\[
S^{j}_{N}(\eta, k) = \{t \in S_{N}(\eta, k): \sup_{d}|t(d) - |t| f_j(d)| \le \eta \}.
\]
A further consequence of the convergence of empirical distributions is that, when $N k \rightarrow \infty$ and $\eta \rightarrow 0$, the sets $(S^{j}_{N}(\eta, k))_{j = 1}^J$ are disjoint.
Additionally, for all $t \in S_{N}^j (\eta, k)$, there is a uniform lower bound on the probability that the state is $\theta_j$ given that the sender is of type $t$, which we call $w(k, \eta, N)$, with $\lim_{k \rightarrow 0} \lim_{\eta \rightarrow 0} \lim_{N \rightarrow \infty} w(k, \eta, N) = 1$. 

In addition, we can lower-bound $q_{N}(\{t \in S_{N}^j (\eta, k): \underline \mu f_j \tilde \subseteq t \tilde \subseteq \bar \mu f_j \})$ for all $k < \underline \mu < \bar \mu$. Let $\Delta(N)$ be a bound on $\sup_{n, j} |(G^j_N(n)) - G^j(n)|$ that goes to $0$ as $N \rightarrow \infty$. Observe that if $\underline \mu + \eta < |t| < \bar \mu - \eta$ and $t \in S_{N}^j (\eta, k)$, then $\underline \mu f_j \tilde \subseteq t \tilde \subseteq \bar \mu f_j$, so a lower bound is 
\begin{equation} \label{eq:l_bound_q}
q_{N}(\{t \in S_{N}^j (\eta, k): \underline \mu f_j \tilde \subseteq t \tilde \subseteq \bar \mu f_j \}) \ge \beta_0(\theta_j)(1-b_{=}(Nk, \eta))(G^j(\bar \mu - D\eta) - G^j(\underline \mu + D\eta) - \Delta(N)).
\end{equation}
Similarly, there is an upper bound on $q_{N}(\{t \in S_{N}^j (\eta, k): t \not \tilde \subseteq \underline \mu f_j  \text{ and } \bar \mu f_j\not \tilde \subseteq t \})$:
\begin{equation} \label{eq:h_bound_q}
q_{N}(\{t \in S_{N}^j (\eta, k): t \not \tilde \subseteq \underline \mu f_j,  \ \  \bar \mu f_j \not \tilde \subseteq t \}) \le \beta_0(\theta_j)(G^j(\bar \mu + D \eta) - G^j(\underline \mu - D \eta) + \Delta(N)) + (1-\beta_0(\theta_j)) b_{\neq}(k N, \eta).
\end{equation}

Now we proceed to construct a lower bound for $u_{\sigma_{N}}(\hat{\boldsymbol{\mu}} f_{\hat j})$. First, recall that $u_{\sigma_N}(\mu f_j) \ge \max_{\{f \in \mathcal{T}_{N}: t \tilde \subseteq \mu f_j\}} u_{\sigma_N}(t)$. Observe that there exists a dataset $\hat t = \frac{1}{N}(\lfloor N \hat{\boldsymbol{\mu}} f_{\hat \theta}(1) \rfloor, \ldots, \lfloor N \hat{\boldsymbol{\mu}} f_{\hat \theta}(k) \rfloor)$ in $\mathcal{T}_{N}$ and that $u_{\sigma_N}(\hat{\boldsymbol{\mu}} f_{\hat j}) \ge u_{\sigma_N}(\hat t)$. 

For a given $N$, suppose $\hat t$ belongs to the $m$th upper pool under the algorithm that constructs $\sigma_{N}$. Denote by $\hat M_N(m-1)$ the set of messages that implement the upper pools in step $1, \ldots, m-1$, and fix ${\mathcal{T}}_{N, m} = \mathcal{T}^+_N(\hat M_N(m-1))$ to be the set of remaining types at the start of the $m$th step of the algorithm that constructs $\sigma_{N}$; therefore, $\hat t$ belongs to ${\mathcal{T}}_{N, m}$. 

Let $\underline M(\epsilon, N)$ be the set of on-path messages that result in a payoff of $u_\sigma((\hat{\boldsymbol{\mu}} - \epsilon) f_{\hat j})$ under infinite data. 
We see that the set of types in $\mathcal{T}_{N,m}^+(\underline M(\epsilon, N))$ includes $\hat f$ when $N$ is large enough. From Lemma \ref{lem:pool_payoff}, there is an upper pool in $\mathcal{T}_{\hat N, m}$ that achieves a payoff of at least $u(\mathcal{T}_{\hat N, m}^+(\underline M(\epsilon, N)))$, so $u_{\sigma_{N}}(\hat{\boldsymbol{\mu}} f_{\hat j})$ is lower-bounded by $u(\mathcal{T}_{\hat N, m}^+(\underline M(\epsilon, N)))$.

Let $(\underline \mu_1(\epsilon, N), \ldots, \underline \mu_J(\epsilon, N))$ be a vector that gives the minimum mass of data under distributions $f_{1}, \ldots, f_J$, respectively, such that the dataset contains some message in $\underline M(\epsilon, N)$, and let $(\bar \mu_1(N), \ldots, \bar \mu_J(N))$ be the maximum mass of data under each distribution such that there does not exist $t \in \mathcal{T}_{\hat N, m}$ such that $t \tilde \subseteq \bar \mu_j f_j$. All $t \in \mathcal{T}_{\hat N, m}^+(\underline M(\epsilon))$ satisfy $t \not \tilde \subseteq \underline \mu_j(\epsilon, N) f_j$ and $\bar \mu_j(N) f_j \not \tilde \subseteq t$, and all $t$ satisfying $\underline \mu_j(\epsilon, N) f_j \tilde \subseteq t \tilde \subseteq \bar \mu_j(N) f_j$ for some $j$ are in $\mathcal{T}_{\hat N, m}^+(\underline M(\epsilon))$.

We may rewrite
\begin{equation}
    u(\mathcal{T}_{\hat N, m}^+(\underline M(\epsilon))) = \frac{\sum_{j=1}^J \sum_{t \in \mathcal{T}_{\hat N, m}^+(\underline M(\epsilon))} q_N(t) \theta_j \pi_N(\theta_j|t)}{\sum_{t \in \mathcal{T}_{\hat N, m}^+(\underline M(\epsilon))} q_N(t)}.
\end{equation}

Let the numerator be $Q(N, \hat{\boldsymbol{\mu}} f_{\hat j}, \epsilon)$ and the denominator be $R(N, \hat{\boldsymbol{\mu}} f_{\hat j}, \epsilon)$. Analogously to eq. \ref{eq:l_bound_q}, a lower bound for $Q(N, \hat{\boldsymbol{\mu}} f_{\hat j}, \epsilon)$ is
\begin{equation}\label{eq:l_bound_q_final}
    \underline Q(N, \hat{\boldsymbol{\mu}} f_{\hat j}, \epsilon) = \sum_j \beta_0(\theta_j)\theta_j[G^j(\bar \mu_j(N) - \eta D) - G^j(\max(\underline \mu_j(\epsilon, N) + \eta D, k)) - \Delta(N)]w(k, \eta, N)(1-b_{=}(k, \eta)),
\end{equation}
and it follows from eq. \ref{eq:h_bound_q} that an upper bound for $R$ is
\begin{equation} \label{eq:h_bound_r}
\begin{split}
    \bar R(N, \hat{\boldsymbol{\mu}} f_{\hat j}, \epsilon) = &\left( \sum_j \beta_0(\theta_j) [G^j(\bar \mu_j(N) + \eta D) - G^j(\underline \mu_j(\epsilon, N) - \eta D) + \Delta(N)] \right)\\
    & + J(1-b_{\neq}(k, \eta)) + (1-q_N(S_N(\eta, k))).
\end{split}
\end{equation}
We have 
\[
\lim_{k \rightarrow 0} \lim_{\eta \rightarrow 0} \lim \inf_{N \rightarrow \infty} \underline Q \ge \lim \inf_{N \rightarrow \infty} \sum_{j=1}^J \beta_0(\theta_j) \theta_j (G^j(\bar \mu_j(N)) - G^j(\underline \mu_j(\epsilon, N)))
\]
and 
\[
\lim_{k \rightarrow 0} \lim_{\eta \rightarrow 0} \lim \inf_{N \rightarrow \infty}\bar R \le \lim \inf_{N \rightarrow \infty} \sum_{j=1}^J \beta_0(\theta_j)(G^j(\bar \mu_j(N)) - G^j(\underline \mu_j(\epsilon, N))).
\]
Both of the RHS are finite and strictly positive for all $N$ and $\epsilon > 0$; therefore,
\begin{equation}
\begin{split}
\lim_{k \rightarrow 0} \lim_{\eta \rightarrow 0} \lim \inf_{N \rightarrow \infty} \frac{\ \underline Q\ }{\ \bar R \ } &\ge \lim \inf_N \frac{\sum_{j=1}^J \beta_0(\theta_j) \theta_j (G^j(\bar \mu_j(N)) - G^j(\underline \mu_j(\epsilon, N)))}{\sum_{j=1}^J \beta_0(\theta_j) (G^j(\bar \mu_j(N)) - G^j(\underline \mu_j(\epsilon, N)))}\\
&= \lim \inf_N \mathbb{E}[\theta|t \in T(u_{\sigma}((\hat{\boldsymbol{\mu}} -\epsilon) f_{\hat j}), \hat M_N(m-1))]\\
&\ge u_{\sigma}((\hat{\boldsymbol{\mu}} -\epsilon) f_{\hat j}),
\end{split}
\end{equation}
where the last inequality follows from Lemma \ref{lem:pool_payoff}.

Because $k$ and $\eta$ are arbitrary variables used to obtain the bound, it follows from this that $\lim_{N \rightarrow \infty} u(\mathcal{T}_{\hat N, m}^+(\underline M(\epsilon))) \ge u_{\sigma}((\hat{\boldsymbol{\mu}} -\epsilon) f_{\hat j})$. Finally, because payoffs are continuous, taking a sequence of bounds as $\epsilon \rightarrow 0$ implies that $\lim \inf_{N \rightarrow \infty} u_{\sigma_N}(\hat{\boldsymbol{\mu}} f_{\hat j}) \ge \lim_{\epsilon \rightarrow 0} \lim \inf_{N \rightarrow \infty} u(\mathcal{T}_{\hat N, m}^+(\underline M(\epsilon))) \ge u_{\sigma}(\hat{\boldsymbol{\mu}} f_{\hat j})$.

The last step is to show that
\[
\lim_{N\rightarrow \infty} \sum_{j=1}^J \beta_0(\theta_j) \int_{\mu=0}^1 u_{\sigma_N}(\mu f_j) g^j(\mu) d\mu = \mathbb{E}_{\beta_0}[\theta].
\]
Since we know already that
\[
\lim_{N\rightarrow \infty} \sum_{j=1}^J \beta_0(\theta_j) \int_{\mu=0}^1 u_{\sigma}(\mu f_j) g^j(\mu) d\mu = \mathbb{E}_{\beta_0}[\theta]
\]
and $\lim \inf_{N \rightarrow \infty} u_{\sigma_N}(\mu f_j) \ge u_{\sigma}(\mu f_j)$ for all $\mu f_j \in \mathcal{T}$, this additional fact suffices to ensure that $u_{\sigma_N}(\cdot) = u_{\sigma}(\cdot)$ over $\mathcal{T}$.

The proof comes from dividing $\mu \in (k,1)$ into $X$ chunks, with the $x$th chunk given by $(\mu_{x-1}, \mu_x]$ where $\mu_x = x\frac{1-k}{X} + k$. 

Consider types $t \in S^j_N(\eta, k)$ such that $\mu_{x-1} f_j \tilde \subseteq t \tilde \subseteq \mu_{x} f_j$: their payoff under $\sigma_N$ has to be in $[u_{\sigma_N}(\mu_{x-1} f_j), u_{\sigma_N}(\mu_{x} f_j)]$. 
This implies that
\begin{equation}
\begin{split}
    \underline V_N(k,\eta,X) &= \sum_{j=1}^J \beta_0(\theta_j) \sum_{x=1}^X u_{\sigma_N}(\mu_x f_j) [G^j(\mu_{x+1} - \eta D) - G^j(\mu_x + \eta D) - \Delta(N)](1-b_{=}(k, \eta))\\
    &\le \mathbb{E}_{\beta_0}[\theta],
\end{split}
\end{equation}

since $\underline V_N(k,\eta,X)$ is a lower bound for the total probability-weighted sum of payoffs under $\sigma_N$ over $t \in \mathcal{T}_N \bigcup S_N(\eta, k)$, while $\mathbb{E}_{\beta_0}[\theta]$ is equal to the total probability-weighted sum of payoffs under $\sigma_N$ of all types in $\mathcal{T}_N$.

Finally, the difference between $\sum_{j=1}^J \beta_0(\theta_j) \int_{\mu=0}^1 u_{\sigma_N}(\mu f_j)g^j(\mu) d\mu$ and $\underline V_N(k,\eta,X)$ vanishes as $X \rightarrow \infty$, $k \rightarrow 0$, $\eta \rightarrow 0$, and $N \rightarrow \infty$. To see this, observe that if $c$ is an upper bound on $g$ (which exists because $g$ is continuous on compact interval $[0,1]$),
\begin{equation}
\begin{split}
V_N(k,\eta,X) \ge \sum_{j=1}^J \beta_0(\theta_j) \Big( \sum_{x=1}^X & u_{\sigma_N}(\mu_x f_j) [G^j(\mu_{x+1}) - G^j(\mu_x)]\\
 &- (\theta_J b_{=}(k, \eta)[G^j(\mu_{x+1}) - G^j(\mu_x)] + 2c\eta D + \Delta(N)) \Big)\\
 \ge \sum_{j=1}^J \beta_0(\theta_j) \sum_{x=1}^X & u_{\sigma_N}(\mu_x f_j) [G^j(\mu_{x+1}) - G^j(\mu_x)]\\
 - JX \theta_J (b_{=}&(k, \eta) + 2c\eta D + \Delta(N)).
\end{split}
\end{equation}
Then, for any $\epsilon$ and $j$, define $\xi_N^j(\epsilon, X)$ to be the set of values of $x$ such that $u_{\sigma_N}(\mu_{x+1} f_j) - u_{\sigma_N}(\mu_{x} f_j) > \epsilon$. The size of $\xi_N^j(\epsilon, X)$ is at most $\frac{\theta_J}{\epsilon}$. For all $x \not \in \xi_N^j(\epsilon, X)$, we have the bound $\int_{\mu_x}^{\mu_{x+1}} u_{\sigma_N}(\mu f_j) g^j(\mu) d\mu - u_{\sigma_N}(\mu_x f_j)[G^j(\mu_{x+1}) - G^j(\mu_x)] < \epsilon [G^j(\mu_{x+1}) - G^j(\mu_x)]$. So,
\begin{equation}
\begin{split}
    \sum_{j=1}^J \beta_0(\theta_j) &\int_{\mu=0}^2 u_{\sigma_N}(\mu f_j)g^j(\mu) d\mu - \underline V_N(k,\eta,X)\\
    \le & \left(\sum_{j=1}^J \beta_0(\theta_j) \sum_{x=1}^X \left(\int_{\mu_x}^{\mu_{x+1}} u_{\sigma_N}(\mu f_j) g^j(\mu) d\mu - u_{\sigma_N}(\mu_x f_j)[G^j(\mu_{x+1}) - G^j(\mu_x)]\right)\right)\\
    & + JX \theta_J ( b_{=}(k, \eta) + 2c\eta D + \Delta(N) + (1-q_N(S_N(\eta, k))))\\
    \le & \sum_{j=1}^J \left( \beta_0(\theta_j) \left( \sum_{x \not \in \xi_N^j(\epsilon, X)} \epsilon [G^j(\mu_{x+1}) - G^j(\mu_x)] \right) + \left( \sum_{x \in \xi_N^j(\epsilon, X)} \theta_J [G^j(\mu_{x+1}) - G^j(\mu_x)]\right)\right)\\
    & + JX \theta_J (b_{=}(k, \eta) + 2c\eta D + \Delta(N) +(1-q_N(S_N(\eta, k))))\\
    \le & \epsilon + J\frac{c(1-k)}{X} \frac{\theta_J^2}{\epsilon} + JX \theta_J \Big(b_{=}(k, \eta) + 2c\eta D + \Delta(N) +(1-q_N(S_N(\eta, k)))\Big)
\end{split}
\end{equation}
since $\sum_{x \in \xi_N^j(\epsilon, X)} [G^j(\mu_{x+1}) - G^j(\mu_x)] \le \frac{c(1-k)}{X} \frac{\theta_J}{\epsilon}$. Then
\[
\lim_{\epsilon \rightarrow 0} \lim_{X \rightarrow \infty} \lim_{k\rightarrow 0} \lim_{\eta \rightarrow 0} \lim_{N\rightarrow \infty}\sum_{j=1}^J \beta_0(\theta_j) \int_{\mu=0}^2 u_{\sigma_N}(\mu f_j)g^j(\mu) d\mu - \underline V_N(k,\eta,X) = \lim_{\epsilon \rightarrow 0} \lim_{X \rightarrow \infty} \epsilon + J\frac{c(1-k)}{X} \frac{\theta_J^2}{\epsilon} = 0.
\]
Again, since $\epsilon$, $X$, $k$, and $\eta$ were all constructed variables, this implies that
\[
\lim_{N\rightarrow \infty}\sum_{j=1}^J \beta_0(\theta_j) \int_{\mu=0}^2 u_{\sigma_N}(\mu f_j)g^j(\mu) d\mu = \lim_{\epsilon \rightarrow 0} \lim_{X \rightarrow \infty} \lim_{k\rightarrow 0} \lim_{\eta \rightarrow 0} \lim_{N\rightarrow \infty} \underline V_N(k,\eta,X) \le \mathbb{E}_{\beta_0}[\theta].
\]
As it is already clear from the lower bound on $u_{\sigma_N}(\mu f_j)$ that $\lim_{N\rightarrow \infty}\sum_{j=1}^J \beta_0(\theta_j) \int_{\mu=0}^2 u_{\sigma_N}(\mu f_j)g^j(\mu) d\mu \ge \mathbb{E}_{\beta_0}[\theta]$, equality obtains.
\end{proof}

\section{Strategic convergence}
\begin{proposition} \label{strategic_convergence}
    Suppose that $\{\mathcal{G}_{N}\}_{N=1}^\infty$ converge to $\mathcal{G}_\infty$. Then for all $p^*, \rho, \eta > 0$, there is $\underline{N}(p^*, \eta)$  such that for all $N > \underline{N}(p^*, \eta)$, conditional on $|u_{\sigma^*_{N}}(t) - \theta_k| > \eta$ for all $k$, there is at least  probability $1-p^*$ that $t$ sends a message with (sup norm) distance at most $\delta$ from some $t_\infty \in \mathcal{T}_\infty$ that is on-path in $\sigma^*_\infty$ and such that $|u_{\sigma^*_N}(t) - u_{\sigma^*_\infty}(t_\infty)| \le \rho$.
\end{proposition}
This is a partial characterization of large-$N$ equilibrium strategies, saying that among types that obtain payoff bounded away by an arbitrarily small amount from the rewards to certainty about any particular state, the likelihood of playing a message close to their optimal message under limit-game beliefs is very high when there is plentiful access to data. In other words, these types play imitation-like strategies. This follows from the convergence theorem, since in truth-leaning equilibrium a type that receives payoff less than its full-information payoff always discloses its full dataset, so types with $u_{\sigma^*_N}(t) << \mathbb{E}_{\pi(\cdot|t)}[\theta]$ tell the truth; convergence of outcomes implies that they receive payoffs similar to those obtained by nearby types in $T_\infty$ under $\beta_{\sigma^*_\infty}$, and convergence of the type distribution implies that most such senders are indeed near some type in $T_\infty$. In aggregate a similar set of imitators must pool with such senders as the set of imitators pooling with better-state senders in $\sigma^*_\infty$, which means that types with $u_{\sigma^*_N}(t) >> \mathbb{E}_{\pi(\cdot|t)}[\theta]$ play messages close to $\mathcal{T}_\infty$ with high probability. 

The caveat is that when $|\mathbb{E}_{\beta_{\sigma^*_N}(\cdot|m)}[\theta] - \theta_k|$ is small for some $k$, then there may be no significant mass of senders playing $m$ to earn a payoff much greater or much less than their full-information payoff, which makes it hard to apply the technique of matching imitators to the imitated, though we do not have a counterexample for this case. From Corollary \ref{thm:segment}, we know that there is a positive-measure set of types that receive payoffs close to their full-information payoffs in $\sigma_\infty^*$, and the proposition does not pin down the large-$N$ limit of equilibrium strategies of types close to them, but generically, besides these, the set of types excluded from the proposition is measure-0.\footnote{Genericity here can be with respect to perturbations in $\beta_0$ or $\theta_1, \ldots, \theta_J$.}

\begin{proof}[Proof of Prop. \ref{strategic_convergence}]
Define $m_{\sigma^*_N}(t)$ to be the realization of the message played when the sender's type is $t$ -- formally, $m_{\sigma^*_N}(t)$ is a random variable with outcomes in $\mathcal{M}_N$ whose distribution is  given by the equilibrium strategy $\sigma^*_N(\cdot|t)$.

Define $A_N(x, \Delta; \epsilon)$ to be the set of types $t \in \mathcal{T}_N \bigcap T(\epsilon)$ such that $u_{\sigma^*_N}(t) > \mathbb{E}_{\pi(\cdot|t)}[\theta]$, and $u_{\sigma^*_N}(t) \in (x,x+\Delta]$.

Define $B_N(x, \Delta; \epsilon)$ to be the set of types $t' \in \mathcal{T}_N \bigcap T(\epsilon)$ with $u_{\sigma^*_N}(t') < \mathbb{E}_{\pi(\cdot|t)}[\theta]$, and $u_{\sigma^*_N}(t') \in (x,x+\Delta]$.

Define $X(\eta, \xi, \omega) = \{x \in \bigg[\max_{j} \hat u_{j}(\xi) + \omega, u_{\sigma^*}(f_{\theta_J})\bigg): \min_k |u_{\sigma^*_N}(t) - \theta_k| > \eta \}$.

For small enough $\eta$, $\xi$ and $\omega$, $X(\eta, \xi,\omega)$ is nonempty. On the other hand, $A_N(x, \Delta; \epsilon)$ and $B_N(x, \Delta; \epsilon)$ may be empty, in particular for small $N$. However, for $x \in X(\eta, \xi, \omega)$, there is large-enough $\underline{N}^*$ so that they are nonempty for all $N > \underline{N}^*$. Continuity of $u_{\sigma^*}(\mu f_j)$ ensures there is positive-measure set of types in $\mathcal{T}$ with $u_{\sigma^*_{infty}}(t) \in [x, x+\Delta]$; the bound away from $\theta_k$ for all $k$ ensures that some such types have $u_{\sigma^*_{infty}}(t) - \mathbb{E}_{\pi(\cdot|t)}[\theta] \ge \eta$ and some have $u_{\sigma^*_infty}(t) \le \mathbb{E}_{\pi(\cdot|t)}[\theta] < -\eta$, and so there is a positive-measure set of types nearby with the same properties under $\sigma^*_N$ in $\mathcal{T}_N$ for large-enough $N$.

We first prove a claim. 
\begin{claim}\label{claim:strat_conv}
    If $\{\sigma^*_{N}\}_{N=1}^\infty$ are truth-leaning equilibria of games $\mathcal{G}_{N}$ that converge to limit game $\mathcal{G}$ with imitation equilibrium $\sigma^*$, then for any $\eta > 0$, $\xi > 0$, $\omega > 0$ and $p > 0$, there exists $\bar \epsilon > 0$ and $\bar \Delta > 0$ such that, for all $x \in X(\eta, \xi, \omega)$, the probability conditional on $t \in A_N(x,\Delta;\epsilon)$ that $m_{\sigma^*_N}(t) \in B_N(x, \Delta; \epsilon)$ is at least $1-p$ in the limit as $N \rightarrow 0$ for all $\epsilon < \bar \epsilon$ and $\Delta < \bar \Delta$.
\end{claim}

\begin{proof}[Proof of Claim \ref{claim:strat_conv}]
    Expanding out the realization of $m_{\sigma^*_N}(t)$, this is equivalent to saying that for given $\eta > 0$, $\xi > 0$, $\omega>0$, $p > 0$, there exists $\bar \Delta > 0$ and $\bar \epsilon > 0$ so that for all $\Delta < \bar \Delta$ and $\epsilon < \bar \epsilon$,
\[
\lim_{N \rightarrow \infty} \frac{\sum_{t \in A_N(x,\Delta;\epsilon)} \left[q_N(t) \sum_{t' \in B_N(x, \Delta; \epsilon)} \sigma^*_N(t'|t)\right]}{\sum_{t \in A_N(x,\Delta;\epsilon)} q_N(t)} \ge 1-p.
\]
for all  $x \in X(\eta, \xi, \omega)$.

We have that for any $\xi > 0$, $\omega > 0$,
\[
\lim_{\epsilon \rightarrow 0} \lim_{N \rightarrow \infty} \min_{\theta_j} \max_{t \in \mathcal{T}_N \bigcap T(\epsilon): u_{\sigma^*_N}(t) \ge \max_{j} \hat u_{j}(\xi) + \omega} |\theta_j - \mathbb{E}_{\pi(\cdot|t)}[\theta]| = 0.
\]
Thus for any $\nu > 0$, $\omega > 0$ and $\xi > 0$, there exist small-enough $\bar \epsilon(\nu, \xi, \omega) > 0$ and large-enough $\underline N(\nu, \xi, \omega, \epsilon)$ defined for $\epsilon < \bar \epsilon(\nu, \xi)$ such that $\min_{\theta_j} \max_{t \in \mathcal{T}_N \bigcap T(\epsilon): u_{\sigma^*_N}(t) \ge \max_{j} \hat u_{j}(\xi) + \omega} |\theta_j - \mathbb{E}_{\pi(\cdot|t)}| < \nu$ for all $\epsilon < \bar \epsilon(\nu, \xi, \omega)$, $N > \underline{N}(\nu, \xi, \omega, \epsilon)$.

If we take $\Delta < \eta/3$ and $\nu < \eta/3$, then whenever $|\theta_k - \mathbb{E}_{\pi(\cdot|t)}| < \nu$ for some $k$ and $u_{\sigma^*_N}(t) \in [x, x+\Delta]$ for $x$ in $X(\eta, \xi, \omega)$, we have that $|u_{\sigma^*_N}(t) - \mathbb{E}_{\pi(\cdot|t)}| > \eta/3$. In particular, $u_{\sigma^*_N}(t) \neq \mathbb{E}_{\pi(\cdot|t)}$, so, for any $x \in X(\eta, \xi, \omega)$, and for any $\Delta, \nu < \eta/3$ and any $\xi, \omega$ and $\epsilon < \bar \epsilon(\nu, \xi, \omega)$, $N > \underline N(\nu, \xi, \omega, \epsilon)$,
\[
A_N(x, \Delta; \epsilon) \bigcup B_N(x, \Delta; \epsilon) = \{t \in T(\epsilon) \bigcap \mathcal{T}_N: u_{\sigma^*_N}(t) \in (x,x+\Delta]\}.
\]

In addition, the uniform convergence of outcomes on $\mathcal{T}$ and of the type distribution conditional on each state ensures that for all $\epsilon, \xi, \Delta$, and for all $x \ge \max_j \hat u_j(\xi) + \omega$,
\begin{equation}
    \lim_{N \rightarrow \infty} \sum_{t \in A_N(x, \Delta; \epsilon)} q(t) \mathbb{E}_{\pi(\cdot|t)}[\theta] = \sum_{\theta_j < \theta_k} \theta_j \beta_0(\theta_j)[G^j(\hat \mu_j(x+\Delta)) -G^j(\hat \mu_j(x))]
\end{equation}
and
\begin{equation}
    \lim_{N \rightarrow \infty} \sum_{t \in A_N(x, \Delta; \epsilon)} q(t)  = \sum_{\theta_j < \theta_k} \beta_0(\theta_j)[G^j(\hat \mu_j(x+\Delta)) -G^j(\hat \mu_j(x))]
\end{equation}
and likewise
\begin{equation}
    \lim_{N \rightarrow \infty} \sum_{t \in B_N(x, \Delta; \epsilon)} q(t) \mathbb{E}_{\pi(\cdot|t)}[\theta] = \sum_{\theta_j > \theta_k} \theta_j \beta_0(\theta_j)[G^j(\hat \mu_j(x+\Delta)) -G^j(\hat \mu_j(x))]
\end{equation}
and
\begin{equation}
    \lim_{N \rightarrow \infty} \sum_{t \in B_N(x, \Delta; \epsilon)} q(t) = \sum_{\theta_j > \theta_k} \beta_0(\theta_j)[G^j(\hat \mu_j(x+\Delta)) -G^j(\hat \mu_j(x))].
\end{equation}
Supposing that $x \in X(\eta, \xi, \omega)$ for some $\theta_k$, and $\epsilon < \bar \epsilon(\nu, \xi, \omega)$,  and $\Delta, \nu < \eta/3$, we know from the above that
\begin{equation}
\begin{split}
x & \le \frac{\sum_{\theta_j} \theta_j \beta_0(\theta_j)[G^j(\hat \mu_j(x+\Delta)) -G^j(\hat \mu_j(x))]}{\sum_{\theta_j} \beta_0(\theta_j)[G^j(\hat \mu_j(x+\Delta)) -G^j(\hat \mu_j(x))]}\\
& = \lim_{N \rightarrow \infty}\frac{\sum_{t \in B_N(x, \Delta; \epsilon)} q(t) \mathbb{E}_{\pi(\cdot|t)}[\theta] + \sum_{t \in A_N(x, \Delta; \epsilon)} q(t) \mathbb{E}_{\pi(\cdot|t)}[\theta]}{\sum_{t \in B_N(x, \Delta; \epsilon)} q(t) + \sum_{t \in A_N(x, \Delta; \epsilon)} q(t)}
\end{split}
\end{equation}

On the other hand, 
\begin{equation}
\begin{split}
 & x+\Delta \\
 & \ge \lim_{N \rightarrow \infty}\frac{\sum_{t \in B_N(x, \Delta; \epsilon)} q(t) \mathbb{E}_{\pi(\cdot|t)}[\theta] + \sum_{t \in A_N(x, \Delta; \epsilon)} \left[q(t) \mathbb{E}_{\pi(\cdot|t)}[\theta] \sum_{t' \in B_N(x, \Delta; \epsilon)}\sigma_N^*(t'|t)\right]}{\sum_{t \in B_N(x, \Delta; \epsilon)} q(t) + \sum_{t \in A_N(x, \Delta; \epsilon)} \left[q(t) \sum_{t' \in B_N(x, \Delta; \epsilon)}\sigma_N^*(t'|t)\right]}\\
 & = \lim_{N \rightarrow \infty} \frac{\sum_{\theta_j} \theta_j \beta_0(\theta_j)[G^j(\hat \mu_j(x+\Delta)) -G^j(\hat \mu_j(x))] -  \sum_{t \in A_N(x, \Delta;\epsilon)}q(t) \mathbb{E}_{\pi(\cdot|t)}[\theta]\sum_{t' \in B_N(x, \Delta;\epsilon)}(1-\sigma_N^*(t|t'))}{\sum_{\theta_j} \beta_0(\theta_j)[G^j(\hat \mu_j(x+\Delta)) -G^j(\hat \mu_j(x))] - \sum_{t \in A_N(x, \Delta;\epsilon)}q(t) \sum_{t' \in B_N(x, \Delta;\epsilon)}(1-\sigma_N^*(t|t'))}.
\end{split}
\end{equation}
But, we know that $\mathbb{E}_{\pi(\cdot|t)}[\theta] < u_{\sigma^*_N}(t) -  \eta/3 \le x+\Delta - \eta/3$ for any $t \in A_N(x, \Delta; \epsilon)$. Then, combining, we have
\begin{equation} \label{eq:delta_bound}
    \begin{split}
        0 \le & (x + \Delta) \left(\sum_{\theta_j} \beta_0(\theta_j)[G^j(\hat \mu_j(x+\Delta)) -G^j(\hat \mu_j(x))] - \lim_{N \rightarrow \infty} \sum_{t \in A_N(x, \Delta;\epsilon)}q(t) \sum_{t' \in B_N(x, \Delta;\epsilon)}(1-\sigma_N^*(t|t')) \right)\\
        & - x\left(\sum_{\theta_j}  \beta_0(\theta_j)[G^j(\hat \mu_j(x+\Delta)) -G^j(\hat \mu_j(x))]\right) -  \lim_{N \rightarrow \infty} \sum_{t \in A_N(x, \Delta;\epsilon)}q(t) \mathbb{E}_{\pi(\cdot|t)}[\theta]\sum_{t' \in B_N(x, \Delta;\epsilon)}(1-\sigma_N^*(t|t'))\\
        \le & \Delta \left(\sum_{\theta_j}  \beta_0(\theta_j)[G^j(\hat \mu_j(x+\Delta)) -G^j(\hat \mu_j(x))]\right)\\
        & \ \ \ \ \ \ \ \ - [(x+\Delta)-(x+\Delta - \eta/3)] \lim_{N \rightarrow \infty} \sum_{t \in A_N(x, \Delta;\epsilon)}q(t) \sum_{t' \in B_N(x, \Delta;\epsilon)}(1-\sigma_N^*(t|t'))\\
        = & \lim_{N \rightarrow \infty} \Delta  \left(\sum_{t \in A_N(x, \Delta; \epsilon) \bigcup B_N(x, \Delta; \epsilon)} q(t)\right) - \frac{\eta}{3} \lim_{N \rightarrow \infty} \sum_{t \in A_N(x, \Delta;\epsilon)}q(t) \sum_{t' \in B_N(x, \Delta;\epsilon)}(1-\sigma_N^*(t|t')).
    \end{split}
\end{equation}

Finally, since 
\[
\lim_{N\rightarrow \infty}\frac{\sum_{t \in A_N(x, \Delta; \epsilon) \bigcup B_N(x, \Delta; \epsilon)} \mathbb{E}_{\pi(\cdot|t)}[\theta]q(t)}{\sum_{t \in A_N(x, \Delta; \epsilon) \bigcup B_N(x, \Delta; \epsilon)} q(t)} \le x+ \Delta,
\]
we have
\[
\sum_{t \in B_N(x, \Delta; \epsilon)} (\theta_{k+1} - x - \Delta)q(t) \le \sum_{t \in A_N(x, \Delta; \epsilon)} (x + \Delta - \theta_{1})q(t)
\]
and so
\[
\frac{\sum_{t \in A_N(x, \Delta; \epsilon)} q(t)}{\sum_{t \in B_N(x, \Delta; \epsilon)} q(t)} \ge \frac{\eta/3}{\theta_J}.
\]
From the above and eq. \ref{eq:delta_bound}, we have 
\[
\lim_{N \rightarrow \infty} \frac{\sum_{t \in A_N(x,\Delta;\epsilon)} \left[q_N(t) \sum_{t' \in B_N(x, \Delta; \epsilon)} (1-\sigma^*_N(t'|t))\right]}{\sum_{t \in A_N(x,\Delta;\epsilon)} q_N(t)} \le \frac{\Delta(\theta_J + \eta)}{(\eta/3)^2}.
\]
Then for given $p$, $\xi$, $\omega$, $\eta$, and $\nu < \eta/3$ and $\epsilon < \bar \epsilon(\nu, \xi, \omega)$, as long as $\Delta \le \frac{p \eta^2}{9 (\theta_J + \eta)}$, we have 
\[
\lim_{N \rightarrow \infty} \frac{\sum_{t \in A_N(x,\Delta;\epsilon, \xi)} \left[q_N(t) \sum_{t' \in B_N(x, \Delta; \epsilon, \xi)} \sigma^*_N(t'|t)\right]}{\sum_{t \in A_N(x,\Delta;\epsilon, \xi)} q_N(t)} \ge 1-p
\]
and this bound is independent of $x$. So, letting $\bar \epsilon = \bar \epsilon(\eta/3, \xi, \omega)$ and $\bar \Delta = \frac{p \eta^2}{9 (\theta_J + \eta)}$, we have proven the claim.
\end{proof}

Next, let $\underline{u} = \hat u_j(0)$, where the choice of $j$ for the definition does not matter. Suppose the following condition holds for some positive $\eta$:
\begin{condition} \label{cond:u_lbar}
$\min_k|\theta_k - \underline{u}| > \eta$ and there is $\tilde \xi > 0$ such that $G^j(\tilde \xi) > 0$ for all $j$, and $\hat u_j(\tilde \xi) = \underline{u}$ for some $j$.
\end{condition}
This says that a sender with a positive amount $\tilde \xi$ of distribution $f_j$ gets the same payoff as the sender with no data, and that that payoff is bounded away from any $\theta_j$ by $\eta$. When showing that senders must play similarly under $\sigma^*_N$ in the limit as the average dataset becomes large, we consider separately the small fraction of senders that, by chance, receive very little data, i.e. those with $|t| \le \xi$, and in this case
\[
\lim_{\xi \rightarrow 0} \lim_{\omega \rightarrow 0} \min\{\mu: \exists j \text{ s.t. } \hat u_j(\mu) > \xi + \omega\} > 0,
\]
and there is a positive-measure set of types that may be pooled with those low-data senders.

Let us prove a similar claim to the previous one.
\begin{claim}\label{claim:lower_set}
    If there exists $\eta$ such that condition \ref{cond:u_lbar} holds, then for given $p >0$, there exists $\bar \epsilon > 0$ such that for $\epsilon < \bar \epsilon$, if we define
\[
S(\xi, \omega, \epsilon) = \{t \in \mathcal{T}_N \bigcap T(\epsilon): |t| \le \xi \text{ and } u_{\sigma^*_N}(t) \in (x,x+\Delta] \},
\]
then letting $a_N(\xi, \omega, \epsilon) = A_N(\underline{u} - \omega, 2 \omega; \epsilon)\setminus S(\xi,\omega,\epsilon)$ and $b_N(\xi, \omega, \epsilon) = B_N(\underline{u} - \omega, 2 \omega; \epsilon)\bigcup S(\xi,\omega,\epsilon)$, we have
\[
\lim_{\xi \rightarrow 0} \lim_{\omega \rightarrow 0} \lim_{N \rightarrow \infty} \frac{\sum_{t \in a_N(\xi, \omega, \epsilon)} \left[q_N(t) \sum_{t' \in b_N(\xi, \omega, \epsilon)} \sigma^*_N(t'|t)\right]}{\sum_{t \in a_N(\xi, \omega, \epsilon)} q_N(t)} \ge 1-p.
\]
\end{claim}

\begin{proof}[Proof of Claim \ref{claim:lower_set}]
    To start, note that in this case, we have for all $\xi > 0$ that
\[
\lim_{\epsilon \rightarrow 0} \lim_{N \rightarrow 0} \min_{\theta_j}\max_{t: |t| \le \xi} |\theta_j - \mathbb{E}_{\pi(\cdot|t)}[\theta]| = 0.
\]
Then for all $t \in a_N(\xi, \omega, \epsilon)$, there is some $\bar \epsilon'(\eta, \xi)$ and $\underline{N}'(\eta,\xi, \epsilon)$ so that for all $\epsilon < \bar \epsilon'(\eta, \xi)$ and $N > \underline{N}'(\eta,\xi,\epsilon)$, we have $u_{\sigma^*_N}(t) - \mathbb{E}_{\pi(\cdot|t)}[\theta] \ge \eta/3$.

We know that 
\begin{equation}
\begin{split}
\underline{u} - \omega &\le \lim_{N \rightarrow \infty}\frac{\sum_{t \in b_N(\xi, \omega, \epsilon)} q(t) \mathbb{E}_{\pi(\cdot|t)}[\theta] + \sum_{t \in a_N(\xi, \omega, \epsilon)} q(t) \mathbb{E}_{\pi(\cdot|t)}[\theta]}{\sum_{t \in b_N(\xi, \omega, \epsilon)} q(t) + \sum_{t \in a_N(\xi, \omega, \epsilon)} q(t)}\\
& = \frac{\sum_{\theta_j} \theta_j \beta_0(\theta_j)G^j(\hat \mu_j(\underline{u} + \omega))}{\sum_{\theta_j} \beta_0(\theta_j)G^j(\hat \mu_j(\underline{u} + \omega))}
\end{split}
\end{equation}
and
\begin{equation}
\begin{split}
\underline{u} + \omega & \ge \lim_{N \rightarrow \infty}\frac{\sum_{t \in b_N(\xi, \omega, \epsilon)} q(t) \mathbb{E}_{\pi(\cdot|t)}[\theta] + \sum_{t \in a_N(\xi, \omega, \epsilon)} \left[q(t) \mathbb{E}_{\pi(\cdot|t)}[\theta] \sum_{t' \in b_N(\xi, \omega, \epsilon)}\sigma_N^*(t'|t)\right]}{\sum_{t \in b_N(\xi, \omega, \epsilon)} q(t) + \sum_{t \in a_N(\xi, \omega, \epsilon)} \left[q(t) \sum_{t' \in b_N(\xi, \omega, \epsilon)}\sigma_N^*(t'|t)\right]}\\
 & = \lim_{N \rightarrow \infty} \frac{\sum_{\theta_j} \theta_j \beta_0(\theta_j)G^j(\hat \mu_j(\underline{u} + \omega)) -  \sum_{t \in a_N(\xi, \omega, \epsilon)}q(t) \mathbb{E}_{\pi(\cdot|t)}[\theta]\sum_{t' \in b_N(\xi, \omega, \epsilon)}(1-\sigma_N^*(t|t'))}{\sum_{\theta_j} \beta_0(\theta_j)G^j(\hat \mu_j(\underline{u} + \omega)) - \sum_{t \in a_N(\xi, \omega, \epsilon)}q(t) \sum_{t' \in b_N(\xi, \omega, \epsilon)}(1-\sigma_N^*(t|t'))}.
\end{split}
\end{equation}
Then, just as in eq. \ref{eq:delta_bound}, we have when $\epsilon < \bar \epsilon'(\eta, \xi, \omega)$ that
\[
\lim_{N \rightarrow \infty} 2 \omega \left(\sum_{t \in a_N(\xi, \omega, \epsilon) \bigcup b_N(\xi, \omega, \epsilon)} q(t)\right) - \frac{\eta}{3} \lim_{N \rightarrow \infty} \sum_{t \in a_N(x, \Delta; \epsilon)} q(t) \sum_{t' \in B_N(x, \Delta; \epsilon)}(1-\sigma^*_N(t|t')) \ge 0.
\]
Since there is some $j$ such that $\hat u_j(\tilde \xi)  = \underline{u}$, we have the bound
\[
\lim_{N \rightarrow \infty}\sum_{t \in a_N(\xi, \omega, \epsilon)} q_N(t) \ge \beta_0(\theta_j) [G^j(\tilde \xi) - G^j(\xi)].
\]
So, we have that 
\[
\lim_{N \rightarrow \infty} \frac{\sum_{t \in a_N(\xi,\omega, \epsilon)} \left[q_N(t) \sum_{t' \in b_N(\xi, \omega, \epsilon)} (1-\sigma^*_N(t'|t))\right]}{\sum_{t \in a_N(\xi,\omega, \epsilon)} q_N(t)} \le \frac{2 \omega (1 + \beta_0(\theta_j) [G^j(\tilde \xi) - G^j(\xi)])}{\beta_0(\theta_j) [G^j(\tilde \xi) - G^j(\xi)] \eta/3}.
\]
This implies the claim.
\end{proof}

Finally, we use these claims to prove the proposition. 

For any $\delta$ and $\rho$, there are $\xi^*$, $\epsilon^* > 0$ and $N^*$ such that for all $\xi < \xi^*$, $\epsilon < \epsilon^*,$  $N > N^*(\epsilon, \xi)$, and $\omega > 0$, any $t'$ is in either $B_N(x, \Delta; \epsilon)$ for some $\Delta$ and $x > \xi + \omega$, or in $b_N(\xi, \omega, \epsilon)$ if it is at most a distance $\delta$ away from some $t \in \mathcal{T}$ with $|u_{\sigma^*}(t) - u_{\sigma^*_N}(t')| \le \rho$.

In particular, find $l$ such that $\theta_l \le \bar u < \theta_{l+1}$, and then for any $K$ we can construct the collection of sets
\[
\left\{B_N\left(\xi + \omega + k\frac{\theta_l - (\xi + \omega)}{K}, \frac{\theta_l - (\xi + \omega)}{K}; \epsilon \right)\right\}_{k = 0}^{K-1}
\]
and
\[
\left\{B_N\left(\theta_{j-1} + \eta + k\frac{\theta_{j} - \eta - (\theta_{j-1} + \eta)}{K}, \frac{\theta_{j} - \eta - (\theta_{j-1} + \eta)}{K}; \epsilon \right)\right\}_{k=0}^K, \ \ \  \text{for all } j > l,
\]
essentially partitioning the imitated senders by the payoffs they receive, into intervals that are disjoint, cover all attained payoffs except $[0, \xi + \omega]$ and the intervals $[\theta_j - \eta, \theta_j + \eta)$ and are arbitrarily small as $K \rightarrow \infty$.

Let $C(N, K, \xi, \omega, \eta; \epsilon)$ be the collection that is the union of these collections, and also includes, if condition \ref{cond:u_lbar} holds for $\eta$, the set $b_N(\xi, \omega, \epsilon)$. Call the elements of $C(N, K, \xi, \omega, \eta; \epsilon)$ by $C_1(N, K, \xi, \omega, \eta; \epsilon), \ldots, C_I(N, K, \xi, \omega, \eta; \epsilon)$.

Likewise, we can construct the collection of sets
\[
\left\{A_N\left(\xi + \omega + k\frac{\theta_l - (\xi + \omega)}{K}, \frac{\theta_l - (\xi + \omega)}{K}; \epsilon \right)\right\}_{k = 0}^{K-1}
\]
and
\[
\left\{A_N\left(\theta_{j-1} + \eta + k\frac{\theta_{j} - \eta - (\theta_{j-1} + \eta)}{K}, \frac{\theta_{j} - \eta - (\theta_{j-1} + \eta)}{K}; \epsilon \right)\right\}_{k=0}^K, \ \ \  \text{for all } j > l,
\]
which are corresponding sets of imitating types; let $D(N, K, \xi, \omega, \eta; \epsilon)$ be the collection containing these as well as $a_N(\xi, \omega, \epsilon)$ if condition \ref{cond:u_lbar} holds for $\eta$. Call the elements of $D(N, K, \xi, \omega, \eta; \epsilon)$ by $D_1(N, K, \xi, \omega, \eta; \epsilon), \ldots, D_I(N, K, \xi, \omega, \eta; \epsilon)$.

The proposition follows from proving that the ex-ante probability that the sender imitates some $t'$ that is in an element of $C(N, K, \xi, \omega, \eta; \epsilon)$ converges to $1$ in the large $N$ limit and in the limit as $K \rightarrow \infty$ and $\epsilon, \omega, \xi, \eta \rightarrow 0$.

To see this, first observe that, from the above two claims, if we define
\[
LB(N, K, \xi, \omega, \eta; \epsilon) = \min_{i} Pr(m_{\sigma^*_N}(t) \in C_i(N, K, \xi, \omega, \eta; \epsilon)| t \in D_i(N, K, \xi, \omega, \eta; \epsilon)),
\]
then $\lim_{\xi \rightarrow 0} \lim_{\omega \rightarrow 0}\lim_{K \rightarrow \infty}\lim_{\epsilon \rightarrow 0}\lim_{N \rightarrow \infty} UB(N, K, \xi, \omega, \eta; \epsilon) = 1$; note that this is a uniform bound over all $i$.

Then, letting $T$ denote a set of types that is an element of $C(N, K, \xi, \omega, \eta; \epsilon)$, we have
\begin{equation}
    \begin{split}
        &Pr\left(m_{\sigma^*_N}(t) \in [\bigcup_{C_N(K, \xi, \omega, \eta; \epsilon)} T] \bigg| \min_k |u_{\sigma^*_N}(t) - \theta_k| > \eta \right) \\
        \ge &\sum_{i = 1}^I \Bigg[ Pr\left(t \in D_i(N, K, \xi, \omega, \eta; \epsilon)| \min_k |u_{\sigma^*_N}(t) - \theta_k| > \eta \right)\\
        & \ \ \ \ \ \ \ \ \ \cdot Pr(m_{\sigma^*_N}(t) \in C_i(N, K, \xi, \omega, \eta; \epsilon)| t \in D_i(N, K, \xi, \omega, \eta; \epsilon))\\
        & \ \ + Pr(m_{\sigma^*_N}(t) \in C_i(N, K, \xi, \omega, \eta; \epsilon)\bigg| \min_k |u_{\sigma^*_N}(t) - \theta_k| > \eta)\Bigg]\\
        & \ge \frac{Pr(t \in \bigcup_i D_i(N, K, \xi, \omega, \eta; \epsilon)) LB(N, K, \xi, \omega, \eta; \epsilon) + Pr(t \in \bigcup_i C_i(N, K, \xi, \omega, \eta; \epsilon))}{Pr(\min_k |u_{\sigma^*_N}(t) - \theta_k| > \eta)}.
    \end{split}
\end{equation}
Since $\lim_{\xi \rightarrow 0} \lim_{\omega \rightarrow 0}\lim_{K \rightarrow \infty}\lim_{\epsilon \rightarrow 0}\lim_{N \rightarrow \infty} PR(t \in \bigcup_i [D_i(N, K, \xi, \omega, \eta; \epsilon) \bigcup D_i(N, K, \xi, \omega, \eta; \epsilon)]) = Pr(\min_k |u_{\sigma^*_N}(t) - \theta_k > \eta)$, we have
\begin{equation}
    \begin{split}
&\lim_{\xi \rightarrow 0} \lim_{\omega \rightarrow 0}\lim_{K \rightarrow \infty}\lim_{\epsilon \rightarrow 0}\lim_{N \rightarrow \infty}\\
& \ \ \ \frac{Pr(t \in \bigcup_i D_i(N, K, \xi, \omega, \eta; \epsilon)) LB(N, K, \xi, \omega, \eta; \epsilon) + Pr(t \in \bigcup_i C_i(N, K, \xi, \omega, \eta; \epsilon))}{Pr(\min_k |u_{\sigma^*_N}(t) - \theta_k| > \eta)}\\
& = 1
\end{split}
\end{equation}
for all $\eta$, thus proving the proposition.

\end{proof}

\end{document}